%% file: main.tex
\documentclass[11pt]{article}

\usepackage{fullpage}

\input{command-template}

\title{Edge-Coloring Algorithms for Bounded Degree Multigraphs}

\usepackage{authblk}

\author[1]{Abhishek Dhawan\thanks{Email: \href{mailto:abhishek.dhawan@math.gatech.edu}{abhishek.dhawan@math.gatech.edu}}}

\affil[1]{School of Mathematics, Georgia Institute of Technology}

\date{\today}

\begin{document}

\maketitle

\begin{abstract}
    In this paper, we consider algorithms for edge-coloring multigraphs $G$ of bounded maximum degree, i.e., $\Delta(G) = O(1)$. 
    Shannon's theorem states that any multigraph of maximum degree $\Delta$ can be properly edge-colored with $\Sha$ colors.
    Our main results include algorithms for computing such colorings.
    We design deterministic and randomized sequential algorithms with running time $O(n\log n)$ and $O(n)$, respectively. 
    This is the first improvement since the $O(n^2)$ algorithm in Shannon's original paper, and our randomized algorithm is optimal up to constant factors.
    We also develop distributed algorithms in the \LOCAL model of computation.
    Namely, we design deterministic and randomized \LOCAL algorithms with running time $\tilde O(\log^5 n)$ and $O(\log^2n)$, respectively.
    The deterministic sequential algorithm is a simplified extension of earlier work of Gabow et al. in edge-coloring simple graphs.
    The other algorithms apply the entropy compression method in a similar way to recent work by the author and Bernshteyn, where the authors design algorithms for Vizing's theorem for simple graphs.
    We also extend those results to Vizing's theorem for multigraphs.
\end{abstract}

\tableofcontents

\input{intro}

\input{prelim}

\input{shannon_deterministic}
\input{MSSA}

\input{shannon}

\section*{Acknowledgements}
I would like to thank my advisor, Anton Bernshteyn, for his helps and encouragements, without which this research would have been impossible. 
Also, I would like to thank the anonymous referees for their helpful comments.

\bibliographystyle{alpha}
\bibliography{references}

\appendix
\input{vizing}

\end{document}

%% file: command-template.tex
\usepackage{comment}
\usepackage{epsf}
\usepackage{subcaption}
\usepackage{graphicx}
\usepackage{bbm}
\usepackage{color}
\usepackage{nicefrac}
\usepackage{mathtools}
\usepackage{amsfonts}
\usepackage{amsthm}
\usepackage{amsmath, bm}
\usepackage{amssymb}
\usepackage{textcomp}
\usepackage{lipsum}
\usepackage{xcolor}
\usepackage{graphicx}
\usepackage{wrapfig}
\usepackage{mathtools}
\usepackage{enumitem}
\usepackage{xspace}

\usepackage{algorithm}
\usepackage{algpseudocode}
\counterwithin{algorithm}{section}

\usepackage{tikz}
\usetikzlibrary{shapes,snakes}
\usetikzlibrary{arrows.meta}
\usetikzlibrary{decorations.pathmorphing}
\usetikzlibrary{patterns}
\usetikzlibrary{calc}

\usepackage{amsbsy}

\usepackage{hyperref}
\definecolor{OliveGreen}{HTML}{3C8031}
\hypersetup{
colorlinks=true,
urlcolor=blue,
linkcolor=blue,
citecolor=OliveGreen,
}

\newenvironment{fminipage}%
  {\end{minipage}\end{Sbox}\fbox{\TheSbox}}

\newtheorem{theorem}{Theorem}[section]

\newtheorem{lemma}[theorem]{Lemma}
\newtheorem{corollary}[theorem]{Corollary}
\newtheorem{definition}[theorem]{Definition}

\newtheorem{claim}{Claim}[theorem]
\newtheorem{fact}[theorem]{Fact}

\newenvironment{claimproof}{\noindent$\rhd$\hspace{1em}}{\hfill$\lhd$}

\newlength{\widebarargwidth}
\newlength{\widebarargheight}
\newlength{\widebarargdepth}

\makeatletter
\long\def\@makecaption#1#2{
       \vskip 0.8ex
       \setbox\@tempboxa\hbox{\small {\bf #1:} #2}
       \parindent 1.5em 
       \dimen0=\hsize
       \advance\dimen0 by -3em
       \ifdim \wd\@tempboxa >\dimen0
               \hbox to \hsize{
                       \parindent 0em
                       \hfil 
                       \parbox{\dimen0}{\def\baselinestretch{0.96}\small
                               {\bf #1.} #2
                               } 
                       \hfil}
       \else \hbox to \hsize{\hfil \box\@tempboxa \hfil}
       \fi
       }
\makeatother

\long\def\comment#1{}

\newcommand{\EE}{\ensuremath{{\mathbb{E}}}}

\newcommand{\bbone}[1]{\mathbbm{1}\{#1\}}
\newcommand{\set}[1]{\{#1\}}
\newcommand{\defeq}{\coloneqq}
\newcommand{\N}{\mathbb{N}}
\newcommand{\Z}{\mathbb{Z}}

\newcommand{\E}{\mathbb{E}}

\newcommand{\IE}{E_{\mathsf{int}}}
\newcommand{\IV}{V_{\mathsf{int}}}
\newcommand{\Sha}{\lfloor 3\Delta/2\rfloor}
\newcommand{\aug}{\mathsf{Aug}}
\newcommand{\LOCAL}{$\mathsf{LOCAL}$\xspace}
\newcommand{\poly}{\mathsf{poly}}
\newcommand{\Shift}{\mathsf{Shift}}
\newcommand{\vend}{\mathsf{vEnd}}
\newcommand{\vstart}{\mathsf{vStart}}
\newcommand{\Start}{\mathsf{Start}}
\newcommand{\End}{\mathsf{End}}
\newcommand{\Pivot}{\mathsf{Pivot}}
\newcommand{\length}{\mathsf{length}}
\newcommand{\visited}{\mathsf{visited}}
\newcommand{\dom}{\mathsf{dom}}
\newcommand{\0}{\varnothing}
\newcommand{\bemph}[1]{{\normalfont#1}} 
\newcommand{\ep}[1]{\bemph{(}#1\bemph{)}} 

\newcommand{\emphdef}[1]{\textbf{\textit{{#1}}}}
\newcommand{\emphd}[1]{\emphdef{#1}}
\newcommand{\blank}{\mathsf{blank}}

\newcommand{\val}{\mathsf{val}}
\newcommand{\wt}{\mathsf{wt}}

%% file: intro.tex
\section{Introduction}\label{section:intro}

All multigraphs considered in this paper are finite, undirected and loopless.
Edge-coloring is a central topic in graph theory with numerous applications (in scheduling problems, designing communication networks and round robin tournaments, to name a few).
In this paper, we are interested in algorithms for edge-coloring multigraphs with few colors.
To begin with, we will make some definitions.
For $r \in \N$, we let $[r] \defeq \set{1, \ldots, r}$.
Let $G$ be a multigraph with $n$ vertices and $m$ edges.
For $x,y \in V(G)$ and $e \in E(G)$, we define $E_G(x)$ to be the set of edges incident to $x$, $E_G(x,y)$ to be the set of edges between $x$ and $y$, $N_G(x)$ to be the set of vertices $z$ such that $E_G(x,z) \neq \0$, and $V(e)$ to be the set of endpoints of $e$.
Furthermore, we define the degree of $x$ to be the number of edges incident to $x$ and the multiplicity of $\set{x, y}$ to be the number of edges between $x$ and $y$, i.e., 
\[\deg_G(x) \defeq |E_G(x)|, \quad \mu(x,y) \defeq |E_G(x, y)|.\] 
The maximum degree and maximum multiplicity of $G$ (denoted $\Delta(G)$ and $\mu(G)$, respectively) are defined to be the maximum values attained above, respectively.

\begin{definition}\label{defn:edge_coloring}
    Given a multigraph $G$, we say that two edges $e, f \in E(G)$ are \emphd{adjacent} if $V(e) \cap V(f) \neq \0$.
    A proper $r$-edge-coloring of $G$ is a function $\phi\,:\,E(G) \to [r]$ such that $\phi(e) \neq \phi(f)$ for any adjacent edges $e, f$.
    The chromatic index of $G$, denoted by $\chi'(G)$, is the minimum $r$ such that $G$ admits a proper $r$-edge-coloring.
\end{definition}

It is easy to see that $\chi'(G) \geq \Delta(G)$.
Shannon and Vizing famously proved upper bounds on $\chi'(G)$ in terms of the maximum degree and maximum multiplicity of the multigraph.
We note that these bounds are optimal as a result of the construction in \cite{vizing1965chromatic}.

\begin{theorem}[Shannon's Theorem \cite{Shannon}]\label{theo:shannon_bound}
    If $G$ is a multigraph of maximum degree $\Delta$, then $\chi'(G) \leq \Sha$.
\end{theorem}

\begin{theorem}[Vizing's Theorem \cite{Vizing}]\label{theo:vizing_bound}
    If $G$ is a multigraph of maximum degree $\Delta$ and maximum multiplicity $\mu$, then $\chi'(G) \leq \Delta + \mu$.
\end{theorem}

It is \textsf{NP}-hard to determine $\chi'(G)$, even when $\Delta(G) = 3$ \cite{Holyer}.
Therefore, it is natural to design algorithms for edge-coloring multigraphs with $\Sha$ or $\Delta + \mu$ colors.
For simple graphs ($\mu = 1$), there has been much work on algorithms for Vizing's bound.
The original proof itself is algorithmic, simplifications of which appeared in \cite{Bollobas, RD, MG}. 
These algorithms run in $O(mn)$ time as they iteratively color a single edge by modifying the colors on $O(n)$ edges.
Gabow, Nishizeki, Kariv, Leven and Terada \cite{GNKLT} defined two recursive algorithms, which run in $O(m\sqrt{n\log n})$ and $O(m\log n)$ time.
There was no improvement until a recent paper of Sinnamon \cite{Sinnamon}, where the author defined a $O(m\sqrt{n})$ time algorithm.
As noted by the authors in \cite{EdgeColoringMonograph}, the problem is considerably harder when $\mu > 1$, which is the main focus of this paper.

We will primarily be considering algorithms for \hyperref[theo:shannon_bound]{Shannon's Theorem} when $\Delta(G) = O(1)$.
The original proof of the theorem yields a $O(m(\Delta + n)) = O(n^2)$ algorithm for $\Sha$-edge-coloring (in the $\Delta = O(1)$ regime, as $m \leq n\Delta$).
To the best of our knowledge, this is the first paper to consider improving the above runtime.
However, there has been progress in edge-coloring multigraphs of special structure, specifically bipartite multigraphs.
Gabow and Kariv exhibited a $\Delta$-edge-coloring algorithm which runs in $O(\min\set{m\log^2 n, n^2\log n})$ time \cite{gabow1982algorithms}.
Schrijver improved upon this by designing a $O(\Delta\,m)$ time algorithm \cite{schrijver1998bipartite}, and Cole, Ost and Schirra further improved to develop one that runs in $O(m\log\Delta)$ time \cite{cole2001edge}.
Furthermore, most of the algorithms mentioned in the previous paragraph can be implemented for \hyperref[theo:vizing_bound]{Vizing's Theorem} for multigraphs with $\mu > 1$, with minor modifications to certain subprocedures.

The first result of this paper defines an efficient deterministic algorithm for $\Sha$-edge-coloring.
The algorithm is inspired by the $O(\Delta\,m\log n)$ algorithm of \cite{GNKLT}.
As we are interested in the regime where $\Delta = O(1)$, we are able to avoid the recursive calls in their algorithm.
This is the key difference in our approach.
As we can now define an iterative algorithm, the analysis is much easier to follow (although more technical due to the additional structure of multigraphs).
Furthermore, while our focus in this paper is on time complexity, we remark that our algorithm is more space efficient as well.
Due to the recursive step, the algorithm in \cite{GNKLT} requires defining $\Theta(\Delta)$ new graphs.
By circumventing recursion, we do not need to define \textbf{any} new graphs during our algorithm.

\begin{theorem}[Deterministic algorithm for Shannon's theorem]\label{theo:shannon_deterministic}
    Let $G$ be an $n$-vertex loopless multigraph of maximum degree $\Delta = O(1)$.
    There is a deterministic sequential algorithm that finds a proper $\Sha$-edge-coloring of $G$ in time $O(n\log n)$.
\end{theorem}

The remaining results of this paper concern an extension of recent work of the author and Bernshteyn in algorithms for Vizing's theorem for simple graphs. 
The authors consider randomized sequential algorithms as well as distributed algorithms in the \LOCAL model of computation introduced by Linial in \cite{Linial}.
We refer the reader to the book \cite{BE} for a thorough introduction to this subject.
In the \LOCAL model, a communication network is modeled by a graph $G$. 
Computation proceeds in \emphd{rounds} where in each round the vertices of $G$ perform some local computations and then broadcast their results to their neighbors.
The broadcast step happens synchronously across all vertices.
There are no restrictions on the complexity of the local computations or on the length of the messages being broadcast.
Efficiency of distributed \LOCAL algorithms are measured in worst case number of rounds required to produce the desired output.

\begin{table}[b]
    \centering
    \begin{tabular}{|c|c|c|}
        \hline
        Reference(s) & Number of Colors & Complexity \\
         \hline
        \cite{GPSh} & $2\Delta - 1$ & $O(\Delta^2) + \log^*n$ \\
        \cite{ABI, Luby}$^\star$ & $2\Delta - 1$ & $O(\log n)$ \\
        \cite{GKMU} & $\Sha$ & $\poly(\Delta)\log^8 n$ \\
        \cite{CHLPU}$^\star$ & $\Delta + \sqrt{\Delta}\poly(\log \Delta)$ & $\poly(\Delta, \log\log n)$ \\
        \cite{SV}$^\star$ & $\Delta + 2$ & $\poly(\Delta)\log^3n$ \\
        \cite{VizingChain}$^\star$ & $\Delta + 1$ & $\poly(\Delta)\log^5n$ \\
        \cite{VizingChain} & $\Delta + 1$ & $\poly(\Delta, \log \log n)\log^{11}n$ \\
        \cite{Christ} & $\Delta + 1$ & $\poly(\Delta, \log \log n)\log^6n$ \\
        \cite{fastEdgeColoring}$^\star$ & $\Delta + 1$ & $\poly(\Delta)\log^2n$ \\
        \cite{fastEdgeColoring} & $\Delta + 1$ & $\poly(\Delta, \log \log n)\log^5n$ \\
         \hline
    \end{tabular}
    \caption{Results on distributed algorithms ($\star$ indicates randomized).}
    \label{table:dist}
\end{table}

As the vertices are synchronously running the same algorithm, we need a way to distinguish them from one another.
There are two standard approaches to do so.
\begin{itemize}
    \item In the \emphd{deterministic} version of the \LOCAL model, each vertex is identified by a unique $\Theta(\log n)$-bit ID, where $n \defeq |V(G)|$.
    The algorithm must output a correct solution, regardless of the identifiers.

    \item In the \emphd{randomized} version of the \LOCAL model, each vertex independently generates an arbitrarily long sequence of bits uniformly at random.
    The algorithm must output a correct solution with probability at least $1 - 1/\poly(n)$.
    
\end{itemize}
The study of \LOCAL algorithms for edge-coloring simple graphs has a long history. 
We mention a few results in Table~\ref{table:dist} to provide context for our work and direct the reader to the surveys \cite{CHLPU, GKMU} for a more thorough history.

The $\Sha$-edge-coloring algorithm in \cite{GKMU} can be implemented for multigraphs as well (although the authors don't explicitly state so).
Apart from this, there are no distributed algorithms designed specifically for $\Sha$-edge-colorings of multigraphs to the best of our knowledge. 
However, some of the results in Table~\ref{table:dist} can be extended to the multigraph case with minor modifications to the algorithms (and an additional $\sim\mu$ colors).

We are now ready to state our main results, which extend those of \cite[Theorems 1.6, 1.8]{fastEdgeColoring} to $\Sha$-edge-colorings of multigraphs.

\begin{theorem}[Algorithms for Shannon's theorem]\label{theo:shannon}
    Let $G$ be an $n$-vertex loopless multigraph of maximum degree $\Delta = O(1)$.

    \begin{enumerate}
        \item\label{item:seq_sha} There is a randomized sequential algorithm that finds a proper $\Sha$-edge-coloring of $G$ in time $O(n)$ with probability at least $1 - 1/\Delta^n$.
        
        \item\label{item:dist_rand_sha} There is a randomized \LOCAL algorithm that finds a proper $\Sha$-edge-coloring of $G$ in $O(\log^2 n)$ rounds.
        
        \item\label{item:dist_det_sha} There is a deterministic \LOCAL algorithm that finds a proper $\Sha$-edge-coloring of $G$ in $\tilde O(\log^5 n)$ rounds.

    \end{enumerate}
\end{theorem}

The $\tilde O(\cdot)$ in \ref{item:dist_det_sha} hides $\poly(\log \log n)$ factors.
We also remark that \cite[Theorems 1.6, 1.8]{fastEdgeColoring} can be extended to Vizing's theorem for multigraphs.
The extension requires changes to certain subprocedures, details of which are provided in \S\ref{section:vizing} where we prove the following theorem.

\begin{theorem}[Algorithms for Vizing's theorem]\label{theo:vizing}
    Let $G$ be an $n$-vertex loopless multigraph of maximum degree $\Delta = O(1)$ and maximum multiplicity $\mu$.

    \begin{enumerate}
        \item\label{item:seq_viz} There is a randomized sequential algorithm that finds a proper $(\Delta + \mu)$-edge-coloring of $G$ in time $O(n)$ with probability at least $1 - 1/\Delta^n$.
        
        \item\label{item:dist_rand_viz} There is a randomized \LOCAL algorithm that finds a proper $(\Delta + \mu)$-edge-coloring of $G$ in $O(\log^2 n)$ rounds.
        
        \item\label{item:dist_det_viz} There is a deterministic \LOCAL algorithm that finds a proper $(\Delta + \mu)$-edge-coloring of $G$ in $\tilde O(\log^5 n)$ rounds.
        
    \end{enumerate}
\end{theorem}

We remark that the constant factors in the $O(\cdot)$ for our theorems contain $\poly(\Delta)$ factors (the explicit exponents are provided in our proofs).
It follows that our results on sequential algorithms are an improvement in the $\Delta = n^{o(1)}$ regime, and our distributed algorithms improve in the $\Delta = (\log n)^{o(1)}$ regime as well.
Since the proofs of Theorems~\ref{theo:shannon_deterministic} and \ref{theo:shannon} are rather technical, we will provide an overview in the next subsection.
In \S\ref{sec:notation}, we will introduce some notation and background facts that will be used in our proofs.
In \S\ref{section:shannon_det}, we will prove Theorem~\ref{theo:shannon_deterministic}.
In \S\ref{section:MSSA}, we will describe a key subprocedure used in the algorithms in Theorem~\ref{theo:shannon}, which we will prove in \S\ref{section:shannon_proof}.

\subsection{Proof Overview}\label{subsection:overview}

Throughout the remainder of the paper, we fix a multigraph $G$ of maximum degree $\Delta$ and maximum multiplicity $\mu$.
We call a function $\phi \,:\, E(G) \to [r] \cup \set{\blank}$ a partial $r$-edge-coloring, where $\phi(e) = \blank$ indicates an uncolored edge.
The \emphd{domain} of $\phi$ is the set of edges colored under $\phi$, i.e., $\dom(\phi) \defeq \set{e\,:\, \phi(e) \neq \blank}$.
A standard technique in algorithms for edge-colorings is to find so-called \emphd{augmenting subgraphs} with respect to a partial coloring.

\begin{definition}[Augmenting subgraphs for partial colorings]\label{defn:aug}
    Given a partial $r$-edge-coloring $\phi$, a subgraph $H \subseteq G$ is called $S$-augmenting for a set of uncolored edges $S$ if there is a proper partial coloring $\phi'$ such that $\dom(\phi') = \dom(\phi) \cup S$, and $\phi'(f) = \phi(f)$ for all $f \notin E(H)$ (if $S = \set{e}$, we say $H$ is $e$-augmenting).
    In particular, we can modify the colors of the edges in $H$ to define a new coloring under which every edge in $S$ is now colored.
    We say $\phi'$ is constructed from $\phi$ by \emphd{augmenting} $H$.
\end{definition}

This yields an algorithmic framework for sequential edge-coloring, i.e., iteratively choose an edge $e$ (deterministically or randomly) and color it by finding an $e$-augmenting subgraph.
The original proof of \hyperref[theo:shannon_bound]{Shannon's Theorem} describes a procedure to find an $e$-augmenting subgraph that we call a \emphd{Shannon chain} (see Definition~\ref{def:shannon_chain} and Fig.~\ref{subfig:shannon_chain}).
The colors $\alpha$ and $\beta$ in Fig.~\ref{subfig:shannon_chain} satisfy certain conditions described in Lemma~\ref{lemma:first_fan_shannon}.
The need for $\Sha$ colors is essential in defining $\beta$ (see Fact~\ref{fact:sha}).
This procedure yields a $O(n^2)$ algorithm as it takes $O(\length(C))$ time to augment a chain $C$, and in the worst case a Shannon chain has length $\Theta(n)$.

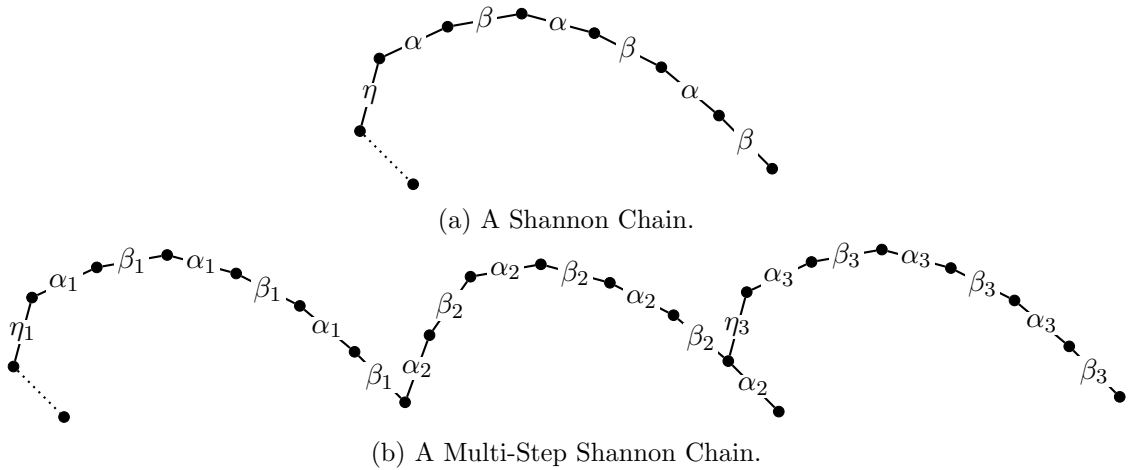
\begin{figure}[h]
    \begin{subfigure}[t]{\textwidth}
        \centering
    	\begin{tikzpicture}
    	    \node[circle,fill=black,draw,inner sep=0pt,minimum size=4pt] (a) at (0,0) {};
        	\path (a) ++(-45:1) node[circle,fill=black,draw,inner sep=0pt,minimum size=4pt] (b) {};
        	\path (a) ++(75:1) node[circle,fill=black,draw,inner sep=0pt,minimum size=4pt] (c) {};

                \path (c) ++(25:1) node[circle,fill=black,draw,inner sep=0pt,minimum size=4pt] (d) {};
                \path (d) ++(10:1) node[circle,fill=black,draw,inner sep=0pt,minimum size=4pt] (e) {};
                \path (e) ++(-15:1) node[circle,fill=black,draw,inner sep=0pt,minimum size=4pt] (f) {};
                \path (f) ++(-27:1) node[circle,fill=black,draw,inner sep=0pt,minimum size=4pt] (g) {};
                \path (g) ++(-40:1) node[circle,fill=black,draw,inner sep=0pt,minimum size=4pt] (h) {};
                \path (h) ++(-45:1) node[circle,fill=black,draw,inner sep=0pt,minimum size=4pt] (i) {};

        	\draw[thick,dotted] (a) -- (b);
        	\draw[thick] (a) to node[midway,inner sep=1pt,outer sep=1pt,minimum size=4pt,fill=white] {$\eta$} (c) to node[midway,inner sep=1pt,outer sep=1pt,minimum size=4pt,fill=white] {$\alpha$} (d) to node[midway,inner sep=1pt,outer sep=1pt,minimum size=4pt,fill=white] {$\beta$} (e) to node[midway,inner sep=1pt,outer sep=1pt,minimum size=4pt,fill=white] {$\alpha$} (f) to node[midway,inner sep=1pt,outer sep=1pt,minimum size=4pt,fill=white] {$\beta$} (g) to node[midway,inner sep=1pt,outer sep=1pt,minimum size=4pt,fill=white] {$\alpha$} (h) to node[midway,inner sep=1pt,outer sep=1pt,minimum size=4pt,fill=white] {$\beta$} (i);
    	\end{tikzpicture}
    	\caption{A Shannon Chain.}
            \label{subfig:shannon_chain}
    \end{subfigure}
    \begin{subfigure}[t]{\textwidth}
        \centering
    	\begin{tikzpicture}[xscale=0.95,yscale=0.95]
                \node[circle,fill=black,draw,inner sep=0pt,minimum size=4pt] (a) at (0,0) {};
        	\path (a) ++(-45:1) node[circle,fill=black,draw,inner sep=0pt,minimum size=4pt] (b) {};
        	\path (a) ++(75:1) node[circle,fill=black,draw,inner sep=0pt,minimum size=4pt] (c) {};

                \path (c) ++(25:1) node[circle,fill=black,draw,inner sep=0pt,minimum size=4pt] (d) {};
                \path (d) ++(10:1) node[circle,fill=black,draw,inner sep=0pt,minimum size=4pt] (e) {};
                \path (e) ++(-15:1) node[circle,fill=black,draw,inner sep=0pt,minimum size=4pt] (f) {};
                \path (f) ++(-27:1) node[circle,fill=black,draw,inner sep=0pt,minimum size=4pt] (g) {};
                \path (g) ++(-40:1) node[circle,fill=black,draw,inner sep=0pt,minimum size=4pt] (h) {};
                \path (h) ++(-45:1) node[circle,fill=black,draw,inner sep=0pt,minimum size=4pt] (i) {};

        	\draw[thick,dotted] (a) -- (b);
        	\draw[thick] (a) to node[midway,inner sep=1pt,outer sep=1pt,minimum size=4pt,fill=white] {$\eta_1$} (c) to node[midway,inner sep=1pt,outer sep=1pt,minimum size=4pt,fill=white] {$\alpha_1$} (d) to node[midway,inner sep=1pt,outer sep=1pt,minimum size=4pt,fill=white] {$\beta_1$} (e) to node[midway,inner sep=1pt,outer sep=1pt,minimum size=4pt,fill=white] {$\alpha_1$} (f) to node[midway,inner sep=1pt,outer sep=1pt,minimum size=4pt,fill=white] {$\beta_1$} (g) to node[midway,inner sep=1pt,outer sep=1pt,minimum size=4pt,fill=white] {$\alpha_1$} (h) to node[midway,inner sep=1pt,outer sep=1pt,minimum size=4pt,fill=white] {$\beta_1$} (i);

                \path (i) ++(70:1) node[circle,fill=black,draw,inner sep=0pt,minimum size=4pt] (j) {};
                \path (j) ++(55:1) node[circle,fill=black,draw,inner sep=0pt,minimum size=4pt] (k) {};
                \path (k) ++(10:1) node[circle,fill=black,draw,inner sep=0pt,minimum size=4pt] (l) {};
                \path (l) ++(-15:1) node[circle,fill=black,draw,inner sep=0pt,minimum size=4pt] (m) {};
                \path (m) ++(-27:1) node[circle,fill=black,draw,inner sep=0pt,minimum size=4pt] (n) {};
                \path (n) ++(-40:1) node[circle,fill=black,draw,inner sep=0pt,minimum size=4pt] (o) {};
                \path (o) ++(-45:1) node[circle,fill=black,draw,inner sep=0pt,minimum size=4pt] (p) {};

                \draw[thick] (i) to node[midway,inner sep=1pt,outer sep=1pt,minimum size=4pt,fill=white] {$\alpha_2$} (j) to node[midway,inner sep=1pt,outer sep=1pt,minimum size=4pt,fill=white] {$\beta_2$} (k) to node[midway,inner sep=1pt,outer sep=1pt,minimum size=4pt,fill=white] {$\alpha_2$} (l) to node[midway,inner sep=1pt,outer sep=1pt,minimum size=4pt,fill=white] {$\beta_2$} (m) to node[midway,inner sep=1pt,outer sep=1pt,minimum size=4pt,fill=white] {$\alpha_2$} (n) to node[midway,inner sep=1pt,outer sep=1pt,minimum size=4pt,fill=white] {$\beta_2$} (o) to node[midway,inner sep=1pt,outer sep=1pt,minimum size=4pt,fill=white] {$\alpha_2$} (p);

                \path (o) ++(75:1) node[circle,fill=black,draw,inner sep=0pt,minimum size=4pt] (r) {};

                \path (r) ++(25:1) node[circle,fill=black,draw,inner sep=0pt,minimum size=4pt] (s) {};
                \path (s) ++(10:1) node[circle,fill=black,draw,inner sep=0pt,minimum size=4pt] (t) {};
                \path (t) ++(-15:1) node[circle,fill=black,draw,inner sep=0pt,minimum size=4pt] (u) {};
                \path (u) ++(-27:1) node[circle,fill=black,draw,inner sep=0pt,minimum size=4pt] (v) {};
                \path (v) ++(-40:1) node[circle,fill=black,draw,inner sep=0pt,minimum size=4pt] (w) {};
                \path (w) ++(-45:1) node[circle,fill=black,draw,inner sep=0pt,minimum size=4pt] (x) {};

                \draw[thick] (o) to node[midway,inner sep=1pt,outer sep=1pt,minimum size=4pt,fill=white] {$\eta_3$} (r) to node[midway,inner sep=1pt,outer sep=1pt,minimum size=4pt,fill=white] {$\alpha_3$} (s) to node[midway,inner sep=1pt,outer sep=1pt,minimum size=4pt,fill=white] {$\beta_3$} (t) to node[midway,inner sep=1pt,outer sep=1pt,minimum size=4pt,fill=white] {$\alpha_3$} (u) to node[midway,inner sep=1pt,outer sep=1pt,minimum size=4pt,fill=white] {$\beta_3$} (v) to node[midway,inner sep=1pt,outer sep=1pt,minimum size=4pt,fill=white] {$\alpha_3$} (w) to node[midway,inner sep=1pt,outer sep=1pt,minimum size=4pt,fill=white] {$\beta_3$} (x);
    	\end{tikzpicture}
    	\caption{A Multi-Step Shannon Chain.}
     \label{subfig:multi_shannon_chain}
    \end{subfigure}
    \caption{Augmenting subgraphs (Greek letters represent colors).}
    \label{fig:aug_sub}
\end{figure}

\paragraph{Disjoint Shannon chains.}
The algorithm for Theorem~\ref{theo:shannon_deterministic} uses the following observation made by \cite{GNKLT} in algorithms for Vizing's theorem for simple graphs: while the worst case length of a Vizing chain is $\Theta(n)$, we can find a large subset of \emphd{disjoint Vizing chains} having a shorter length. 
We extend this idea to Shannon chains.
The precise definition of disjoint is provided in Definition~\ref{defn:disjoint}.
For now, we may assume the chains are edge disjoint from where it can be shown that we can augment a set of disjoint chains in $O(m)$ time.
In \S\ref{section:shannon_det}, we will describe an algorithm that, given a partial $\Sha$-edge-coloring $\phi$, finds and augments a set of disjoint Shannon chains for a large subset $S$ of the uncolored edges.
Following the terminology of Definition~\ref{defn:aug}, the subgraph $H$ induced by the union of these chains is $S$-augmenting.
We show that we can always find $S$ such that $|S|$ is at least a $\poly(\Delta)$ fraction of the uncolored edges.
It follows that the graph can be colored in $\poly(\Delta)\,\log n$ applications of this procedure, which would complete the proof of Theorem~\ref{theo:shannon_deterministic}.

\paragraph{Multi-Step Shannon Chains.}
To improve upon the $\poly(\Delta)\,n\log n$ bound, we need to construct smaller augmenting subgraphs.
In \cite{CHLPU}, the authors show that there may not exist an augmenting subgraph having fewer than $\poly(\Delta) \log n$ edges (which is still much lower than the $O(n)$ bound).
In \cite{Christ}, Christiansen proved the existence of augmenting subgraphs for Vizing's theorem for simple graphs containing at most $\poly(\Delta) \log n$ edges and in \cite{fastEdgeColoring}, the author and Bernshteyn provided an explicit construction for such augmenting subgraphs called multi-step Vizing chains.
One of the main contributions of this paper is the introduction of the analogously defined
\emphd{multi-step Shannon chains} (see Definition~\ref{def:multi_shannon_chain} and Fig.~\ref{subfig:multi_shannon_chain}).
In \S\ref{section:MSSA}, we describe the randomized \hyperref[alg:multi_shannon_chain]{Multi-Step Shannon Algorithm} (MSSA for short) to construct short multi-step Shannon chains.
We will provide an overview of this algorithm below, but let us first describe how our algorithms in Theorem~\ref{theo:shannon} use the MSSA as a subprocedure.
\begin{itemize}
    \item Our randomized sequential algorithm iteratively picks a random uncolored edge and colors it by applying the MSSA.
    We are able to show that while the worst-case length of a multi-step Shannon chain is $\poly(\Delta)\log n$, the average case is $\poly(\Delta) = O(1)$, which proves the desired runtime.

    \item The randomized distributed algorithm runs in stages. 
    At each stage, the MSSA is run concurrently for each uncolored edge.
    We show that at least a $\poly(\Delta)$ fraction of the chains computed (in expectation) are vertex-disjoint, implying they can be simultaneously augmented.
    Each stage runs in $\poly(\Delta)\log n$ rounds and as we are coloring a $\poly(\Delta)$ fraction of the uncolored edges (in expectation) per stage, it follows that the multigraph will be colored in $\poly(\Delta) \log n$ stages.
    Therefore, the expected runtime is at most $\poly(\Delta) \log^2n$, which we show can be achieved with high probability.

    \item Finally, our deterministic distributed algorithm follows by derandomizing the randomized version by applying Harris's algorithm for hypergraph maximal matching \cite{harrisdistributed}, which has been applied similarly in \cite{VizingChain, Christ, fastEdgeColoring}.
    
\end{itemize}

\paragraph{Overview of the MSSA.}
Before we provide a sketch of the algorithm, we remark that the actual algorithm is rather technical, and so the steps followed are not identical to what we describe here.
However, this overview provides a general framework for the procedure and outlines the main ideas used in the analysis.
The main shortcoming of Shannon chains is that they can be ``long'', i.e., it could take $O(n)$ time to augment such a chain.
To remedy this, we consider a multi-step Shannon chain $C \defeq C_0 + \cdots + C_{k-1}$ such that $\length(C_i) \leq \ell$ for some parameter $\ell$.
When attaching a new chain $C_k$, we check whether $\length(C_k) \leq \ell$.
If it is, then $C_0 + \cdots + C_k$ is augmenting and short.
If not, we consider an \emphd{initial segment} of $C_k$, i.e., we pick some $\ell' \leq \ell$ and consider $C_k'$ to be the first $\ell'$ edges of $C_k$.
We then continue on to the next iteration with $C' \defeq C_0 + \cdots + C_k'$ (see Fig.~\ref{fig:forward_iteration}).

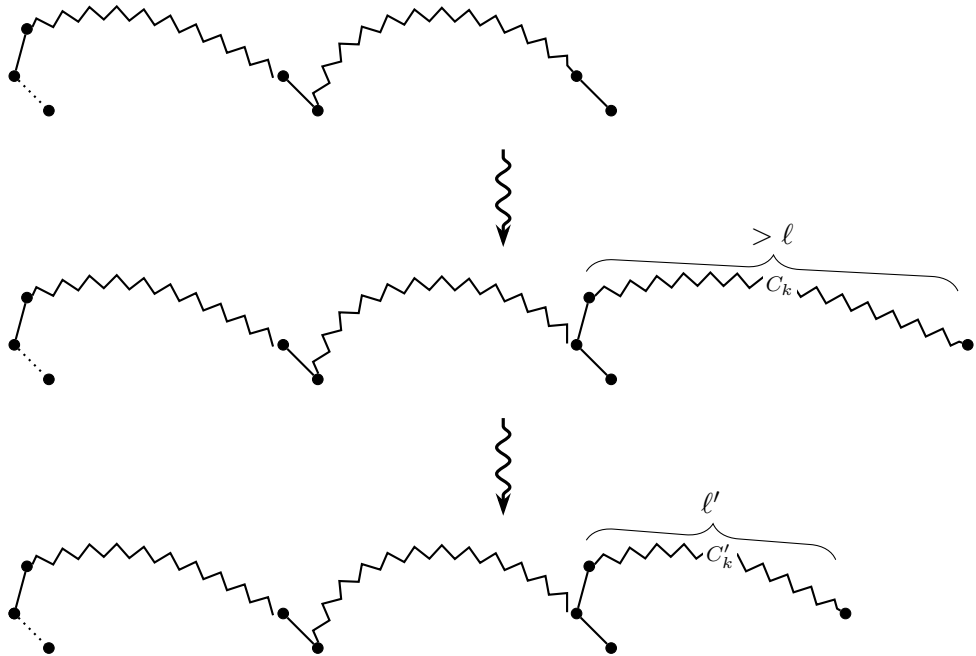
\begin{figure}[t]
    \centering
        \begin{tikzpicture}[xscale = 0.65,yscale=0.65]
            \node[circle,fill=black,draw,inner sep=0pt,minimum size=4pt] (a) at (0,0) {};
        	\path (a) ++(-45:1) node[circle,fill=black,draw,inner sep=0pt,minimum size=4pt] (b) {};
        	\path (a) ++(75:1) node[circle,fill=black,draw,inner sep=0pt,minimum size=4pt] (c) {};

                \path (a) ++(0:5.5) node[circle,fill=black,draw,inner sep=0pt,minimum size=4pt] (d) {};
            
        	\path (d) ++(-45:1) node[circle,fill=black,draw,inner sep=0pt,minimum size=4pt] (e) {};
        	
        	\path (d) ++(0:6) node[circle,fill=black,draw,inner sep=0pt,minimum size=4pt] (g) {};

                \path (g) ++(-45:1) node[circle,fill=black,draw,inner sep=0pt,minimum size=4pt] (h) {};
                \path (g) ++(75:1) node[circle,fill=black,draw,inner sep=0pt,minimum size=4pt] (i) {};
                \path (g) ++(0:8) node[circle,fill=black,draw,inner sep=0pt,minimum size=4pt] (j) {};

               \draw[thick,dotted] (a) -- (b);

               \draw[decoration={brace,amplitude=10pt},decorate] (11.7, 1.5) -- node [midway,above,yshift=10pt,xshift=0pt] {$> \ell$} (19.3,1.1);
        	
        	\draw[thick, decorate,decoration=zigzag] (c) to[out=20,in=150] (d) (e) to[out=85,in=130, looseness = 1.1] (g) (i) to[out=16,in=160] node[font=\fontsize{8}{8},midway,inner sep=1pt,outer sep=1pt,minimum size=4pt,fill=white] {$C_k$} (j);
        	
        	\draw[thick] (a) -- (c) (d) -- (e) (h) -- (g) -- (i);

        \begin{scope}[yshift=5.5cm]
            \node[circle,fill=black,draw,inner sep=0pt,minimum size=4pt] (a) at (0,0) {};
        	\path (a) ++(-45:1) node[circle,fill=black,draw,inner sep=0pt,minimum size=4pt] (b) {};
        	\path (a) ++(75:1) node[circle,fill=black,draw,inner sep=0pt,minimum size=4pt] (c) {};

                \path (a) ++(0:5.5) node[circle,fill=black,draw,inner sep=0pt,minimum size=4pt] (d) {};
            
        	\path (d) ++(-45:1) node[circle,fill=black,draw,inner sep=0pt,minimum size=4pt] (e) {};
        	
        	\path (d) ++(0:6) node[circle,fill=black,draw,inner sep=0pt,minimum size=4pt] (g) {};

                \path (g) ++(-45:1) node[circle,fill=black,draw,inner sep=0pt,minimum size=4pt] (h) {};

               \draw[thick,dotted] (a) -- (b);
               \draw[thick, decorate,decoration=zigzag] (c) to[out=20,in=150] (d) (e) to[out=85,in=130, looseness = 1.1] (g);
        	
        	\draw[thick] (a) -- (c) (d) -- (e) (h) -- (g);
        \end{scope}

        \begin{scope}[yshift=3cm]
            \draw[-{Stealth[length=3mm,width=2mm]},very thick,decoration = {snake,pre length=3pt,post length=7pt,},decorate] (10,1) -- (10,-1);
        \end{scope}

        \begin{scope}[yshift=-5.5cm]
            \node[circle,fill=black,draw,inner sep=0pt,minimum size=4pt] (a) at (0,0) {};
        	\node[circle,fill=black,draw,inner sep=0pt,minimum size=4pt] (a) at (0,0) {};
        	\path (a) ++(-45:1) node[circle,fill=black,draw,inner sep=0pt,minimum size=4pt] (b) {};
        	\path (a) ++(75:1) node[circle,fill=black,draw,inner sep=0pt,minimum size=4pt] (c) {};

                \path (a) ++(0:5.5) node[circle,fill=black,draw,inner sep=0pt,minimum size=4pt] (d) {};
            
        	\path (d) ++(-45:1) node[circle,fill=black,draw,inner sep=0pt,minimum size=4pt] (e) {};
        	
        	\path (d) ++(0:6) node[circle,fill=black,draw,inner sep=0pt,minimum size=4pt] (g) {};

                \path (g) ++(-45:1) node[circle,fill=black,draw,inner sep=0pt,minimum size=4pt] (h) {};
                \path (g) ++(75:1) node[circle,fill=black,draw,inner sep=0pt,minimum size=4pt] (i) {};
                \path (g) ++(0:5.5) node[circle,fill=black,draw,inner sep=0pt,minimum size=4pt] (j) {};

               \draw[thick,dotted] (a) -- (b);
               \draw[decoration={brace,amplitude=10pt},decorate] (11.7, 1.5) -- node [midway,above,yshift=10pt,xshift=0pt] {$\ell'$} (16.8,1.1);
               \draw[thick, decorate,decoration=zigzag] (c) to[out=20,in=150] (d) (e) to[out=85,in=130, looseness = 1.1] (g) (i) to[out=20,in=150] node[font=\fontsize{8}{8},midway,inner sep=1pt,outer sep=1pt,minimum size=4pt,fill=white] {$C_k'$} (j);
        	
        	\draw[thick] (a) -- (c) (d) -- (e) (h) -- (g) -- (i);
        \end{scope}

        \begin{scope}[yshift=-2.5cm]
            \draw[-{Stealth[length=3mm,width=2mm]},very thick,decoration = {snake,pre length=3pt,post length=7pt,},decorate] (10,1) -- (10,-1);
        \end{scope}
        	
        \end{tikzpicture}
    \caption{An iteration where $\length(C_k) > \ell$.}
    \label{fig:forward_iteration}
\end{figure}

However, it turns out this is not enough. 
We need to additionally ensure the multi-step Shannon chain is \emphd{non-intersecting}.
The precise definition of non-intersecting (provided in Definition~\ref{defn:non-int}) is rather technical.
For now, we may assume the $C_i$'s are edge disjoint.
Now we must add an additional step in each iteration to ensure this property holds.
If $C_k$ intersects $C_0 + \cdots + C_{k-1}$ at some chain $C_j$, we truncate $C$ to $C' \defeq C_0 + \cdots + C_{j-1}$ and continue to the next iteration (see Fig.~\ref{fig:backward_iteration}).

\begin{figure}[h]
    \centering
        \begin{tikzpicture}[xscale = 0.65,yscale=0.65]
        \clip (-0.5, -7) rectangle (20, 7.5);
            \node[circle,fill=black,draw,inner sep=0pt,minimum size=4pt] (a) at (0,0) {};
        	\path (a) ++(-45:1) node[circle,fill=black,draw,inner sep=0pt,minimum size=4pt] (b) {};
        	\path (a) ++(75:1) node[circle,fill=black,draw,inner sep=0pt,minimum size=4pt] (c) {};

                \path (a) ++(0:5.5) node[circle,fill=black,draw,inner sep=0pt,minimum size=4pt] (d) {};
            
        	\path (d) ++(-45:1) node[circle,fill=black,draw,inner sep=0pt,minimum size=4pt] (e) {};

                \node[circle, fill=gray,inner sep=1pt] (z) at (8, 1.2) {\textbf{!}};
        	
        	\path (d) ++(0:6) node[circle,fill=black,draw,inner sep=0pt,minimum size=4pt] (g) {};

                \path (g) ++(-45:1) node[circle,fill=black,draw,inner sep=0pt,minimum size=4pt] (h) {};
                \path (g) ++(75:1) node[circle,fill=black,draw,inner sep=0pt,minimum size=4pt] (i) {};
                \path (g) ++(0:5.5) node[circle,fill=black,draw,inner sep=0pt,minimum size=4pt] (j) {};

                \path (j) ++(-45:1) node[circle,fill=black,draw,inner sep=0pt,minimum size=4pt] (k) {};
                \path (j) ++(75:1) node[circle,fill=black,draw,inner sep=0pt,minimum size=4pt] (l) {};
                \path (z) ++(90:1) node[circle,fill=black,draw,inner sep=0pt,minimum size=4pt] (m) {};
        	
        	\draw[thick,dotted] (a) -- (b);
        	
        	\draw[thick, decorate,decoration=zigzag] (c) to[out=20,in=150] (d) (e) to[out=85,in=180] (z) to[out=0,in=130] node[font=\fontsize{8}{8},midway,inner sep=1pt,outer sep=1pt,minimum size=4pt,fill=white] {$C_j$} (g) (i) to[out=20,in=150] (j) (l) to[out=-10,in=-70, looseness=2] node[font=\fontsize{8}{8},midway,inner sep=1pt,outer sep=1pt,minimum size=4pt,fill=white] {$C_k$} (z) -- (m);
        	
        	\draw[thick] (a) -- (c) (d) -- (e) (h) -- (g) -- (i) (k) -- (j) -- (l);

        \begin{scope}[yshift=6cm]
            \node[circle,fill=black,draw,inner sep=0pt,minimum size=4pt] (a) at (0,0) {};
        	\path (a) ++(-45:1) node[circle,fill=black,draw,inner sep=0pt,minimum size=4pt] (b) {};
        	\path (a) ++(75:1) node[circle,fill=black,draw,inner sep=0pt,minimum size=4pt] (c) {};

                \path (a) ++(0:5.5) node[circle,fill=black,draw,inner sep=0pt,minimum size=4pt] (d) {};
            
        	\path (d) ++(-45:1) node[circle,fill=black,draw,inner sep=0pt,minimum size=4pt] (e) {};
        	
        	\path (d) ++(0:6) node[circle,fill=black,draw,inner sep=0pt,minimum size=4pt] (g) {};

                \path (g) ++(-45:1) node[circle,fill=black,draw,inner sep=0pt,minimum size=4pt] (h) {};
                \path (g) ++(75:1) node[circle,fill=black,draw,inner sep=0pt,minimum size=4pt] (i) {};
                \path (g) ++(0:5.5) node[circle,fill=black,draw,inner sep=0pt,minimum size=4pt] (j) {};

                \path (j) ++(-45:1) node[circle,fill=black,draw,inner sep=0pt,minimum size=4pt] (k) {};

               \draw[thick,dotted] (a) -- (b);
               \draw[thick, decorate,decoration=zigzag] (c) to[out=20,in=150] (d) (e) to[out=85,in=130, looseness = 1.1] (g) (i) to[out=20,in=150] (j);
        	
        	\draw[thick] (a) -- (c) (d) -- (e) (h) -- (g) -- (i) (k) -- (j);
        \end{scope}

        \begin{scope}[yshift=3.5cm]
            \draw[-{Stealth[length=3mm,width=2mm]},very thick,decoration = {snake,pre length=3pt,post length=7pt,},decorate] (10,1) -- (10,-1);
        \end{scope}

        \begin{scope}[yshift=-6cm]
            \node[circle,fill=black,draw,inner sep=0pt,minimum size=4pt] (a) at (0,0) {};
        	\node[circle,fill=black,draw,inner sep=0pt,minimum size=4pt] (a) at (0,0) {};
        	\path (a) ++(-45:1) node[circle,fill=black,draw,inner sep=0pt,minimum size=4pt] (b) {};
        	\path (a) ++(75:1) node[circle,fill=black,draw,inner sep=0pt,minimum size=4pt] (c) {};

                \path (a) ++(0:5.5) node[circle,fill=black,draw,inner sep=0pt,minimum size=4pt] (d) {};
            
        	\path (d) ++(-45:1) node[circle,fill=black,draw,inner sep=0pt,minimum size=4pt] (e) {};
        	


               \draw[thick,dotted] (a) -- (b);
               \draw[thick, decorate,decoration=zigzag] (c) to[out=20,in=150] node[font=\fontsize{8}{8},midway,inner sep=1pt,outer sep=1pt,minimum size=4pt,fill=white] {$C_{j-1}$} (d);
        	
        	\draw[thick] (a) -- (c) (d) -- (e);
        \end{scope}

        \begin{scope}[yshift=-3.5cm]
            \draw[-{Stealth[length=3mm,width=2mm]},very thick,decoration = {snake,pre length=3pt,post length=7pt,},decorate] (10,1) -- (10,-1);
        \end{scope}
        	
        \end{tikzpicture}
    \caption{An iteration where $C_k$ intersects $C_j$.}
    \label{fig:backward_iteration}
\end{figure}
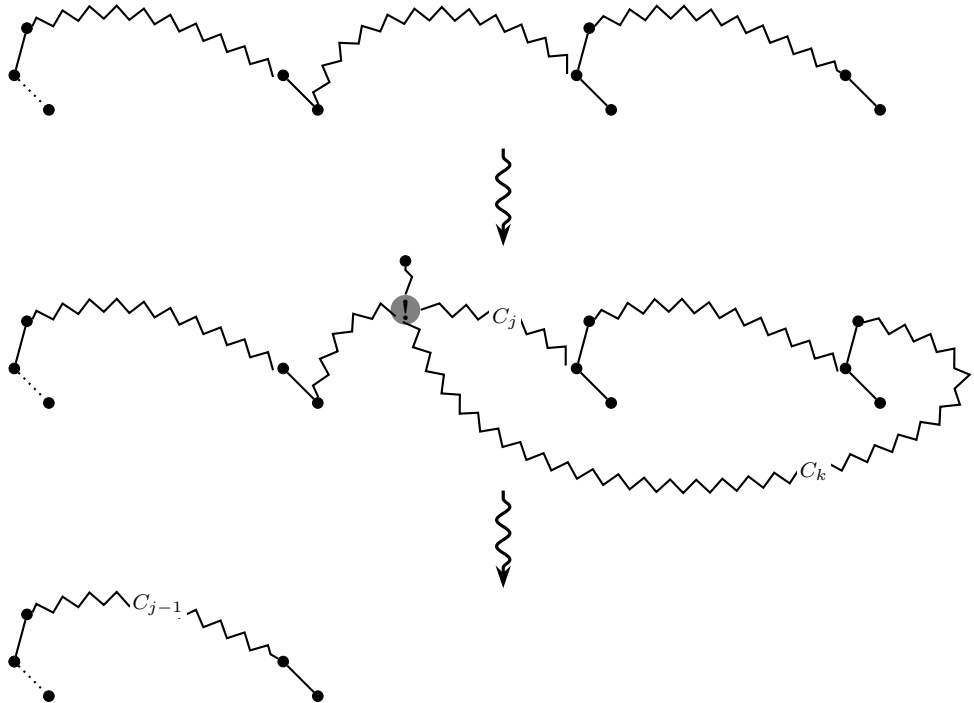

The crucial result of this paper is in Theorem~\ref{theo:MSSA}, which bounds the probability the MSSA lasts at least $t$ iterations under the assumption $\ell = \poly(\Delta)$.
As each iteration involves adding at most $\ell$ edges to the chain, it follows that the length of the chain output is at most $\poly(\Delta)\,t$.
As a corollary, we can show that $t = O(\log n)$ with probability at least $1 - 1/\poly(n)$.
Note that this implies the existence of an augmenting subgraph containing at most $\poly(\Delta)\log n$ edges, which matches the threshold in \cite{CHLPU}.
In order to analyse this algorithm (stated formally in Algorithm~\ref{alg:multi_shannon_chain}), we employ the entropy compression method in a similar way to \cite{fastEdgeColoring}.
The method was introduced by Moser and Tardos in their paper proving an algorithmic version of the Lov\'asz Local Lemma \cite{MT} (the term \textit{entropy compression} was coined by Tao in a blog post \cite{Tao}).
In \cite{Grytczuk}, the authors discovered that the method can circumvent the use of the Lov\'asz Local Lemma leading to improved combinatorial results.
We remark that in most applications of the method there is usually an alternate Lov\'asz Local Lemma argument that achieves a weaker bound, however,
as noted by the author and Bernshteyn in \cite{fastEdgeColoring}, it is unclear how a Lov\'asz Local Lemma argument would work in this setting, which makes the approach somewhat atypical.

The use of the entropy compression method in constructing multi-step augmenting subgraphs is a powerful new technique.
The fact that it extends to multigraphs is testament to its versatility.
The method is used to bound the probability a random process terminates within some specified time.
To do so, one encodes a \emphd{record} of the execution of the process in such a way that the original randomness in the execution can be recovered from the record.
In our application of the method, we encode the following: $+1$ if we are in the situation of Fig.~\ref{fig:forward_iteration}, and $j-k$ if we are in the situation of Fig.~\ref{fig:backward_iteration}.
Additionally, we record the final edge on the chain $C'$ at the end of the iteration.
We remark that this record does not recover the randomness of the process. However, we are able to provide a bound on the number of different processes that can result in a specific record.
This turns out to be enough for our probability bound. 
The details of the analysis are presented in \S\ref{subsection:analysis_MSSA}.

\paragraph{Optimality.}
We conclude this section with a brief discussion on the significance of our results.
First, we note that any sequential edge-coloring algorithm must have running time $\Omega(m)$ as we need to assign a color to each edge. 
With this observation, it follows that our randomized sequential algorithm is optimal (in the $\Delta = O(1)$ regime).
Furthermore, we stress that this is the \textbf{first} improvement in $\Sha$-edge-coloring algorithms for multigraphs since the original proof of \hyperref[theo:shannon_bound]{Shannon's theorem}.
Our randomized distributed algorithm is also optimal with respect to the augmenting subgraph approach.
As a result of \cite[Theorem~12]{CHLPU}, we cannot guarantee a running time better than $O(\log n)$ rounds for each stage.
We claim that $\Theta(\log n)$ stages are required as well, i.e., we cannot expect to color more than a constant fraction of the uncolored edges at most stages.
Consider the blank coloring of a $\Delta$-regular graph for example.
As we require augmenting chains to be vertex disjoint, we can at best color a matching in the first stage (which would contain at most $m/\Delta$ edges).
While the coloring is sparse, similar arguments hold.
In particular, we wouldn't expect to do better than a constant fraction until we are left with $o(n)$ edges, i.e., after $\Theta(\log n)$ stages.
Finally, we stress that our deterministic distributed algorithm improves upon the one in \cite{GKMU} by a factor of $\log^3n$.
This is also the first distributed algorithm for $\Sha$-edge-coloring to apply the augmenting subgraph approach.

%% file: prelim.tex
\section{Notation and Preliminaries on Augmenting Chains}\label{sec:notation}

This section is similar to \cite[\S 3]{fastEdgeColoring}, however, certain definitions differ due to the structure of multigraphs.
We additionally introduce Shannon chains and multi-step Shannon chains.

\subsection{Chains: general definitions}

We will assume $\Delta \geq 2$ as the problem is trivial when $G$ is a matching.
By considering $G$ to be an attribute of the coloring $\phi$, we may treat $G$ and $\Delta$ as global variables for all our algorithms.
Given a partial $r$-edge-coloring $\phi$ and $x\in V(G)$, we let
\[M(\phi, x) \defeq [r]\setminus\{\phi(e)\,:\, e \in E_G(x)\}\] 
be the set of all the \emphd{missing} colors at $x$ under the coloring $\phi$.
For $r = \Sha$, we have $|M(\phi, x)| \geq \lfloor \Delta/2\rfloor \geq 1$, and for $r = \Delta + \mu$, we have $|M(\phi, x)| \geq \mu \geq 1$. 
An uncolored edge $e \in E_G(x,y)$ is \emphd{$\phi$-happy} if $M(\phi, x)\cap M(\phi, y)\neq \0$. 
If $e \in E_G(x,y)$ is $\phi$-happy, we can extend the coloring $\phi$ by assigning any color in $M(\phi, x)\cap M(\phi, y)$ to $e$.

\begin{figure}[t]
	\centering
	\begin{tikzpicture}
	\begin{scope}
	\node[circle,fill=black,draw,inner sep=0pt,minimum size=4pt] (a) at (0,0) {};
	\node[circle,fill=black,draw,inner sep=0pt,minimum size=4pt] (b) at (-1.25,0.625) {};
	\node[circle,fill=black,draw,inner sep=0pt,minimum size=4pt] (c) at (0,1.625) {};
	\node[circle,fill=black,draw,inner sep=0pt,minimum size=4pt] (d) at (1.25,0.625) {};
	\node[circle,fill=black,draw,inner sep=0pt,minimum size=4pt] (e) at (2.5,0) {};
	\node[circle,fill=black,draw,inner sep=0pt,minimum size=4pt] (f) at (3.75,0.625) {};
	\node[circle,fill=black,draw,inner sep=0pt,minimum size=4pt] (g) at (3.75,-1) {};
	\node[circle,fill=black,draw,inner sep=0pt,minimum size=4pt] (h) at (5,0) {};
	
	\draw[thick,dotted] (a) to node[midway,inner sep=0pt,minimum size=4pt] (i) {} (b);
	\draw[thick] (a) to node[midway,inner sep=1pt,outer sep=1pt,minimum size=4pt,fill=white] (j) {$\alpha$} (c);
	\draw[thick] (a) to node[midway,inner sep=1pt,outer sep=1pt,minimum size=4pt,fill=white] (k) {$\beta$} (d);
	\draw[thick] (d) to node[midway,inner sep=1pt,outer sep=1pt,minimum size=4pt,fill=white] (l) {$\gamma$} (e);
	\draw[thick] (f) to node[midway,inner sep=1pt,outer sep=1pt,minimum size=4pt,fill=white] (m) {$\delta$} (e);
	\draw[thick] (f) to node[midway,inner sep=1pt,outer sep=1pt,minimum size=4pt,fill=white] (n) {$\epsilon$} (g);
	\draw[thick] (f) to node[midway,inner sep=1pt,outer sep=1pt,minimum size=4pt,fill=white] (o) {$\zeta$} (h);
	
	\draw[-{Stealth[length=1.6mm]}] (j) to[bend right] (i);
	\draw[-{Stealth[length=1.6mm]}] (k) to[bend right] (j);
	\draw[-{Stealth[length=1.6mm]}] (l) to[bend left] (k);
	\draw[-{Stealth[length=1.6mm]}] (m) to[bend right] (l);
	\draw[-{Stealth[length=1.6mm]}] (n) to[bend left] (m);
	\draw[-{Stealth[length=1.6mm]}] (o) to[bend left] (n);
	\end{scope}
	
	\draw[-{Stealth[length=1.6mm]},very thick,decoration = {snake,pre length=3pt,post length=7pt,},decorate] (5.4,0.3125) -- (6.35,0.3125);
	
	\begin{scope}[xshift=8cm]
	\node[circle,fill=black,draw,inner sep=0pt,minimum size=4pt] (a) at (0,0) {};
	\node[circle,fill=black,draw,inner sep=0pt,minimum size=4pt] (b) at (-1.25,0.625) {};
	\node[circle,fill=black,draw,inner sep=0pt,minimum size=4pt] (c) at (0,1.625) {};
	\node[circle,fill=black,draw,inner sep=0pt,minimum size=4pt] (d) at (1.25,0.625) {};
	\node[circle,fill=black,draw,inner sep=0pt,minimum size=4pt] (e) at (2.5,0) {};
	\node[circle,fill=black,draw,inner sep=0pt,minimum size=4pt] (f) at (3.75,0.625) {};
	\node[circle,fill=black,draw,inner sep=0pt,minimum size=4pt] (g) at (3.75,-1) {};
	\node[circle,fill=black,draw,inner sep=0pt,minimum size=4pt] (h) at (5,0) {};
	
	\draw[thick] (a) to node[midway,inner sep=1pt,outer sep=1pt,minimum size=4pt,fill=white] (i) {$\alpha$} (b);
	\draw[thick] (a) to node[midway,inner sep=1pt,outer sep=1pt,minimum size=4pt,fill=white] (j) {$\beta$} (c);
	\draw[thick] (a) to node[midway,inner sep=1pt,outer sep=1pt,minimum size=4pt,fill=white] (k) {$\gamma$} (d);
	\draw[thick] (d) to node[midway,inner sep=1pt,outer sep=1pt,minimum size=4pt,fill=white] (l) {$\delta$} (e);
	\draw[thick] (f) to node[midway,inner sep=1pt,outer sep=1pt,minimum size=4pt,fill=white] (m) {$\epsilon$} (e);
	\draw[thick] (f) to node[midway,inner sep=1pt,outer sep=1pt,minimum size=4pt,fill=white] (n) {$\zeta$} (g);
	\draw[thick,dotted] (f) to node[midway,inner sep=0pt,minimum size=4pt] (o) {} (h);
	\end{scope}
	\end{tikzpicture}
	\caption{Shifting a coloring along a chain (Greek letters represent colors).}\label{fig:shift}
\end{figure}
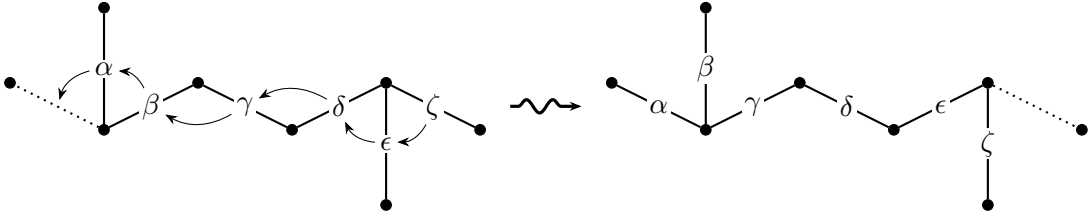

Given a proper partial coloring, we wish to modify it in order to create a partial coloring with a happy edge.
We will do so by finding small augmenting subgraphs called chains.
A \emphd{chain} of length $k$ is a sequence of distinct edges $C = (e_0, \ldots, e_{k-1})$ such that $|V(e_i)\cap V(e_{i+1})| = 1$ for each $0 \leq i < k-1$. 
Let $\Start(C) \defeq e_0$ and $\End(C) \defeq e_{k-1}$ denote the first and the last edges of $C$ respectively and let $\length(C) \defeq k$ be the length of $C$. 
We also let $E(C) \defeq \set{e_0, \ldots, e_{k-1}}$ be the edge set of $C$ and $V(C)$ be the set of vertices incident to the edges in $E(C)$.
To assist with defining our augmenting process, we will define the $\Shift$ operation (see Fig.~\ref{fig:shift}). 
Given a chain $C = (e_0, \ldots, e_{k-1})$, we define the coloring $\Shift(\phi, C)$ as follows:
\[\Shift(\phi, C)(e) \defeq \left\{\begin{array}{cc}
   \phi(e_{i+1})  & e = e_i, \, 0 \leq i < k-1; \\
    \blank & e = e_{k-1}; \\
    \phi(e) & \text{otherwise.}
\end{array}\right.\]
In other words, $\Shift(\phi, C)$ ``shifts'' the color from $e_{i+1}$ to $e_i$, leaves $e_{k-1}$ uncolored, and keeps the coloring on the other edges unchanged.
We call $C$ a \emphd{$\phi$-shiftable chain} if $\phi(e_0) = \blank$ and the coloring $\Shift(\phi, C)$ is proper.
For $C = (e_0, \ldots, e_{k-1})$, we let $C^* = (e_{k-1}, \ldots, e_0)$. 
It is clear that if $C$ is $\phi$-shiftable, then $C^*$ is $\Shift(\phi, C)$-shiftable and $\Shift(\Shift(\phi, C), C^*) = \phi$.

The next definition captures the class of chains that can be used to create a happy edge:

\begin{definition}[Happy chains]
    We say that a chain $C$ is \emphd{$\phi$-happy} for a partial coloring $\phi$ if it is $\phi$-shiftable and the edge $\End(C)$ is $\Shift(\phi, C)$-happy.
\end{definition}

Using the terminology of Definition~\ref{defn:aug}, we observe that a $\phi$-happy chain with $\Start(C) = e$ forms a connected $e$-augmenting subgraph of $G$ with respect to the partial coloring $\phi$. 
For a $\phi$-happy chain $C$, we let $\aug(\phi, C)$ be a proper coloring obtained from $\Shift(\phi, C)$ by assigning to $\End(C)$ some valid color. 
Our aim now becomes to develop algorithms for constructing $\phi$-happy chains for a given partial coloring $\phi$.
Before we describe our chains of interest, we define initial segments and sums of chains.

For a chain $C = (e_0, \ldots, e_{k-1})$ and $1 \leq j \leq k$, we define the \emphd{initial segment} of $C$ as
\[C|j \,\defeq\, (e_0, \ldots, e_{j-1}).\]
Given two chains $C = (e_0, \ldots, e_{k-1})$ and $C' = (f_0, \ldots, f_{j-1})$ with $e_{k-1} = f_0$, we define their \emphd{sum} as follows:
\[C+C' \,\defeq\, (e_0, \ldots, e_{k-1} = f_0, \ldots, f_{j-1}).\]
Note that if $C$ is $\phi$-shiftable and $C'$ is $\Shift(\phi, C)$-shiftable, then $C+C'$ is $\phi$-shiftable.

With these definitions in mind, we are now ready to describe the types of chains that will be used as building blocks in our algorithms.

\subsection{Path chains}\label{subsec:pathchains}

The first special type of chains we will consider are path chains. 

\begin{definition}
    A chain $P = (e_0, \ldots, e_{k-1})$ is a \emphd{path chain} if the edges $e_1$, \ldots, $e_{k-1}$ form a path in $G$, i.e., if there exist distinct vertices $x_1$, \ldots, $x_k$ such that $e_i \in E_G(x_i,x_{i+1})$ for all $1 \leq i \leq k-1$. We let $x_0 \in V(e_0)$ be distinct from $x_1$ and let $\vstart(P) \defeq x_0$, $\vend(P) \defeq x_k$ denote the first and last vertices on the path chain respectively. Note that the vertex $x_0$ may coincide with $x_i$ for some $3 \leq i \leq k$; see Fig.~\ref{subfig:unsucc} for an example. 
    Furthermore, the vertices $\vstart(P)$ and $\vend(P)$ are uniquely determined unless $P$ is a single edge.
\end{definition}

\begin{definition}[{Internal vertices and edges \cite[Definition 3.3]{fastEdgeColoring}}]\label{defn:internal}
    An edge of a path chain $P$ is \emphd{internal} if it is distinct from $\Start(P)$ and $\End(P)$. We let $\IE(P)$ denote the set of all internal edges of $P$. 
    A vertex of $P$ is \emphd{internal} if it is not incident to $\Start(P)$ or $\End(P)$. We let $\IV(P)$ denote the set of all internal vertices of $P$.
\end{definition}

\begin{figure}[t]
	\centering
        \begin{subfigure}[t]{\textwidth}
		\centering
		\begin{tikzpicture}
		\node[circle,fill=black,draw,inner sep=0pt,minimum size=4pt] (x) at (0,0) {};
            \node[circle,fill=black,draw,inner sep=0pt,minimum size=4pt] (y) at (1,0) {};
            \node[circle,fill=black,draw,inner sep=0pt,minimum size=4pt] (a) at (2,0) {};
            \node[circle,fill=black,draw,inner sep=0pt,minimum size=4pt] (b) at (3,0) {};
            \node[circle,fill=black,draw,inner sep=0pt,minimum size=4pt] (c) at (5,0) {};
            \node[circle,fill=black,draw,inner sep=0pt,minimum size=4pt] (d) at (6,0) {};
            \node[circle,fill=black,draw,inner sep=0pt,minimum size=4pt] (e) at (7,0) {};
            \node[circle,fill=black,draw,inner sep=0pt,minimum size=4pt] (f) at (-1,0) {};
            \node[circle,fill=black,draw,inner sep=0pt,minimum size=4pt] (g) at (-2,0) {};
            \node[circle,fill=black,draw,inner sep=0pt,minimum size=4pt] (h) at (-4,0) {};
            \node[circle,fill=black,draw,inner sep=0pt,minimum size=4pt] (i) at (-5,0) {};
            \node[circle,fill=black,draw,inner sep=0pt,minimum size=4pt] (j) at (-6,0) {};
		
		
		\draw[thick,dotted] (x) to node[midway,below,inner sep=1pt,outer sep=1pt,minimum size=4pt,fill=white] {$e$} (y);
		\draw[thick] (y) to node[midway,inner sep=1pt,outer sep=1pt,minimum size=4pt,fill=white] {$\alpha$} (a) to node[midway,inner sep=1pt,outer sep=1pt,minimum size=4pt,fill=white] {$\beta$} (b) (c) to node[midway,inner sep=1pt,outer sep=1pt,minimum size=4pt,fill=white] {$\beta$} (d) to node[midway,inner sep=1pt,outer sep=1pt,minimum size=4pt,fill=white] {$\alpha$} (e) 
        (x) to node[midway,inner sep=1pt,outer sep=1pt,minimum size=4pt,fill=white] {$\beta$} (f) to node[midway,inner sep=1pt,outer sep=1pt,minimum size=4pt,fill=white] {$\alpha$} (g) (h) to node[midway,inner sep=1pt,outer sep=1pt,minimum size=4pt,fill=white] {$\alpha$} (i) to node[midway,inner sep=1pt,outer sep=1pt,minimum size=4pt,fill=white] {$\beta$} (j);
        \draw[thick, snake=zigzag] (b) -- (c) (g) -- (h);

        \draw[decoration={brace,amplitude=10pt},decorate] (-6,0.2) -- node [midway,above,yshift=10pt,xshift=0pt] {$P(e; \phi, \beta\alpha)$} (1,0.2);
        \draw[decoration={brace,amplitude=10pt,mirror},decorate] (0,-0.2) -- node [midway,above,yshift=-28pt,xshift=0pt] {$P(e; \phi, \alpha\beta)$} (7,-0.2);
		
		\end{tikzpicture}
		\caption{$P(e; \phi, \alpha\beta)$ and $P(e; \phi, \beta\alpha)$.}\label{subfig:diff_paths}
	\end{subfigure}%
        \vspace{10pt}
	\begin{subfigure}[t]{.47\textwidth}
		\centering
		\begin{tikzpicture}
		\node[draw=none,minimum size=2.5cm,regular polygon,regular polygon sides=7] (P) {};

		\node[circle,fill=black,draw,inner sep=0pt,minimum size=4pt] (x) at (P.corner 4) {};
		\node[circle,fill=black,draw,inner sep=0pt,minimum size=4pt] (y) at (P.corner 5) {};
		\node[circle,fill=black,draw,inner sep=0pt,minimum size=4pt] (a) at (P.corner 6) {};
		\node[circle,fill=black,draw,inner sep=0pt,minimum size=4pt] (b) at (P.corner 7) {};
		\node[circle,fill=black,draw,inner sep=0pt,minimum size=4pt] (c) at (P.corner 1) {};
		\node[circle,fill=black,draw,inner sep=0pt,minimum size=4pt] (d) at (P.corner 2) {};
		\node[circle,fill=black,draw,inner sep=0pt,minimum size=4pt] (e) at (P.corner 3) {};
		
		\node[anchor=north] at (x) {$x$};
		\node[anchor=north] at (y) {$y$};
		
		\draw[thick,dotted] (x) to node[midway,below,inner sep=1pt,outer sep=1pt,minimum size=4pt,fill=white] {$e$} (y);
		\draw[thick] (y) to node[midway,inner sep=1pt,outer sep=1pt,minimum size=4pt,fill=white] {$\alpha$} (a);
		\draw[thick] (a) to node[midway,inner sep=1pt,outer sep=1pt,minimum size=4pt,fill=white] {$\beta$} (b);
		\draw[thick] (b) to node[midway,inner sep=1pt,outer sep=1pt,minimum size=4pt,fill=white] {$\alpha$} (c);
		\draw[thick] (c) to node[midway,inner sep=1pt,outer sep=1pt,minimum size=4pt,fill=white] {$\beta$} (d);
		\draw[thick] (d) to node[midway,inner sep=1pt,outer sep=1pt,minimum size=4pt,fill=white] {$\alpha$} (e);
		\draw[thick] (e) to node[midway,inner sep=1pt,outer sep=1pt,minimum size=4pt,fill=white] {$\beta$} (x);
		\end{tikzpicture}
		\caption{The edge $e \in E_G(x,y)$ is not $(\phi, \alpha \beta)$-successful.}\label{subfig:unsucc}
	\end{subfigure}%
	\qquad%
	\begin{subfigure}[t]{.47\textwidth}
		\centering
		\begin{tikzpicture}
		\node[draw=none,minimum size=2.5cm,regular polygon,regular polygon sides=7] (P) {};

		\node[circle,fill=black,draw,inner sep=0pt,minimum size=4pt] (x) at (P.corner 4) {};
		\node[circle,fill=black,draw,inner sep=0pt,minimum size=4pt] (y) at (P.corner 5) {};
		\node[circle,fill=black,draw,inner sep=0pt,minimum size=4pt] (a) at (P.corner 6) {};
		\node[circle,fill=black,draw,inner sep=0pt,minimum size=4pt] (b) at (P.corner 7) {};
		\node[circle,fill=black,draw,inner sep=0pt,minimum size=4pt] (c) at (P.corner 1) {};
		\node[circle,fill=black,draw,inner sep=0pt,minimum size=4pt] (d) at (P.corner 2) {};
		\node[circle,fill=black,draw,inner sep=0pt,minimum size=4pt] (e) at (P.corner 3) {};
		\node[circle,fill=black,draw,inner sep=0pt,minimum size=4pt] (f) at (-2.2,0) {}; 
		
		\node[anchor=north] at (x) {$x$};
		\node[anchor=north] at (y) {$y$};
		
		\draw[thick,dotted] (x) to node[midway,below,inner sep=1pt,outer sep=1pt,minimum size=4pt,fill=white] {$e$} (y);
		\draw[thick] (y) to node[midway,inner sep=1pt,outer sep=1pt,minimum size=4pt,fill=white] {$\alpha$} (a);
		\draw[thick] (a) to node[midway,inner sep=1pt,outer sep=1pt,minimum size=4pt,fill=white] {$\beta$} (b);
		\draw[thick] (b) to node[midway,inner sep=1pt,outer sep=1pt,minimum size=4pt,fill=white] {$\alpha$} (c);
		\draw[thick] (c) to node[midway,inner sep=1pt,outer sep=1pt,minimum size=4pt,fill=white] {$\beta$} (d);
		\draw[thick] (d) to node[midway,inner sep=1pt,outer sep=1pt,minimum size=4pt,fill=white] {$\alpha$} (e);
		\draw[thick] (e) to node[midway,inner sep=1pt,outer sep=1pt,minimum size=4pt,fill=white] {$\beta$} (f);
		\end{tikzpicture}
		\caption{The edge $e \in E_G(x,y)$ is $(\phi, \alpha \beta)$-successful.}
	\end{subfigure}
	\caption{Path chains.}
        \label{fig:path}
\end{figure}
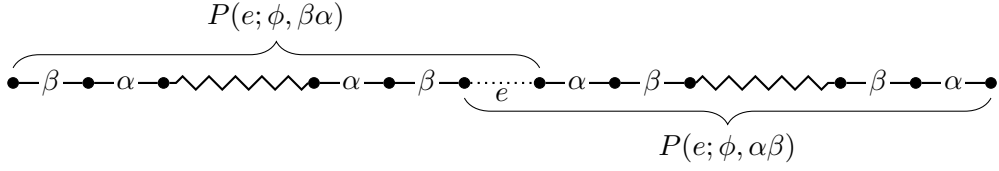
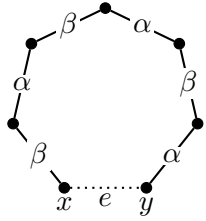
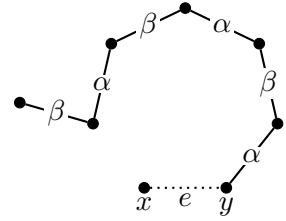

We are particularly interested in bicolored path chains, i.e., path chains containing at most $2$ colors, except possibly at their first edge. Specifically, given a proper partial coloring $\phi$ and $\alpha$, $\beta \in [r]$, we say that a path chain $P$ is an \emphd{$\alpha\beta$-path} under $\phi$ if all edges of $P$ except possibly $\Start(P)$ are colored $\alpha$ or $\beta$. To simplify working with $\alpha\beta$-paths, we introduce the following notation. 
Let $G(\phi, \alpha\beta)$ be the spanning subgraph of $G$ with 
\[E(G(\phi, \alpha\beta)) \,\defeq\, \{e\in E(G)\,:\, \phi(e) \in \{\alpha, \beta\}\}.\]
Since $\phi$ is proper, the maximum degree of $G(\phi, \alpha\beta)$ is at most $2$. Hence, the connected components of $G(\phi, \alpha\beta)$ are paths or cycles (an isolated vertex is a path of length $0$).
Note that there can be cycles of length $2$.
For $x\in V(G)$, let $G(x;\phi, \alpha\beta)$ denote the connected component of $G(\phi, \alpha\beta)$ containing $x$ and $\deg(x; \phi, \alpha\beta)$ denote the degree of $x$ in $G(\phi, \alpha\beta)$. We say that $x$, $y \in V(G)$ are \emphd{$(\phi, \alpha\beta)$-related} if $G(x;\phi, \alpha\beta) = G(y;\phi, \alpha\beta)$, i.e., if $y$ is reachable from $x$ by a path in $G(\phi, \alpha\beta)$.

\begin{definition}[{Hopeful and successful edges \cite[Definition 3.4]{fastEdgeColoring}}]\label{defn:hse}
    Let $\phi$ be a proper partial coloring of $G$ and let $\alpha$, $\beta \in [r]$. Let $x, y \in V(G)$ and let $e \in E_G(x, y)$ be an uncolored edge. 
    We say that $e$ is \emphd{$(\phi, \alpha\beta)$-hopeful} if $\deg(x;\phi, \alpha\beta) < 2$ and $\deg(y;\phi, \alpha\beta) < 2$. Further, we say that $e$ is \emphd{$(\phi, \alpha\beta)$-successful} if it is $(\phi, \alpha\beta)$-hopeful and $x$ and $y$ are not $(\phi, \alpha\beta)$-related.
\end{definition}

Let $\phi$ be a proper partial coloring and let $\alpha$, $\beta \in [r]$. Consider a $(\phi, \alpha\beta)$-hopeful edge $e \in E_G(x,y)$. 
We define the path chains $P(e; \phi, \alpha\beta)$ and $P(e; \phi, \beta\alpha)$ by
\[
    P(e; \phi, \alpha\beta) \,\defeq\, (e_0=e, e_1, \ldots, e_k), \quad P(e; \phi, \beta\alpha) \,\defeq\, (f_0=e, f_1, \ldots, f_l),
\]
where $(e_1, \ldots, e_k)$ is the maximal path in $G(\phi, \alpha\beta)$ such that $\phi(e_1) = \alpha$, and $(f_1, \ldots, f_l)$ is the maximal path in $G(\phi, \alpha\beta)$ such that $\phi(f_1) = \beta$.
It was shown in \cite[Fact 4.4,4.5]{VizingChain} that $P(e; \phi, \alpha\beta)$ is $\phi$-shiftable, and if $e$ is $(\phi, \alpha\beta)$-successful, then $P(e; \phi, \alpha\beta)$ is $\phi$-happy for simple graphs and $r = \Delta + 1$. 
An identical argument (\textit{mutatis mutandis}) works in the multigraph case for $r \in \set{\Sha, \Delta + \mu}$ as well.
See 
Fig.~\ref{fig:path} for an illustration.

\subsection{Fan chains}

The second type of chains we will be working with are fans:

\begin{definition}[Fans]\label{defn:fan}
    A \emphd{fan} is a chain of the form $F = (e_0, \ldots, e_{k-1})$ such that all edges contain a common vertex $x$, referred to as the \emphd{pivot} of the fan.
    With $y_i$ defined as the vertex in $V(e_i)$ distinct from $x$,
    we let $\Pivot(F) \defeq x$, $\vstart(F) \defeq y_0$ and $\vend(F) \defeq y_{k-1}$ denote the pivot, start, and end vertices of a fan $F$. 
    This notation is uniquely determined unless $F$ is a single edge, in which case we let $\vend(F) = \vstart(F)$.
    We note that it may be the case that $y_i = y_j$ for $i < j$ as $G$ is a multigraph.
    This is the only way in which this definition differs from \cite[Definition 3.5]{fastEdgeColoring}.
\end{definition}

The process of shifting a fan where $y_i = y_j$ is illustrated in Fig.~\ref{fig:fan}.

\begin{figure}[t]
	\centering
	\begin{tikzpicture}[xscale=1.1]
	\begin{scope}
	\node[circle,fill=black,draw,inner sep=0pt,minimum size=4pt] (x) at (0,0) {};
	\node[anchor=north] at (x) {$x$};
	
	\coordinate (O) at (0,0);
	\def\radius{2.6cm}
	
	\node[circle,fill=black,draw,inner sep=0pt,minimum size=4pt] (y0) at (190:\radius) {};
	\node at (190:2.9) {$y_0$};
	
	\node[circle,fill=black,draw,inner sep=0pt,minimum size=4pt] (y1) at (165:\radius) {};
	\node at (165:2.9) {$y_1$};
	
	\node[circle,fill=black,draw,inner sep=0pt,minimum size=4pt] (y2) at (140:\radius) {};
	\node at (140:2.9) {$y_2$};
	
	\node[circle,fill=black,draw,inner sep=0pt,minimum size=4pt] (y4) at (90:\radius) {};
	\node at (90:2.9) {$y_{i-1}$};
	
	\node[circle,fill=black,draw,inner sep=0pt,minimum size=4pt] (y5) at (62:\radius) {};
	\node at (62:2.9) {$y_i = y_j$};
	
	\node[circle,fill=black,draw,inner sep=0pt,minimum size=4pt] (y6) at (30:\radius) {};
	\node at (30:3) {$y_{i+1}$};
	
	\node[circle,fill=black,draw,inner sep=0pt,minimum size=4pt] (y8) at (-10:\radius) {};
	\node at (-10:3.1) {$y_{k-1}$};
	
	\node[circle,inner sep=0pt,minimum size=4pt] at (115:2) {$\ldots$}; 
	\node[circle,inner sep=0pt,minimum size=4pt] at (12:2) {$\ldots$}; 
	
	\draw[thick,dotted] (x) to (y0);
	\draw[thick] (x) to node[midway,inner sep=1pt,outer sep=1pt,minimum size=4pt,fill=white] {$\beta_0$} (y1);
	\draw[thick] (x) to node[midway,inner sep=1pt,outer sep=1pt,minimum size=4pt,fill=white] {$\beta_1$} (y2);
	
	\draw[thick] (x) to node[pos=0.55,inner sep=1pt,outer sep=1pt,minimum size=4pt,fill=white] {$\beta_{i-2}$} (y4);
	\draw[thick] (x) to[out=80, in=-140] node[pos=0.75,inner sep=1pt,outer sep=1pt,minimum size=4pt,fill=white] {$\beta_{i-1}$} (y5);
        \draw[thick] (x) to[out=40, in=-100] node[pos=0.6,inner sep=1pt,outer sep=1pt,minimum size=4pt,fill=white] {$\beta_{j-1}$} (y5);
	\draw[thick] (x) to node[midway,inner sep=1pt,outer sep=1pt,minimum size=4pt,fill=white] {$\beta_i$} (y6);
	
	\draw[thick] (x) to node[midway,inner sep=1pt,outer sep=1pt,minimum size=4pt,fill=white] {$\beta_{k-2}$} (y8);
	\end{scope}
	
	\draw[-{Stealth[length=1.6mm]},very thick,decoration = {snake,pre length=3pt,post length=7pt,},decorate] (2.9,1) to node[midway,anchor=south]{$\Shift$} (5,1);
	
	\begin{scope}[xshift=8.3cm]
	\node[circle,fill=black,draw,inner sep=0pt,minimum size=4pt] (x) at (0,0) {};
	\node[anchor=north] at (x) {$x$};
	
	\coordinate (O) at (0,0);
	\def\radius{2.6cm}
	
	\node[circle,fill=black,draw,inner sep=0pt,minimum size=4pt] (y0) at (190:\radius) {};
	\node at (190:2.9) {$y_0$};
	
	\node[circle,fill=black,draw,inner sep=0pt,minimum size=4pt] (y1) at (165:\radius) {};
	\node at (165:2.9) {$y_1$};
	
	\node[circle,fill=black,draw,inner sep=0pt,minimum size=4pt] (y2) at (140:\radius) {};
	\node at (140:2.9) {$y_2$};
	
	\node[circle,fill=black,draw,inner sep=0pt,minimum size=4pt] (y4) at (90:\radius) {};
	\node at (90:2.9) {$y_{i-1}$};
	
	\node[circle,fill=black,draw,inner sep=0pt,minimum size=4pt] (y5) at (62:\radius) {};
	\node at (62:2.9) {$y_i = y_j$};
	
	\node[circle,fill=black,draw,inner sep=0pt,minimum size=4pt] (y6) at (30:\radius) {};
	\node at (30:3) {$y_{i+1}$};
	
	\node[circle,fill=black,draw,inner sep=0pt,minimum size=4pt] (y8) at (-10:\radius) {};
	\node at (-10:3.1) {$y_{k-1}$};
	
	\node[circle,inner sep=0pt,minimum size=4pt] at (115:2) {$\ldots$}; 
	\node[circle,inner sep=0pt,minimum size=4pt] at (12:2) {$\ldots$}; 
	
	\draw[thick] (x) to node[midway,inner sep=1pt,outer sep=1pt,minimum size=4pt,fill=white] {$\beta_0$} (y0);
	\draw[thick] (x) to node[midway,inner sep=1pt,outer sep=1pt,minimum size=4pt,fill=white] {$\beta_1$} (y1);
	\draw[thick] (x) to node[midway,inner sep=1pt,outer sep=1pt,minimum size=4pt,fill=white] {$\beta_2$} (y2);
	
	\draw[thick] (x) to node[pos=0.55,inner sep=1pt,outer sep=1pt,minimum size=4pt,fill=white] {$\beta_{i-1}$} (y4);
	\draw[thick] (x) to[out=80, in=-140] node[pos=0.75,inner sep=1pt,outer sep=1pt,minimum size=4pt,fill=white] {$\beta_{i}$} (y5);
        \draw[thick] (x) to[out=40, in=-100] node[pos=0.6,inner sep=1pt,outer sep=1pt,minimum size=4pt,fill=white] {$\beta_{j}$} (y5);
	\draw[thick] (x) to node[midway,inner sep=1pt,outer sep=1pt,minimum size=4pt,fill=white] {$\beta_{i+1}$} (y6);
	
	\draw[thick, dotted] (x) to (y8);
	\end{scope}
	\end{tikzpicture}
	\caption{The process of shifting a fan.}\label{fig:fan}
\end{figure}
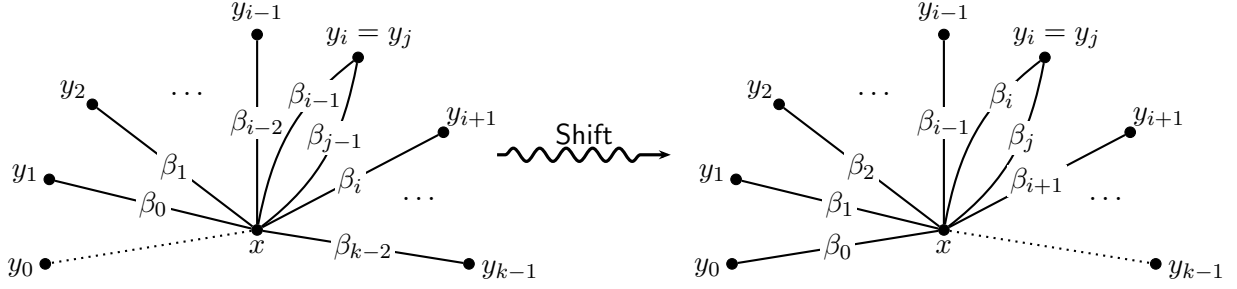


\begin{definition}[{Hopeful, successful, and disappointed fans \cite[Definition 3.6]{fastEdgeColoring}}]\label{defn:hsf}
    Let $\phi$ be a proper partial coloring and let $\alpha$, $\beta \in [r]$. Let $F$ be a $\phi$-shiftable fan with $x \defeq \Pivot(F)$ and $y \defeq \vend(F)$ and suppose that $F$ is not $\phi$-happy (which means that the edge $\End(F) \in E_G(x,y)$ is not $\Shift(\phi, F)$-happy). We say that $F$ is:
    \begin{itemize}
        \item \emphd{$(\phi, \alpha\beta)$-hopeful} if $\deg(x;\phi, \alpha\beta) < 2$ and $\deg(y;\phi, \alpha\beta) < 2$;
        \item \emphd{$(\phi, \alpha\beta)$-successful} if it is $(\phi, \alpha\beta)$-hopeful and $x$ and $y$ are not $(\Shift(\phi, F), \alpha\beta)$-related;
        \item \emphd{$(\phi, \alpha\beta)$-disappointed} if it is $(\phi, \alpha\beta)$-hopeful but not $(\phi, \alpha\beta)$-successful.
    \end{itemize}
\end{definition}

The following Fact will assist with our proofs.

\begin{fact}\label{fact:fan}
    Let $\phi$ be a proper partial coloring, $F = (e_0, \ldots, e_{k-1})$ be a $\phi$-shiftable fan and set $x \defeq \Pivot(F)$, $y \defeq \vend(F)$. 
    Define $\psi\defeq \Shift(\phi, F)$ and let $I = \{0 \leq i < k-1\,:\, e_i \in E_G(x, y)\}$.
    If $F$ is $(\phi, \alpha\beta)$-hopeful (resp. successful) and $\phi(e_{i+1}) \notin \set{\alpha, \beta}$ for each $i \in I$, then the edge $\End(F)$ is $(\psi, \alpha\beta)$-hopeful (resp. successful).
\end{fact}

\begin{proof}
    Note that 
    \[M(\psi, x) = M(\phi, x), \quad M(\psi, y) = M(\phi, y) \cup \set{\phi(e_i)\,:\,i \in I} \setminus \set{\phi(e_{i+1})\,:\, i \in I}.\]
    In particular, as $\phi(e_{i+1}) \notin \set{\alpha,\beta}$ for $i \in I$, we must have
    \[\deg(x;\psi, \alpha\beta) = \deg(x;\phi, \alpha\beta) < 2, \quad \deg(y;\psi, \alpha\beta) \leq \deg(y;\phi, \alpha\beta) < 2.\]
    The claim now follows by Definitions~\ref{defn:hse} and~\ref{defn:hsf}.
\end{proof}


\subsection{Shannon chains}\label{subsec:Sha_def}

The third type of chain we will be working with are Shannon chains.
For $\Sha$-edge-colorings, a Shannon Chain is formed by combining a fan of length at most $2$ and a bicolored path.

\begin{definition}[Shannon chains]\label{def:shannon_chain}
    A \emphd{Shannon chain} in a proper partial coloring $\phi$ is a chain of the form $F + P$, where $F$ is a $(\phi, \alpha\beta)$-hopeful fan for some $\alpha$, $\beta \in [\Sha]$, $\length(F) \leq 2$, and $P$ is an initial segment of the path chain $P(\End(F);  \Shift(\phi, F), \alpha\beta)$ such that $\vstart(P) = \Pivot(F)$.
    In particular, letting $x \defeq \Pivot(F)$ and $y \defeq \vend(F)$, the chain $P$ consists of the edge $\End(F) \in E_G(x,y)$ followed by a (not necessarily maximal) path starting at $y$ whose edges are colored $\alpha$ and $\beta$ under the coloring $\Shift(\phi, F)$ such that $\beta \in M(\Shift(\phi, F), y)$. See Fig.~\ref{fig:shannon_chains} for an illustration. 
    Note that we allow $P$ to comprise only the single edge $\End(F)$, in which case the Shannon chain coincides with the fan $F$.
    Similarly, if $\length(F) = 1$, the Shannon chain coincides with the path $P$.
\end{definition}

\begin{figure}[t]
    \centering
        \begin{tikzpicture}
            \node[circle,fill=black,draw,inner sep=0pt,minimum size=4pt] (a) at (0,0) {};
        	\path (a) ++(-45:1) node[circle,fill=black,draw,inner sep=0pt,minimum size=4pt] (b) {};
        	\path (a) ++(75:1) node[circle,fill=black,draw,inner sep=0pt,minimum size=4pt] (c) {};

                \path (c) ++(25:1) node[circle,fill=black,draw,inner sep=0pt,minimum size=4pt] (d) {};
                \path (d) ++(10:1) node[circle,fill=black,draw,inner sep=0pt,minimum size=4pt] (e) {};
                \path (a) ++(0:5.5) node[circle,fill=black,draw,inner sep=0pt,minimum size=4pt] (g) {};
                \path (g) ++(140:1) node[circle,fill=black,draw,inner sep=0pt,minimum size=4pt] (f) {};
                \path (g) ++(-45:1) node[circle,fill=black,draw,inner sep=0pt,minimum size=4pt] (h) {};
        	
        	\draw[thick,dotted] (h) -- (g);
        	
        	\draw[thick, decorate,decoration=zigzag] (e) to[out=0,in=150] (f);
        	
        	\draw[thick] (b) to node[midway,inner sep=1pt,outer sep=1pt,minimum size=4pt,fill=white] {$\eta$} (a) to node[midway,inner sep=1pt,outer sep=1pt,minimum size=4pt,fill=white] {$\alpha$} (c) to node[midway,inner sep=1pt,outer sep=1pt,minimum size=4pt,fill=white] {$\beta$} (d) to node[midway,inner sep=1pt,outer sep=1pt,minimum size=4pt,fill=white] {$\alpha$} (e) (f) to node[midway,inner sep=1pt,outer sep=1pt,minimum size=4pt,fill=white] {$\alpha$} (g);

        \begin{scope}[yshift=3.3cm]
            \node[circle,fill=black,draw,inner sep=0pt,minimum size=4pt] (a) at (0,0) {};
        	\path (a) ++(-45:1) node[circle,fill=black,draw,inner sep=0pt,minimum size=4pt] (b) {};
        	\path (a) ++(75:1) node[circle,fill=black,draw,inner sep=0pt,minimum size=4pt] (c) {};

                \path (c) ++(25:1) node[circle,fill=black,draw,inner sep=0pt,minimum size=4pt] (d) {};
                \path (d) ++(10:1) node[circle,fill=black,draw,inner sep=0pt,minimum size=4pt] (e) {};
                \path (a) ++(0:5.5) node[circle,fill=black,draw,inner sep=0pt,minimum size=4pt] (g) {};
                \path (g) ++(140:1) node[circle,fill=black,draw,inner sep=0pt,minimum size=4pt] (f) {};
                \path (g) ++(-45:1) node[circle,fill=black,draw,inner sep=0pt,minimum size=4pt] (h) {};

                \draw[decoration={brace,amplitude=10pt},decorate] (0.4,0.95) -- node [midway,above,yshift=-7pt,xshift=15pt] {$F$} (0.85,-0.69);
                \draw[decoration={brace,amplitude=10pt},decorate] (0,1.6) -- node [midway,above,yshift=10pt,xshift=0pt] {$P$} (6.2,1.6);
        	
        	\draw[thick,dotted] (a) -- (b);
        	
        	\draw[thick, decorate,decoration=zigzag] (e) to[out=0,in=150] (f);
        	
        	\draw[thick] (a) to node[midway,inner sep=1pt,outer sep=1pt,minimum size=4pt,fill=white] {$\eta$} (c) to node[midway,inner sep=1pt,outer sep=1pt,minimum size=4pt,fill=white] {$\alpha$} (d) to node[midway,inner sep=1pt,outer sep=1pt,minimum size=4pt,fill=white] {$\beta$} (e) (f) to node[midway,inner sep=1pt,outer sep=1pt,minimum size=4pt,fill=white] {$\beta$} (g) to node[midway,inner sep=1pt,outer sep=1pt,minimum size=4pt,fill=white] {$\alpha$} (h);
        \end{scope}

        \begin{scope}[yshift=3cm]
            \draw[-{Stealth[length=3mm,width=2mm]},very thick,decoration = {snake,pre length=3pt,post length=7pt,},decorate] (3,0) -- (3,-1);
        \end{scope}
        	
        \end{tikzpicture}
    \caption{Shifting a Shannon chain $C = F+P$.}
    \label{fig:shannon_chains}
\end{figure}
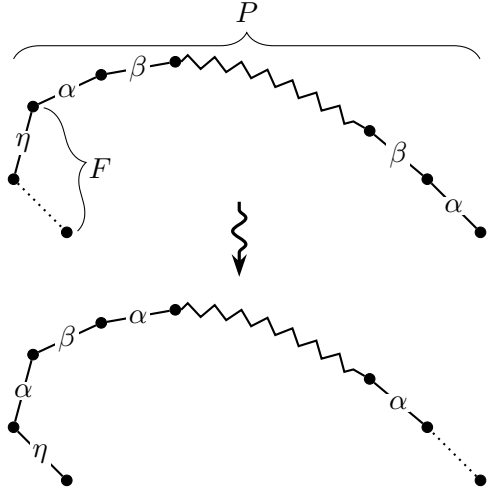

It follows from standard proofs of \hyperref[theo:Shannon]{Shannon's theorem} that for any uncolored edge $e$, one can find a $\phi$-happy chain $C$ with $\Start(C) = e$ such that $C$ is a Shannon chain $F+P$ (see Algorithm~\ref{alg:shannon_chain} for a construction).
We provide such a proof in Corollary~\ref{corl:shannon}.

Combining several Shannon chains yields a multi-step Shannon chain.

\begin{definition}[Multi-step Shannon chains]\label{def:multi_shannon_chain}
    A \emphd{$k$-step Shannon chain} is a chain of the form \[C \,=\, C_0 + \cdots + C_{k-1},\] where $C_i = F_i + P_i$ is a Shannon chain in the coloring $\Shift(\phi, C_0 + \cdots + C_{i-1})$ such that $\vend(P_i) = \vstart(F_{i+1})$ for all $0 \leq i < k-1$. 
    See Fig.~\ref{fig:multi_shannon_chains} for an illustration.
\end{definition}

\begin{figure}[!b]
    \centering
        \begin{tikzpicture}[xscale=0.8,yscale=0.8]
            \node[circle,fill=black,draw,inner sep=0pt,minimum size=4pt] (a) at (0,0) {};
        	\path (a) ++(-45:1) node[circle,fill=black,draw,inner sep=0pt,minimum size=4pt] (b) {};
        	\path (a) ++(75:1) node[circle,fill=black,draw,inner sep=0pt,minimum size=4pt] (c) {};

                \path (c) ++(25:1) node[circle,fill=black,draw,inner sep=0pt,minimum size=4pt] (d) {};
                \path (d) ++(10:1) node[circle,fill=black,draw,inner sep=0pt,minimum size=4pt] (e) {};
                \path (a) ++(0:5.5) node[circle,fill=black,draw,inner sep=0pt,minimum size=4pt] (g) {};
                \path (g) ++(140:1) node[circle,fill=black,draw,inner sep=0pt,minimum size=4pt] (f) {};
                \path (g) ++(-45:1) node[circle,fill=black,draw,inner sep=0pt,minimum size=4pt] (h) {};

                \path (h) ++(70:1) node[circle,fill=black,draw,inner sep=0pt,minimum size=4pt] (i) {};

                \path (i) ++(55:1) node[circle,fill=black,draw,inner sep=0pt,minimum size=4pt] (j) {};
                \path (g) ++(0:6) node[circle,fill=black,draw,inner sep=0pt,minimum size=4pt] (m) {};
                \path (m) ++(140:1) node[circle,fill=black,draw,inner sep=0pt,minimum size=4pt] (l) {};
                \path (m) ++(-45:1) node[circle,fill=black,draw,inner sep=0pt,minimum size=4pt] (n) {};

                \path (m) ++(75:1) node[circle,fill=black,draw,inner sep=0pt,minimum size=4pt] (o) {};

                \path (o) ++(25:1) node[circle,fill=black,draw,inner sep=0pt,minimum size=4pt] (p) {};
                \path (p) ++(10:1) node[circle,fill=black,draw,inner sep=0pt,minimum size=4pt] (q) {};
                \path (m) ++(0:5.5) node[circle,fill=black,draw,inner sep=0pt,minimum size=4pt] (s) {};
                \path (s) ++(140:1) node[circle,fill=black,draw,inner sep=0pt,minimum size=4pt] (r) {};
                \path (s) ++(-45:1) node[circle,fill=black,draw,inner sep=0pt,minimum size=4pt] (t) {};
        	
        	\draw[thick,dotted] (s) -- (t);
        	
        	\draw[thick, decorate,decoration=zigzag] (e) to[out=0,in=150] (f) (q) to[out=0,in=150] (r) (j) to[out=35,in=150] (l);
        	
        	\draw[thick] (b) -- (a) to node[midway,inner sep=1pt,outer sep=1pt,minimum size=4pt,fill=white] {$\alpha_1$} (c) to node[midway,inner sep=1pt,outer sep=1pt,minimum size=4pt,fill=white] {$\beta_1$} (d) to node[midway,inner sep=1pt,outer sep=1pt,minimum size=4pt,fill=white] {$\alpha_1$} (e) (f) to node[midway,inner sep=1pt,outer sep=1pt,minimum size=4pt,fill=white] {$\alpha_1$} (g)  to node[midway,inner sep=1pt,outer sep=1pt,minimum size=4pt,fill=white] {$\alpha_2$} (h) to node[midway,inner sep=1pt,outer sep=1pt,minimum size=4pt,fill=white] {$\beta_2$} (i) to node[midway,inner sep=1pt,outer sep=1pt,minimum size=4pt,fill=white] {$\alpha_2$} (j) (l) to node[midway,inner sep=1pt,outer sep=1pt,minimum size=4pt,fill=white] {$\alpha_2$} (m) -- (n)
                (m)  to node[midway,inner sep=1pt,outer sep=1pt,minimum size=4pt,fill=white] {$\alpha_3$} (o) to node[midway,inner sep=1pt,outer sep=1pt,minimum size=4pt,fill=white] {$\beta_3$} (p) to node[midway,inner sep=1pt,outer sep=1pt,minimum size=4pt,fill=white] {$\alpha_3$} (q) (r) to node[midway,inner sep=1pt,outer sep=1pt,minimum size=4pt,fill=white] {$\alpha_3$} (s)
            ;

        \begin{scope}[yshift=3.5cm]
            \node[circle,fill=black,draw,inner sep=0pt,minimum size=4pt] (a) at (0,0) {};
        	\path (a) ++(-45:1) node[circle,fill=black,draw,inner sep=0pt,minimum size=4pt] (b) {};
        	\path (a) ++(75:1) node[circle,fill=black,draw,inner sep=0pt,minimum size=4pt] (c) {};

                \path (c) ++(25:1) node[circle,fill=black,draw,inner sep=0pt,minimum size=4pt] (d) {};
                \path (d) ++(10:1) node[circle,fill=black,draw,inner sep=0pt,minimum size=4pt] (e) {};
                \path (a) ++(0:5.5) node[circle,fill=black,draw,inner sep=0pt,minimum size=4pt] (g) {};
                \path (g) ++(140:1) node[circle,fill=black,draw,inner sep=0pt,minimum size=4pt] (f) {};
                \path (g) ++(-45:1) node[circle,fill=black,draw,inner sep=0pt,minimum size=4pt] (h) {};

                \path (h) ++(70:1) node[circle,fill=black,draw,inner sep=0pt,minimum size=4pt] (i) {};

                \path (i) ++(55:1) node[circle,fill=black,draw,inner sep=0pt,minimum size=4pt] (j) {};
                \path (g) ++(0:6) node[circle,fill=black,draw,inner sep=0pt,minimum size=4pt] (m) {};
                \path (m) ++(140:1) node[circle,fill=black,draw,inner sep=0pt,minimum size=4pt] (l) {};
                \path (m) ++(-45:1) node[circle,fill=black,draw,inner sep=0pt,minimum size=4pt] (n) {};

                \path (m) ++(75:1) node[circle,fill=black,draw,inner sep=0pt,minimum size=4pt] (o) {};

                \path (o) ++(25:1) node[circle,fill=black,draw,inner sep=0pt,minimum size=4pt] (p) {};
                \path (p) ++(10:1) node[circle,fill=black,draw,inner sep=0pt,minimum size=4pt] (q) {};
                \path (m) ++(0:5.5) node[circle,fill=black,draw,inner sep=0pt,minimum size=4pt] (s) {};
                \path (s) ++(140:1) node[circle,fill=black,draw,inner sep=0pt,minimum size=4pt] (r) {};
                \path (s) ++(-45:1) node[circle,fill=black,draw,inner sep=0pt,minimum size=4pt] (t) {};

                \draw[thick,dotted] (a) -- (b);
        	
        	\draw[thick, decorate,decoration=zigzag] (e) to[out=0,in=150] (f) (q) to[out=0,in=150] (r) (j) to[out=35,in=150] (l);
        	
        	\draw[thick] (a) -- (c) to node[midway,inner sep=1pt,outer sep=1pt,minimum size=4pt,fill=white] {$\alpha_1$} (d) to node[midway,inner sep=1pt,outer sep=1pt,minimum size=4pt,fill=white] {$\beta_1$} (e) (f) to node[midway,inner sep=1pt,outer sep=1pt,minimum size=4pt,fill=white] {$\beta_1$} (g) to node[midway,inner sep=1pt,outer sep=1pt,minimum size=4pt,fill=white] {$\alpha_1$} (h) to node[midway,inner sep=1pt,outer sep=1pt,minimum size=4pt,fill=white] {$\alpha_2$} (i) to node[midway,inner sep=1pt,outer sep=1pt,minimum size=4pt,fill=white] {$\beta_2$} (j) (l) to node[midway,inner sep=1pt,outer sep=1pt,minimum size=4pt,fill=white] {$\beta_2$} (m) to node[midway,inner sep=1pt,outer sep=1pt,minimum size=4pt,fill=white] {$\alpha_2$} (n)
                (m) -- (o) to node[midway,inner sep=1pt,outer sep=1pt,minimum size=4pt,fill=white] {$\alpha_3$} (p) to node[midway,inner sep=1pt,outer sep=1pt,minimum size=4pt,fill=white] {$\beta_3$} (q) (r) to node[midway,inner sep=1pt,outer sep=1pt,minimum size=4pt,fill=white] {$\beta_3$} (s) to node[midway,inner sep=1pt,outer sep=1pt,minimum size=4pt,fill=white] {$\alpha_3$} (t)
            ;
        \end{scope}

        \begin{scope}[yshift=3cm]
            \draw[-{Stealth[length=3mm,width=2mm]},very thick,decoration = {snake,pre length=3pt,post length=7pt,},decorate] (8.5,0) -- (8.5,-1);
        \end{scope}
        	
        \end{tikzpicture}
    \caption{Shifting a $3$-step Shannon chain $C$.}
    \label{fig:multi_shannon_chains}
\end{figure}
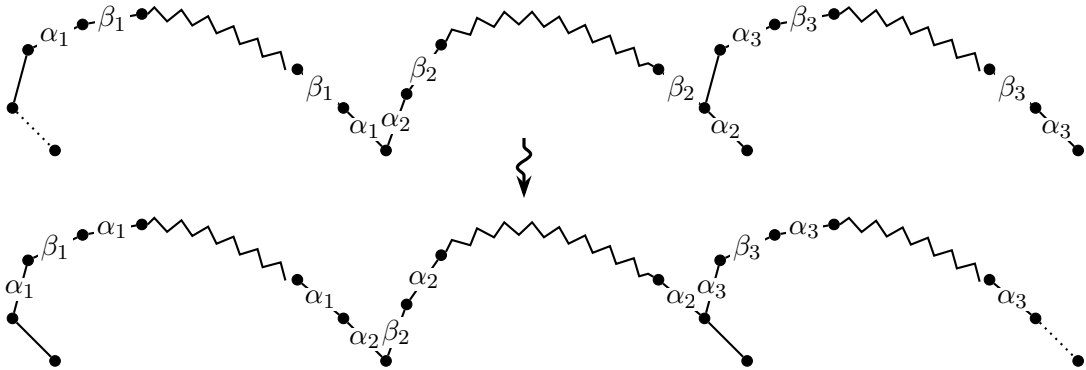

\subsection{Fan Algorithms}

We conclude this section by outlining two procedures to construct fan chains, which will be used as subroutines in our main algorithms for \hyperref[theo:shannon_bound]{Shannon's Theorem}.
For the remainder of the paper (apart from \S\ref{section:vizing}) we will be considering $\Sha$-edge-colorings.

We start with the \hyperref[alg:first_shannon_fan]{First Shannon Fan Algorithm}, which is described formally in Algorithm~\ref{alg:first_shannon_fan}. The algorithm takes as input a proper partial coloring $\phi$, an uncolored edge $e \in E(G)$, and a choice of a pivot vertex $x \in V(e)$.
The output is a tuple $(F, \alpha, \beta)$ such that:
\begin{itemize}
    \item $F$ is a fan with $\Start(F) = e$, $\Pivot(F) = x$ and $\length(F) \leq 2$, 
    \item $\alpha \in [\Sha]$ is a color in $M(\phi, x)$, and
    \item $\beta \in [\Sha]$ is a color such that either $\beta$ is missing at both $\vend(F)$ and $\Pivot(F)$, or $\beta$ is missing at both $\vend(F)$ and $\vstart(F)$.
\end{itemize}
We first check whether $e$ is $\phi$-happy.
If so, we return $((e), \beta, \beta)$, where $\beta \in M(\phi, x) \cap M(\phi, y)$ for $y \in V(e)$ distinct from $x$.
We may now assume $M(\phi, x) \cap M(\phi, y) = \0$.
Let $\eta \in M(\phi, y)$ be arbitrary and let $f \in E_G(x)$ be such that $\phi(f) = \eta$.
The fan we will return is $F = (e,f)$.
It remains to determine $\alpha$ and $\beta$.
Let $z \defeq \vend (F)$ be the vertex in $V(f)$ distinct from $x$. 
Note that $z \neq y$ as $\eta \in M(\phi, y)$.
If $M(\phi, x) \cap M(\phi, z) \neq \0$, we let $\alpha = \beta$ be an arbitrary color missing at both $x$ and $z$. 
If not, we let $\alpha \in M(\phi, x)$ be arbitrary.
Furthermore, we claim that $M(\phi, y) \cap M(\phi, z) \neq \0$ and let $\beta$ be an arbitrary color missing at both $y$ and $z$.

\begin{fact}\label{fact:sha}
    Let $\phi$ be a partial coloring and $e \in E_G(x,y)$ be an uncolored edge. 
    Let $F = (e, f)$ be a fan such that $M(\phi,x) \cap M(\phi, y) = M(\phi,x) \cap M(\phi, z) = \0$, where $z$ is the vertex in $V(f)$ distinct from $x$.
    Then, $M(\phi,y) \cap M(\phi, z) \neq \0$.
\end{fact}

\begin{proof}
    Note by the definition of a chain $|V(e) \cap V(f)| = 1$ and so we must have $y \neq z$.
    As $\phi(e) = \blank$, we have
    \[|M(\phi, x)| \geq \Sha - (\Delta - 1) = \lfloor\Delta/2\rfloor + 1.\]
    The same holds for $|M(\phi, y)|$.
    Suppose all three sets are disjoint.
    Then, we have
    \[|M(\phi, x) \cup M(\phi, y) \cup M(\phi, z)| = |M(\phi, x)| + |M(\phi, y)| + |M(\phi, z)| \geq 3\lfloor\Delta/2\rfloor + 2 > \Sha,\]
    a contradiction.
    So $M(\phi,y) \cap M(\phi, z)$ must be nonempty, as desired.
\end{proof}

\begin{algorithm}[h]\small
\caption{First Shannon Fan}\label{alg:first_shannon_fan}
\begin{flushleft}
\textbf{Input}: A proper partial $\Sha$-edge-coloring $\phi$, an uncolored edge $e \in E(G)$, and a vertex $x \in V(e)$. \\
\textbf{Output}: A fan $F$ with $\Start(F) = e$ and $\Pivot(F) = x$, and a pair of colors $\alpha, \beta \in [\Sha]$ such that $\alpha \in M(\phi, x)$, and either $\beta \in M(\phi, \vend(F))\cap M(\phi, \Pivot(F))$ or $\beta \in M(\phi, \vend(F))\cap M(\phi, \vstart(F))$.
\end{flushleft}
\begin{algorithmic}[1]
    \State Let $y \in V(e)$ be distinct from $x$.
    \If{$M(\phi, x) \cap M(\phi, y) \neq \0$}
        \State $\beta \gets \min M(\phi, x) \cap M(\phi, y)$
        \State \Return $((e), \beta, \beta)$ \label{step:first_happy_edge}
    \EndIf
    \State $\eta \gets \min M(\phi, y)$
    \State Let $f \in E_G(x)$ be such that $\phi(f) = \eta$ and let $z \in V(f)$ be distinct from $x$.
    \If{$M(\phi, x) \cap M(\phi, z) \neq \0$}
        \State $\beta \gets \min M(\phi, x) \cap M(\phi, z)$
        \State \Return $((e, f), \beta, \beta)$ \label{step:first_happy_fan_length_2}
    \EndIf
    \State $\beta \gets \min M(\phi, y) \cap M(\phi, z)$
    \State $\alpha \gets \min M(\phi, x)$
    \State \Return $((e, f), \alpha, \beta)$ \label{step:first_vstart_vend}
\end{algorithmic}
\end{algorithm}

The following lemma proves certain properties of the output of Algorithm~\ref{alg:first_shannon_fan}.

\begin{lemma}\label{lemma:first_fan_shannon}
    Let $(F, \alpha, \beta)$ be the output of Algorithm~\ref{alg:first_shannon_fan} on input $(\phi, e, x)$. 
    Then no edge in $F$ is colored $\alpha$ or $\beta$ and we have
    \begin{itemize}
        \item either $\beta \in M(\phi, x)$ and $F$ is $\phi$-happy, or
        \item $\length(F) = 2$ and $F$ is $(\phi, \alpha\beta)$-successful, or
        \item $e$ is $(\phi, \alpha\beta)$-successful.
    \end{itemize}
\end{lemma}

\begin{figure}[t]
    \centering
    \begin{subfigure}[t]{0.45\textwidth}
        \centering
    	\begin{tikzpicture}
    	    \node[circle,fill=black,draw,inner sep=0pt,minimum size=4pt] (a) at (0,0) {};
    		\node[circle,fill=black,draw,inner sep=0pt,minimum size=4pt] (b) at (1,0) {};
    		\draw[thick, dotted] (a) -- (b);  

                \node[anchor=north] at (b) {$x$};
                \node[anchor=north] at (a) {$y$};
    		
    	\end{tikzpicture}
    	\caption{$F = (e)$ is $\phi$-happy (step~\ref{step:first_happy_edge}).}
    \end{subfigure}
    \begin{subfigure}[t]{0.45\textwidth}
        \centering
    	\begin{tikzpicture}
    	    \node[circle,fill=black,draw,inner sep=0pt,minimum size=4pt] (a) at (0,0) {};
    		\node[circle,fill=black,draw,inner sep=0pt,minimum size=4pt] (b) at (1,0) {};
                \path (b) ++(75:1) node[circle,fill=black,draw,inner sep=0pt,minimum size=4pt] (c) {};
    		\draw[thick, dotted] (a) -- (b);
                \draw[thick] (b) -- (c);

                \node[anchor=north] at (b) {$x$};
                \node[anchor=north] at (a) {$y$};
                \node[anchor=south] at (c) {$z$};
    		
    	\end{tikzpicture}
    	\caption{$F = (e,f)$ is $\phi$-happy (step~\ref{step:first_happy_fan_length_2}).}
    \end{subfigure}
    
    \begin{subfigure}[t]{0.45\textwidth}
        \centering
    	\begin{tikzpicture}
                \node[circle,fill=black,draw,inner sep=0pt,minimum size=4pt] (a) at (0,0) {};
    		\node[circle,fill=black,draw,inner sep=0pt,minimum size=4pt] (b) at (1,0) {};
                \path (b) ++(75:1) node[circle,fill=black,draw,inner sep=0pt,minimum size=4pt] (c) {};
                \path (c) ++(30:1) node[circle,fill=black,draw,inner sep=0pt,minimum size=4pt] (d) {};
                \path (d) ++(0:1) node[circle,fill=black,draw,inner sep=0pt,minimum size=4pt] (e) {};
                \path (e) ++(0:2) node[circle,fill=black,draw,inner sep=0pt,minimum size=4pt] (f) {};
                \path (f) ++(0:1) node[circle,fill=black,draw,inner sep=0pt,minimum size=4pt] (g) {};
    		\draw[thick, dotted] (a) -- (b);
                \draw[thick, decorate,decoration=zigzag] (e) -- (f);
                \draw[thick] (b) -- (c) to node[font=\fontsize{8}{8},midway,inner sep=1pt,outer sep=1pt,minimum size=4pt,fill=white] {$\alpha$} (d) to node[font=\fontsize{8}{8},midway,inner sep=1pt,outer sep=1pt,minimum size=4pt,fill=white] {$\beta$} (e) (f) to node[font=\fontsize{8}{8},midway,inner sep=1pt,outer sep=1pt,minimum size=4pt,fill=white] {$\beta$} (g);

                \node[anchor=north] at (b) {$x$};
                \node[anchor=north] at (a) {$y$};
                \node[anchor=south] at (c) {$z$};
    		
    	\end{tikzpicture}
    	\caption{$F = (e,f)$ is $(\phi, \alpha\beta)$-successful (step~\ref{step:first_vstart_vend}).}
    \end{subfigure}
    \begin{subfigure}[t]{0.45\textwidth}
        \centering
    	\begin{tikzpicture}
                \clip (-1, 2) rectangle (4.5, -1.5);
                \node[circle,fill=black,draw,inner sep=0pt,minimum size=4pt] (a) at (0,0) {};
    		\node[circle,fill=black,draw,inner sep=0pt,minimum size=4pt] (b) at (1,0) {};
                \path (b) ++(75:1) node[circle,fill=black,draw,inner sep=0pt,minimum size=4pt] (c) {};
                \path (c) ++(30:1) node[circle,fill=black,draw,inner sep=0pt,minimum size=4pt] (d) {};
                \path (d) ++(0:1) node[circle,fill=black,draw,inner sep=0pt,minimum size=4pt] (e) {};
                \path (b) ++(-30:1) node[circle,fill=black,draw,inner sep=0pt,minimum size=4pt] (f) {};

                \path (a) ++(-135:1) node[circle,fill=black,draw,inner sep=0pt,minimum size=4pt] (g) {};
                \path (g) ++(-45:1) node[circle,fill=black,draw,inner sep=0pt,minimum size=4pt] (h) {};
                \path (h) ++(0:4) node[circle,fill=black,draw,inner sep=0pt,minimum size=4pt] (i) {};

                \draw[thick, dotted] (a) -- (b);
                \draw[thick, decorate,decoration=zigzag] (e) to[out=0, in=-15, looseness=2] (f) (h) -- (i);
                \draw[thick] (b) -- (c) to node[font=\fontsize{8}{8},midway,inner sep=1pt,outer sep=1pt,minimum size=4pt,fill=white] {$\alpha$} (d) to node[font=\fontsize{8}{8},midway,inner sep=1pt,outer sep=1pt,minimum size=4pt,fill=white] {$\beta$} (e) (b) to node[font=\fontsize{8}{8},midway,inner sep=1pt,outer sep=1pt,minimum size=4pt,fill=white] {$\beta$} (f) (a) to node[font=\fontsize{8}{8},midway,inner sep=1pt,outer sep=1pt,minimum size=4pt,fill=white] {$\alpha$} (g) to node[font=\fontsize{8}{8},midway,inner sep=1pt,outer sep=1pt,minimum size=4pt,fill=white] {$\beta$} (h);    

                \node[anchor=north] at (b) {$x$};
                \node[anchor=north] at (a) {$y$};
                \node[anchor=south] at (c) {$z$};
    		
    	\end{tikzpicture}
    	\caption{$F = (e, f)$ is $(\phi, \alpha\beta)$-disappointed and $e$ is $(\phi, \alpha\beta)$-successful (step~\ref{step:first_vstart_vend}).}
    \end{subfigure}
    \caption{Possible outputs $F$ of Algorithm~\ref{alg:first_shannon_fan}.}
    \label{fig:first_fan_cases}
\end{figure}
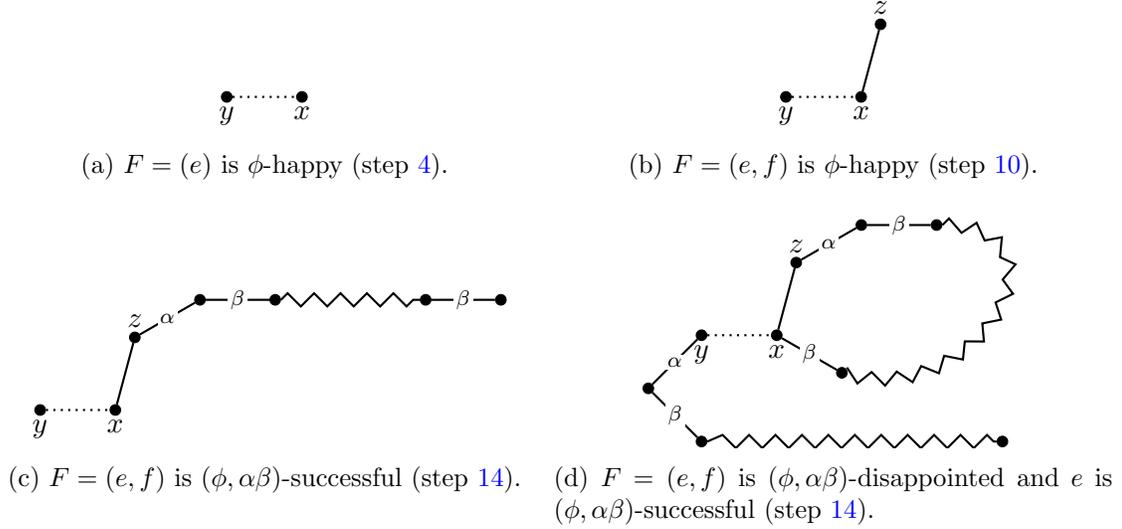

\begin{proof}
    Fig.~\ref{fig:first_fan_cases} describes each possible output.
    Clearly, no edge in $F$ is colored $\alpha$ or $\beta$.
    If $\beta \in M(\phi, x)$, then we are either at step~\ref{step:first_happy_edge} or~\ref{step:first_happy_fan_length_2}. 
    In either case, the fan $F$ is $\phi$-happy.
    Let $F = (e, f)$ and $y, z \in V(G)$ such that $e \in E_G(x,y)$ and $f \in E_G(x,z)$.

    Now, we must have $\beta \in M(\phi, y) \cap M(\phi, z)$ and $\alpha \in M(\phi, x)$.
    As we reach step~\ref{step:first_vstart_vend}, $\alpha \notin M(\phi, y)$ and $ \alpha \notin M(\phi, z)$.
    In particular, we have $\deg(v; \phi, \alpha\beta) < 2$ for $v \in \set{x,y,z}$.
    According to Definition~\ref{defn:hsf}, $F$ is a $(\phi, \alpha\beta)$-hopeful fan of length $2$.
    Suppose it is not successful, i.e., $x, z$ are $(\Shift(\phi, F), \alpha\beta)$-related.
    It remains to show that $x,y$ are not $(\phi, \alpha\beta)$-related.
    As $\phi(f) \notin \set{\alpha, \beta}$, $P(f;\Shift(\phi, F), \alpha\beta) = P(f;\phi, \alpha\beta)$.
    Let $P \defeq P(f;\phi, \alpha\beta),\, P' \defeq P(e; \phi, \beta\alpha)$ and define $Q,\,Q'$ by removing the first edges on $P,\,P'$ respectively.
    Note that $Q,\,Q'$ are path components of $G(\phi, \alpha\beta)$.
    We claim that $Q' = Q^*$.
    As $\vend(Q) = x$ and $\alpha \in M(\phi, x)$, we must have $\phi(\End(Q)) = \beta$.
    Similarly, by definition of $P'$, we must have $\phi(\Start(Q')) = \beta$.
    The claim now follows as $\phi$ is proper and $\beta \in M(\phi, y)$.
    In particular, we have shown that $\vend(Q') = z \neq y$.
    As $\beta \in M(\phi, y)$, if $x$ and $y$ are $(\phi, \alpha\beta)$-related, then $y = \vend(Q') = z$, a contradiction.
\end{proof}

By storing $M(\phi, z)$ as a hash map with key set $V$, we can retrieve $M(\phi, z)$ in amortized constant time.
Each \textsf{if} statement takes $O(\Delta)$ time and computing $\eta$ and $f$ also takes $O(\Delta)$ time.
It follows that Algorithm~\ref{alg:first_shannon_fan} runs in $O(\Delta)$ time.

The next algorithm we describe is the \hyperref[alg:next_shannon_fan]{Next Shannon Fan Algorithm}, Algorithm~\ref{alg:next_shannon_fan}. It will be used to construct the next fan on every iteration of our procedure for building multi-step Shannon chains;
see \S\ref{section:MSSA} for details. 
This algorithm is inspired by the paper \cite{GP} by Greb\'ik and Pikhurko. 
Algorithm~\ref{alg:next_shannon_fan} takes as input a proper partial edge-coloring $\phi$, an uncolored edge $e \in E(G)$, a vertex $x \in V(e)$, and a pair of colors $\alpha, \beta \in [\Sha]$ such that for $y \in V(e)$ distinct from $x$:
\begin{itemize}
    \item $\alpha \in M(\phi, x) \setminus M(\phi, y)$, and
    \item $\beta \in M(\phi, y)$.
\end{itemize}
Note that we must have $M(\phi, x) \setminus M(\phi, y) \neq \0$ .
The output is a tuple $(F, \gamma, \delta)$ such that:
\begin{itemize}
    \item $F$ is a fan with $\Start(F) = e$, $\Pivot(F) = x$ and $\length(F) \leq 2$,
    \item $\gamma \in [\Sha]$ is a color missing at $x$, and
    \item $\delta \in [\Sha]$ is a color such that either $\delta$ is missing at both $\vend(F)$ and $\Pivot(F)$, or $\delta$ is missing at both $\vend(F)$ and $\vstart(F)$.
\end{itemize}
The algorithm is similar to Algorithm~\ref{alg:first_shannon_fan} with a few differences including the restriction on the choice of $\eta$, the added case on the value of $\delta$, and the restriction on the choice of $\gamma$.

\begin{algorithm}[h]\small
\caption{Next Shannon Fan}\label{alg:next_shannon_fan}
\begin{flushleft}
\textbf{Input}: A proper partial $\Sha$-edge-coloring $\phi$, an uncolored edge $e \in E(G)$, a vertex $x \in V(e)$, and a pair of colors $\alpha, \beta \in [\Sha]$ such that $\alpha \in M(\phi, x) \setminus M(\phi, y)$ and $\beta \in M(\phi, y)$ where $y \in V(e)$ is distinct from $x$. \\
\textbf{Output}: A fan $F$ with $\Start(F) = e$ and $\Pivot(F) = x$, and a pair of colors $\gamma, \delta \in [\Sha]$ such that $\gamma \in M(\phi, x)$, and either $\delta \in M(\phi, \vend(F))\cap M(\phi, \Pivot(F))$ or $\delta \in M(\phi, \vend(F))\cap M(\phi, \vstart(F))$.
\end{flushleft}
\begin{algorithmic}[1]
    \State Let $y \in V(e)$ be distinct from $x$.
    \If{$M(\phi, x) \cap M(\phi, y) \neq \0$}
        \State $\delta \gets \min M(\phi, x) \cap M(\phi, y)$
        \State \Return $((e), \delta, \delta)$ \label{step:next_happy_edge}
    \EndIf
    \State $\eta \gets \min M(\phi, y)\setminus \set{\beta}$
    \State Let $f \in E_G(x)$ be such that $\phi(f) = \eta$ and let $z \in V(f)$ be distinct from $x$.
    \If{$M(\phi, x) \cap M(\phi, z) \neq \0$}
        \State $\delta \gets \min M(\phi, x) \cap M(\phi, z)$
        \State \Return $((e, f), \delta, \delta)$ \label{step:next_happy_fan_length_2}
    \EndIf
    \State $\delta \gets \min M(\phi, y) \cap M(\phi, z)$
    \If{$\delta = \beta$}
        \State \Return $((e, f), \alpha, \beta)$
    \EndIf
    \State $\gamma \gets \min M(\phi, x) \setminus\set{\alpha}$ \label{step:gamma_neq_alpha}
    \State \Return $((e, f), \gamma, \delta)$
\end{algorithmic}
\end{algorithm}

Let us now prove certain properties about the output of Algorithm~\ref{alg:next_shannon_fan}.

\begin{lemma}\label{lemma:next_shannon_fan}
    Let $(F, \gamma, \delta)$ be the output of Algorithm~\ref{alg:next_shannon_fan} on input $(\phi, e, x, \alpha, \beta)$. 
    Then, no edge in $F$ is colored $\alpha,\,\beta,\,\gamma$ or $\delta$, and one of the following holds:
    \begin{itemize}
        \item either $\delta \in M(\phi, x)$ and $F$ is $\phi$-happy, or
        \item $\length(F) = 2$, $\delta = \beta$ and $F$ is $(\phi, \alpha\beta)$-hopeful, or
        \item $\length(F) = 2$, $\delta \neq \beta$ and $F$ is $(\phi, \gamma\delta)$-successful, or 
        \item $\delta \neq \beta$ and $e$ is $(\phi, \gamma\delta)$-successful.
    \end{itemize}
\end{lemma}

\begin{proof}
    We note that as we assume $\Delta \geq 2$ and since $\phi(e) = \blank$, we have $|M(\phi, x)|, |M(\phi, y)| \geq 2$.
    In particular, $\eta$ and $\gamma$ are well defined.
    As in the proof of Lemma~\ref{lemma:first_fan_shannon}, if we reach step~\ref{step:next_happy_edge} or~\ref{step:next_happy_fan_length_2}, then $F$ is $\phi$-happy.
    Let $F = (e, f)$ and $y, z \in V(G)$ such that $e \in E_G(x,y)$ and $f \in E_G(x,z)$.
    
    By our restriction on the value of $\eta$, $\phi(f) \neq \beta$ and so no edge in $F$ is colored $\alpha$ or $\beta$.
    Furthermore, it is clear that no edge in $F$ is colored $\gamma$ or $\delta$.
    We now consider two cases:
    \begin{enumerate}[label=\ep{\normalfont{}\textbf{Case\arabic*}},labelindent=0pt,leftmargin=*]
        \item $\delta = \beta$. As $\alpha \in M(\phi, x)$ and $\delta \in M(\phi, z)$, it follows that $F$ is a $(\phi, \alpha\beta)$-hopeful fan of length $2$.

        \item $\delta \neq \beta$. Following an identical argument to that in Lemma~\ref{lemma:first_fan_shannon}, we can show that either $F$ or $e$ is $(\phi, \gamma\delta)$-successful.
    \end{enumerate}
    This covers all cases and completes the proof.
\end{proof}

As for Algorithm~\ref{alg:first_shannon_fan}, one can easily bound the runtime of Algorithm~\ref{alg:next_shannon_fan} by $O(\Delta)$.
As an illustration, we can now use Lemma~\ref{lemma:first_fan_shannon} to prove \hyperref[theo:shannon_bound]{Shannon's Theorem}.

\begin{corollary}[to Lemma~\ref{lemma:first_fan_shannon}]\label{corl:shannon}
    For every proper partial $\Sha$-edge-coloring $\phi$ and every uncolored edge $e$, there is a connected $e$-augmenting subgraph $H \subseteq G$.
\end{corollary}

\begin{proof}
    Let $x, y \in V(G)$ such that $e \in E_G(x,y)$.
    It is enough to find a $\phi$-happy chain $C$ with $\Start(C) = e$ and set $H = G[C]$.
    To this end, let $(F, \alpha, \beta)$ be the output of Algorithm~\ref{alg:first_shannon_fan} on input $(\phi, e, x)$. 
    By Lemma~\ref{lemma:first_fan_shannon}, we have:
    \begin{itemize}
        \item either $\beta \in M(\phi, x)$ and $F$ is $\phi$-happy, or
        \item $\length(F) = 2$ and $F$ is $(\phi, \alpha\beta)$-successful, or
        \item $e$ is $(\phi, \alpha\beta)$-successful.
    \end{itemize}
    In the first case, we let $C = F$, in the second we let $C = F + P(\End(F); \Shift(\phi, F), \alpha\beta)$ and in the last case, we let $C = P(e; \phi, \alpha\beta)$.
\end{proof}

%% file: shannon_deterministic.tex
\section{Proof of Theorem~\ref{theo:shannon_deterministic}}\label{section:shannon_det}

In this section, we will prove Theorem~\ref{theo:shannon_deterministic}.
We will describe an algorithm that takes as input a multigraph $G$ of maximum degree $\Delta$ and returns a proper $\Sha$-edge-coloring of $G$.
As mentioned in \S\ref{subsection:overview}, we will describe a procedure to construct and augment Shannon chains for a large subset of the uncolored edges.
Our algorithm will iteratively apply this procedure until all edges in $G$ are colored.
This subset of chains will satisfy a special property we will now define.

\begin{definition}[Disjoint Shannon Chains]\label{defn:disjoint}
    Let $C_1 = F_1+P_1$ and $C_2 = F_2+P_2$ be distinct Shannon chains.
    We say $C_1$ and $C_2$ are \emphd{disjoint} if:
    \begin{enumerate}[label = \ep{\textbf{Disjoint\arabic*}}, wide]
        \item\label{disjoint:edge} $E(C_1) \cap E(C_2) = \0$, and
        \item\label{disjoint:vertex} $(V(F_1) \cup \set{\vend(P_1)}) \cap (V(F_2) \cup \set{\vend(P_2)}) = \0$.
    \end{enumerate}
\end{definition}

Before we describe how we take advantage of disjoint chains, we provide an overview of the \hyperref[alg:shannon_chain]{Shannon Chain Algorithm}, which is formally presented as Algorithm~\ref{alg:shannon_chain}.
It takes as input a proper partial edge-coloring $\phi$, an uncolored edge $e \in E(G)$ and a pivot vertex $x\in V(e)$, and outputs a fan $F$ and a path $P$ such that $\Start(F) = e$, $\Pivot(F) = x$ and $F+P$ is a $\phi$-happy Shannon chain.
We start by applying the \hyperref[alg:first_shannon_fan]{First Shannon Fan Algorithm} (Algorithm~\ref{alg:first_shannon_fan}) as a subprocedure. Let $(F, \alpha, \beta)$ be its output.
If $\beta \in M(\phi, x)$, we return the $\phi$-happy fan $F$.
If not, by Lemma~\ref{lemma:first_fan_shannon}, either $F$ or $e$ is $(\phi, \alpha\beta)$-successful.
Let 
\[P \,\defeq\, P(\End(F); \Shift(\phi, F), \alpha\beta) \quad \text{and} \quad P' \,\defeq\, P(e; \phi, \alpha\beta)\] 
(this notation is defined in \S\ref{subsec:pathchains}). 
If $\vend(P) \neq x$, we return $F, P$; otherwise, we return $(e),P'$.

\begin{algorithm}[h]\small
\caption{Shannon Chain}\label{alg:shannon_chain}
\begin{flushleft}
\textbf{Input}: A proper partial $\Sha$-edge-coloring $\phi$, an uncolored edge $e \in E(G)$, and a vertex $x \in V(e)$. \\
\textbf{Output}: A fan $F$ with $\Start(F) = e$ and $\Pivot(F) = x$, and a path $P$ with $\Start(P) = \End(F)$ and $\vstart(P) = \Pivot(F) = x$ such that $F+P$ is a $\phi$-happy Shannon chain. 
\end{flushleft}
\begin{algorithmic}[1]
    \State $(F, \alpha, \beta) \gets \hyperref[alg:first_shannon_fan]{\mathsf{FirstFan}}(\phi, e, x)$ \Comment{Algorithm~\ref{alg:first_shannon_fan}} \label{step:first_fan}
    \If{$\beta \in M(\phi, x)$}
        \State \Return $F$, $(\End(F))$ \label{step:happy_fan}
    \EndIf
    \State $P \gets P(\End(F);\, \Shift(\phi, F),\, \alpha\beta), \quad P' \gets P(e; \, \phi,\, \alpha\beta)$
    \If{$\vend(P) \neq x$}
        \State \Return $F$, $P$
    \Else
        \State \Return $(e)$, $P'$
    \EndIf
\end{algorithmic}
\end{algorithm}

The correctness of Algorithm~\ref{alg:shannon_chain} follows from Lemma~\ref{lemma:first_fan_shannon}.
As we can compute $P,\,P'$ in time $O(\Delta\length(P))$ and $O(\Delta \length (P'))$ respectively, and since the runtime of Algorithm~\ref{alg:first_shannon_fan} is $O(\Delta)$, it follows that the runtime of Algorithm~\ref{alg:shannon_chain} is
\begin{align}\label{eqn:runtime_shannon_chain}
    O(\Delta + \Delta\length(P) + \Delta\length(P')).
\end{align}

Let $(F, P)$ be the output of Algorithm~\ref{alg:shannon_chain} on input $(\phi, e, x)$ such that $(F', \alpha,\beta)$ was the output of Algorithm~\ref{alg:first_shannon_fan} at step~\ref{step:first_fan}.
When augmenting $\phi$ with $F+P$, we color the $\Shift(\phi, F+P)$-happy edge $\End(P)$ either $\alpha$ or $\beta$ (whichever is valid).
Let us now describe the importance of disjoint chains with respect to Algorithm~\ref{alg:shannon_chain}.

\begin{lemma}\label{lemma:disjoint_chains}
    Let $\phi$ be a proper partial $\Sha$-edge-coloring and let $C_1 = F_1 + P_1$, $C_2 = F_2 + P_2$ be disjoint $\phi$-happy Shannon chains constructed using Algorithm~\ref{alg:shannon_chain}. 
    Then $C_1$ is $\aug(\phi, C_2)$-happy and vice-versa.
\end{lemma}

\begin{proof}
    By symmetry, it is enough to show that $C_1$ is $\aug(\phi, C_2)$-happy.
    Let $\psi \defeq \aug(\phi, C_2)$ and let $\alpha, \beta, \gamma, \delta \in [\Sha]$ be such that:
    \[P_1 = P(\End(F_1); \Shift(\phi, F_1), \alpha\beta), \quad P_2 = P(\End(F_2); \Shift(\phi, F_2), \gamma\delta).\]
    We claim the following:
    \begin{enumerate}
        \item\label{fan_shiftable} $F_1$ is $\psi$-shiftable, and
        \item\label{path_equivalence} $P_1 = P(\End(F_1); \Shift(\psi, F_1), \alpha\beta)$.
    \end{enumerate}
    Clearly if these conditions hold, then $C_1$ is $\psi$-happy.

    Let us first show \ref{fan_shiftable}. 
    We claim that $M(\psi, z) = M(\phi, z)$ for each $z \in V(F_1)$.
    If this holds, then \ref{fan_shiftable} follows. 
    Let $z \in V(F_1)$ such that $M(\psi, z) \neq M(\phi, z)$. 
    Then, $z \in V(C_1)$.
    By \ref{disjoint:vertex}, it must be the case that 
    \begin{itemize}
        \item either $z\in \IV(P_2) \setminus V(F_2)$, or 
        \item $z \in V(\End(P_2))$ and $z \neq \vend(P_2)$.
    \end{itemize}
    In either case, the only changes on the neighborhood of $z$ are that the edges colored $\gamma$ and $\delta$ swap values, i.e., the missing set remains unchanged, as desired.

    Now, let us show \ref{path_equivalence}.
    If $P_1 = (\End(F_1))$, we are done.
    Let assume $\length(P_1) \geq 2$.
    Note that by \ref{disjoint:edge}, the edges on $C_1$ are colored the same under $\phi$ and $\psi$.
    Furthermore, an identical argument to the previous paragraph shows that $M(\psi, \vend(P_1)) = M(\phi, \vend(P_1))$. 
    In particular, we have
    \[\deg(\vend(F_1); \psi, \alpha\beta) < 2, \quad \deg(\vend(P_1); \psi, \alpha\beta) < 2.\]
    As no edge in $F_1$ is colored $\alpha$ or $\beta$ by Lemma~\ref{lemma:first_fan_shannon}, the above holds with $\psi$ replaced by $\Shift(\psi, F_1)$.
    It remains to show that $\vend(F_1)$ and $\vend(P_1)$ are $(\Shift(\psi, F_1), \alpha\beta)$-related.
    This follows by \ref{disjoint:edge} and since no edge in $F_1$ is colored $\alpha$ or $\beta$.
\end{proof}

The goal now is to augment a large set of pairwise disjoint chains at each iteration.
To this end, we will make a few definitions.
Let $\phi$ be a proper partial $\Sha$-edge-coloring of $G$ and let $U_\phi$ be the set of uncolored edges under $\phi$.
For the rest of this section, we will assume the vertex set $V(G)$ is ordered.
Let $F_e,\,P_e$ be the output after running Algorithm~\ref{alg:shannon_chain} on input $(\phi, e, \min V(e))$ for some $e \in U_\phi$.
Let $\alpha(e), \beta(e)$ be the colors on the path $P_e$ ($\alpha(e) = \beta(e)$ if we reach step~\ref{step:happy_fan}).
We define the following sets for $\alpha,\beta \in [\Sha]$:
\[\Gamma_{\alpha, \beta}(\phi) \defeq \set{e \in U_\phi\,:\, \set{\alpha(e), \beta(e)} = \set{\alpha, \beta}}.\]

Let us now describe how to augment a large fraction of edges in such a set.
The details are provided in Algorithm~\ref{alg:augment_large}, but we will first give an informal overview.
The algorithm takes as input a proper partial $\Sha$-edge-coloring $\phi$ and a set $S$, where $S = \Gamma_{\alpha, \beta}(\phi)$ for some $\alpha, \beta \in [\Sha]$.
First, we apply Algorithm~\ref{alg:shannon_chain} for each edge $e \in S$ to construct a collection of chains $C_e = F_e + P_e$.
As a result of Lemma~\ref{lemma:disjoint_chains}, we may augment any subset of disjoint chains from this collection.
In order to track intersections between chains, we define an array $\visited$ which satisfies the following:
\[\visited(v) = 1 \iff v \in V(F_e) \cup \set{\vend(P_e)}, \quad \text{for some } F_e + P_e \text{ augmented so far},\]
i.e., before augmenting a chain, $\visited$ ensures~\ref{disjoint:vertex} holds.
The following lemma proves that this is enough.

\begin{lemma}\label{lemma:disjoint_alpha_beta}
    Let $\phi$ be a proper partial $\Sha$-edge-coloring and let $C_1 = F_1+P_1$, $C_2 = F_2+P_2$ be Shannon chains constructed using Algorithm~\ref{alg:shannon_chain} such that 
    \begin{enumerate}
        \item $P_1,\,P_2$ are $\alpha\beta$-paths for some $\alpha,\beta \in [\Sha]$, and 
        \item\label{item:disjoint}~\ref{disjoint:vertex} holds.
    \end{enumerate}
    Then,~\ref{disjoint:edge} holds as well.
\end{lemma}

\begin{proof}
    Suppose~\ref{disjoint:edge} does not hold.
    Let $Q_1,\, Q_2$ be obtained from $P_1,\,P_2$ by removing the first edge on the respective path.
    As $P_1,\, P_2$ are maximal $\alpha\beta$-paths and $\phi$ is a proper coloring, either $Q_1 = Q_2$, $Q_1 = Q_2^*$ or $V(Q_1) \cap V(Q_2) = \0$.

    If $Q_1 = Q_2$ or $Q_1 = Q_2^*$, either $\vend(P_1) = \vend(P_2)$ or $\vend(P_1) = \vend(F_2)$, implying~\ref{disjoint:vertex} is violated, contradicting~\ref{item:disjoint}.

    Let us assume $V(Q_1) \cap V(Q_2) = \0$. 
    We note that no edge in $F_1$ or $F_2$ is colored $\alpha$ or $\beta$ by Lemma~\ref{lemma:first_fan_shannon} and so $E(F_1) \cap E(Q_2) = E(F_2) \cap E(Q_1) = \0$.
    At this point, we must either have $\Start(P_1) \in E(F_2)$, or $\Start(P_2) \in E(F_1)$.
    However, as $\Start(P_i) = \End(F_i)$ this implies $V(F_1) \cap V(F_2) \neq \0$, once again contradicting~\ref{item:disjoint}.

    This covers all cases and so~\ref{disjoint:edge} must hold.
\end{proof}

\begin{algorithm}[h]\small
\caption{Augment a large set of chains}\label{alg:augment_large}
\begin{flushleft}
\textbf{Input}: A proper partial $\Sha$-edge-coloring $\phi$, a set of edges $S$ where $S = \Gamma_{\alpha, \beta}(\phi)$ for some colors $\alpha, \beta \in [\Sha]$. \\
\textbf{Output}: A coloring $\psi$ that results from augmenting a subset of chains formed by edges in $S$. 
\end{flushleft}
\begin{algorithmic}[1]
    \State $F_e, P_e \gets \hyperref[alg:shannon_chain]{\mathsf{Shannon Chain}}(\phi, e, \min V(e))$ \textbf{for each} $e \in S$ \label{step:compute_chains} \Comment{Algorithm~\ref{alg:shannon_chain}}
    \State $\psi \gets \phi$, \quad $\visited(v) \gets 0$ \textbf{for each} $v \in V(G)$
    \For{$e\in S$}
        \If{$\visited(v) = 0$ \textbf{for each} $v \in V(F_e) \cup \set{\vend(P_e)}$} \Comment{Ensure chains are disjoint}
            \State $\psi \gets \aug(\psi, F_e + P_e)$ \label{step:augment_chain}
            \State $\visited(v) \gets 1$ \textbf{for each} $v \in V(F_e) \cup \set{\vend(P_e)}$
        \EndIf
    \EndFor
    \State \Return $\psi$
\end{algorithmic}
\end{algorithm}




The correctness of Algorithm~\ref{alg:augment_large} follows from Lemmas~\ref{lemma:first_fan_shannon}, \ref{lemma:disjoint_chains} and \ref{lemma:disjoint_alpha_beta}.
The following lemma shows that the algorithm colors a large fraction of edges in $S$.

\begin{lemma}\label{lemma:augment_size}
    Let $\psi$ be the output after running Algorithm~\ref{alg:augment_large} with input $(\phi, \Gamma_{\alpha, \beta}(\phi))$ for some $\alpha, \beta \in [\Sha]$.
    Then,
    \[|\dom(\psi)| \geq |\dom(\phi)| + \frac{|\Gamma_{\alpha, \beta}(\phi)|}{20\Delta^2}.\]
\end{lemma}

\begin{proof}
    Let $C_e = F_e + P_e$ for each $e \in \Gamma_{\alpha, \beta}(\phi)$ be the chains at step~\ref{step:compute_chains}.
    For a fixed $e \in \Gamma_{\alpha, \beta}(\phi)$, we claim there are at most $16\Delta^2$ edges $f \in \Gamma_{\alpha, \beta}(\phi)$ such that $C_e$, $C_f$  are \textbf{not} disjoint.
    Since $16\Delta^2 + 1 \leq 20\Delta^2$, this implies the desired result.
    As a result of Lemma~\ref{lemma:disjoint_alpha_beta}, we need only consider the case where~\ref{disjoint:vertex} is violated.
    
    There are at most $4$ choices for the intersection vertex on $C_e$.
    Let $v \in V(C_e)$ be the vertex of intersection.
    For each case, we will consider all possible values of $f = \Start(F_f)$.
    \begin{enumerate}[label=\ep{\normalfont{}\textbf{Case\arabic*}},labelindent=0pt,leftmargin=*]
        \item If $v = \Pivot(F_f)$, then $f \in E_G(v)$ and so there are at most $\Delta$ choices for $f$.
        \item Similarly, if $v = \vstart(F_f)$, there are at most $\Delta$ choices for $f$.
        \item\label{item:count_f} If $v = \vend(F_f)$, then either $f \in E_G(v)$, or $f\in E_G(u)$ for some $u \in N_G(v)$.
        In particular, there are at most $\Delta + \Delta^2$ possible choices for $f$.
        \item Finally, if $v = \vend(P_f)$, we can determine $\vend(F_f)$ by following the $\alpha\beta$-path from $v$ to it's other endpoint and apply the previous case.
    \end{enumerate}
    In total, there are at most 
    \[4(2\Delta + 2(\Delta + \Delta^2) \leq 16\Delta^2\]
    choices for $f$, where we use the fact that $\Delta \geq 2$.
\end{proof}

Let us now bound the runtime of Algorithm~\ref{alg:augment_large}.

\begin{lemma}\label{lemma:augment_runtime}
    Algorithm~\ref{alg:augment_large} runs in $O(\Delta^3m)$ time.
\end{lemma}

\begin{proof}
    Consider running Algorithm~\ref{alg:augment_large} on input $(\phi, S)$, where $S = \Gamma_{\alpha, \beta}(\phi)$ for some $\alpha,\beta \in [\Sha]$.
    Let $P_e,\, P_e'$ be the paths computed in the call to Algorithm~\ref{alg:shannon_chain} at step~\ref{step:compute_chains}.
    From \eqref{eqn:runtime_shannon_chain}, step~\ref{step:compute_chains} takes time
    \[O\left(\Delta |S| + \Delta \sum_{e \in S}\length(P_e) + \Delta \sum_{e \in S}\length(P_e')\right).\]
    We note that step~\ref{step:augment_chain} takes $O(\Delta\,\length(F_e + P_e))$ time (the factor of $\Delta$ comes from updating missing sets).
    Furthermore, by storing $\visited$ as a hash map with key set $V(G)$, all accesses to $\visited$ occur in amortized constant time.
    Since the set of chains augmented are edge-disjoint, we conclude the \textsf{for} loop runs in $O(\Delta\,m)$ time.

    It remains to bound $\sum_{e \in S}\length(P_e)$ and $\sum_{e \in S}\length(P_e')$. 
    We will just consider $\sum_{e \in S}\length(P_e)$ as the other case follows identically, \textit{mutatis mutandis}.
    Let $\mathcal{P}$ be the set of all path components of $G(\phi, \alpha\beta)$.
    For each path $P_e$, let $Q_e$ be obtained from $P_e$ by removing $\Start(P_e)$. 
    As $\Start(P_e) = \End(F_e)$ and no edge in $F_e$ is colored $\alpha$ or $\beta$ by Lemma~\ref{lemma:first_fan_shannon}, $Q_e \in \mathcal{P}$.
    We have
    \[\sum_{e \in S}\length(P_e) = \sum_{e \in S}(1 + \length(Q_e)) = |S| + \sum_{e \in S}\length(Q_e).\]
    Let us now consider the term on the right:
    \begin{align*}
        \sum_{e \in S}\length(Q_e) &= \sum_{e \in S}\sum_{Q \in \mathcal{P}}\length(Q)\bbone{Q = Q_e} \\
        &= \sum_{Q \in \mathcal{P}}\length(Q)\sum_{e \in S}\bbone{Q = Q_e}.
    \end{align*}
    For a fixed $Q$, let us bound the number of edges $e$ such that $Q_e = Q$.
    We note that one endpoint of $Q$ is $\vend(P_e)$ and the other is $\vend(F_e)$ and so there are $2$ choices for $\vend(F_e)$. An identical argument to \ref{item:count_f} in Lemma~\ref{lemma:augment_size} shows that there are at most $\Delta + \Delta^2 \leq 2\Delta^2$ choices for $e$ given $\vend(F_e)$.
    With this in hand, we have
    \[\sum_{e \in S}\length(P_e) \leq |S| + 4\Delta^2\sum_{Q \in \mathcal{P}}\length(Q) = O(\Delta^2 m),\]
    as the paths in $\mathcal{P}$ are edge disjoint and $|S| \leq m$.
    Therefore, the overall runtime is $O(\Delta^3\,m)$ as claimed.
\end{proof}

We are now ready to describe our deterministic $\Sha$-edge-coloring algorithm. 
The details are in Algorithm~\ref{alg:shannon_deterministic}, but we first provide an informal overview.
We start with the blank coloring $\phi_0$ ($\phi_0(e) = \blank$ for all $e \in E(G)$) and begin iterating.
At each iteration, we first compute the sets  $\Gamma_{\alpha, \beta}(\phi_{i-1})$ and let $S_i$ be the largest such set.
We then apply Algorithm~\ref{alg:augment_large} on $(\phi_{i-1}, S_i)$ to obtain the coloring $\phi_i$.
Let $U_i$ be the set of uncolored edges under $\phi_i$.
We iterate until $U_i$ is empty.

\begin{algorithm}[h]\small
\caption{Deterministic Sequential Coloring with Shannon Chains}\label{alg:shannon_deterministic}
\begin{flushleft}
\textbf{Input}: A multigraph $G = (V, E)$ of maximum degree $\Delta$. \\
\textbf{Output}: A proper $\Sha$-edge-coloring $\phi$ of $G$.
\end{flushleft}
\begin{algorithmic}[1]
    \State $U \gets E$, \quad $\phi(e) \gets \blank$ \textbf{for each} $e \in U$
    \While{$U \neq \0$}
        \State Compute the sets $\Gamma_{\alpha, \beta}(\phi)$ \textbf{for each} $\alpha, \beta \in [\Sha]$  \label{step:compute_Gamma}
        \State $S \gets \Gamma_{\gamma, \delta}(\phi)$, where $(\gamma, \delta)$ maximizes $|\Gamma_{\alpha, \beta}(\phi)|$.
        \State $\phi \gets \hyperref[alg:augment_large]{\mathsf{AugmentChains}}(\phi, S)$ \Comment{Algorithm~\ref{alg:augment_large}}
        \State $U \gets \set{e\in U\,:\, \phi(e) = \blank}$
    \EndWhile
    \State \Return $\phi$
\end{algorithmic}
\end{algorithm}

Note that $\alpha(e),\,\beta(e)$ are the colors output by Algorithm~\ref{alg:first_shannon_fan}.
In particular, step~\ref{step:compute_Gamma} in Algorithm~\ref{alg:shannon_deterministic} can be implemented in $O(\Delta |U_{i-1}|)$ time.
By this observation and Lemma~\ref{lemma:augment_runtime}, each iteration of the \textsf{while} loop in Algorithm~\ref{alg:shannon_deterministic} takes time $O(\Delta^3 m)$.
It remains to bound the number of iterations.
Note the following:
\[|S_i| \geq \frac{1}{\Sha^2}\sum_{\alpha, \beta \in [\Sha]}|\Gamma_{\alpha, \beta}(\phi_{i-1})| \geq \frac{4|U_{i-1}|}{9\Delta^2}.\]
From Lemma~\ref{lemma:augment_size}, it follows that 
\[|U_T| \leq \left(1 - \frac{1}{45\Delta^4}\right)|U_{T-1}| \leq m\exp\left(-\frac{T}{45\Delta^4}\right).\]
In particular, the entire multigraph is colored after $T = O(\Delta^4 \log n)$ iterations.
It follows that Algorithm~\ref{alg:shannon_deterministic} computes a proper $\Sha$-edge-coloring in $O(\Delta^8n\log n)$ time.
This completes the proof of Theorem~\ref{theo:shannon_deterministic}.

%% file: MSSA.tex
\section{The Multi-Step Shannon Algorithm}\label{section:MSSA}

In this section, we will describe the coloring procedure that will be used to prove Theorem~\ref{theo:shannon}.
We will split this into three subsections.
First, we will describe the \hyperref[alg:multi_shannon_chain]{Multi-Step Shannon Algorithm} (MSSA), which constructs a $\phi$-happy multi-step Shannon chain $C$ given an uncolored edge $e$ and a vertex $x \in V(e)$.
In the second subsection, we will prove the correctness of the \hyperref[alg:multi_shannon_chain]{MSSA} along with certain properties of the chain computed.
In the third subsection, we will analyse the runtime of the \hyperref[alg:multi_shannon_chain]{MSSA}.
We will bound the runtime and $\length(C)$ in terms of $n$, $\Delta$ and a parameter $\ell$.
The value $\ell$ will be specified later; for now, we shall take it to be a polynomial in $\Delta$ of sufficiently large degree.

\subsection{Algorithm Overview}\label{subsec:alg_overview}

Our goal is to build a multi-step Shannon chain $C$. 
We will split the algorithm into subprocedures to distinguish between building the first Shannon chain and subsequent Shannon chains on $C$.
Let us first describe the \hyperref[alg:first_shannon_chain]{First Chain Algorithm}, which is nearly identical to Algorithm~\ref{alg:shannon_chain} with the only difference that we do not compute the entire paths $P$, $P'$. 
It takes as input a proper partial edge-coloring $\phi$, an uncolored edge $e \in E(G)$ and a pivot vertex $x\in V(e)$, and outputs a fan $F$ and a path $P$ such that $\Start(F) = e$, $\Pivot(F) = x$ and $F+P$ is a (not necessarily $\phi$-happy) Shannon chain.

\begin{algorithm}[h]\small
\caption{First Shannon Chain}\label{alg:first_shannon_chain}
\begin{flushleft}
\textbf{Input}: A proper partial $\Sha$-edge-coloring $\phi$, an uncolored edge $e\in E(G)$, and a vertex $x \in V(e)$. \\
\textbf{Output}: A fan $F$ with $\Start(F) = e$ and $\Pivot(F) = x$ and a path $P$ with $\Start(P) = \End(F)$ and $\vstart(P) = \Pivot(F) = x$. 
\end{flushleft}
\begin{algorithmic}[1]
    \State $(F, \alpha, \beta) \gets \hyperref[alg:first_shannon_fan]{\mathsf{FirstFan}}(\phi, e, x)$ \label{step:first_fan}\Comment{Algorithm~\ref{alg:first_shannon_fan}}
    \If{$\beta \in M(\phi, x)$}
        \State \Return $F$, $(\End(F))$ \label{step:first_happy_fan}
    \EndIf
    \State $P \gets P(\End(F);\, \Shift(\phi, F),\, \alpha\beta)$, \quad $P' \gets P(e;\, \phi,\, \alpha\beta)$ \label{step:first_path_def}
    \If{$\length(P) > 2\ell$ or $\vend(P) \neq x$}
        \State \Return $F$, $P|2\ell$ \label{step:first_length_2}
    \Else
        \State \Return $(e)$, $P'|2\ell$ \label{step:first_successful_path}
    \EndIf
\end{algorithmic}
\end{algorithm}

As mentioned earlier, we need only compute $P|2\ell$ and $P'|2\ell$ at step~\ref{step:first_path_def}.
As computing a path chain $Q$ takes time $O(\Delta\length(Q))$, it follows that the running time of Algorithm~\ref{alg:first_shannon_chain} is $O(\Delta+\Delta\,\ell) = O(\Delta\,\ell)$ as our choice for $\ell$ will be large.

Let us now describe the \hyperref[alg:next_shannon_chain]{Next Chain Algorithm}, which
is the subprocedure to build subsequent Shannon chains on $C$. 
It takes as input a proper partial edge-coloring $\phi$, an uncolored edge $e \in E(G)$, a vertex $x \in V(e)$, and a pair of colors $\alpha, \beta \in [\Sha]$ such that for $y \in V(e)$ distinct from $x$:
\begin{itemize}
    \item $\alpha \in M(\phi, x) \setminus M(\phi, y)$, and
    \item $\beta \in M(\phi, y)$.
\end{itemize}
The output is a fan $F$ and a path $P$ such that $\Start(F) = e$, $\Pivot(F) = x$ and $F+P$ is 
a (not necessarily $\phi$-happy) Shannon chain.  
The colors $\alpha$, $\beta$ represent the colors on the path of the previous Shannon chain on $C$ and the coloring $\phi$ refers to the shifted coloring with respect to $C$.
Following a nearly identical argument as for Algorithm~\ref{alg:first_shannon_chain}, we see that the runtime of Algorithm~\ref{alg:next_shannon_chain} is bounded by $O(\Delta\,\ell)$ as well.

\begin{algorithm}[h]\small
\caption{Next Shannon Chain}\label{alg:next_shannon_chain}
\begin{flushleft}
\textbf{Input}: A proper partial $\Sha$-edge-coloring $\phi$, an uncolored edge $e \in E(G)$, a vertex $x \in V(e)$, and a pair of colors $\alpha \in M(\phi, x) \setminus M(\phi, y)$ and $\beta \in M(\phi, y)$ where $y \in V(e)$ is distinct from $x$. \\
\textbf{Output}: A fan $F$ with $\Start(F) = e$ and $\Pivot(F) = x$ and a path $P$ with $\Start(P) = \End(F)$ and $\vstart(P) = \Pivot(F) = x$.
\end{flushleft}
\begin{algorithmic}[1]
    \State $(F, \gamma, \delta) \gets \hyperref[alg:next_shannon_fan]{\mathsf{NextFan}}(\phi, e, x, \alpha, \beta)$ \label{step:next_fan}\Comment{Algorithm~\ref{alg:next_shannon_fan}}
    \If{$\delta \in M(\phi, x)$}
        \State \Return $F$, $(\End(F))$ \label{step:next_happy_fan}
    \EndIf
    \If{$\delta  = \beta$}
        \State $P \gets P(\End(F);\, \Shift(\phi, F),\, \alpha\beta)$
        \State \Return $F$, $P|2\ell$ \label{step:delta_beta}
    \EndIf
    \State $P \gets P(\End(F);\, \Shift(\phi, F),\, \gamma\delta)$, \quad $P' \gets P(e;\, \phi,\, \gamma\delta)$
    \If{$\length(P) > 2\ell$ or $\vend(P) \neq x$}
        \State \Return $F$, $P|2\ell$ \label{step:next_length_2}
    \Else
        \State \Return $(e)$, $P'|2\ell$
    \EndIf
\end{algorithmic}
\end{algorithm}

Before we describe our \hyperref[alg:multi_shannon_chain]{Multi-Step Shannon Algorithm}, we define the non-intersecting property of a multi-step Shannon chain $C$, which is identical to that of multi-step Vizing Chains in \cite[Definition 5.1]{fastEdgeColoring}.
Recall the sets $\IE(P)$ and $\IV(P)$ of a path chain $P$, introduced in Definition~\ref{defn:internal}.

\begin{definition}[Non-intersecting chains]\label{defn:non-int}
    A $k$-step Shannon chain $C = F_0 + P_0 + \cdots + F_{k-1} + P_{k-1}$ is \emphd{non-intersecting} if for all $0\leq i < j < k$, $\IE(P_i) \cap E(F_j + P_j) = \0$ and $V(F_i) \cap V(F_j + P_j) = \0$.
\end{definition}

In our algorithm, we build a non-intersecting multi-step Shannon chain $C$.
We shall first provide an informal overview; the full details are provided in Algorithm~\ref{alg:multi_shannon_chain}.
We begin with a chain $C = (e)$ containing just the uncolored edge. Using Algorithm~\ref{alg:first_shannon_chain}, we find the first Shannon chain $F+P$. 
With this chain defined, we begin iterating.

At the start of each iteration, we have a non-intersecting chain $C = F_0 + P_0 + \cdots + F_{k-1} +P_{k-1}$ and a \emphd{candidate chain} $F+P$.
For example, at the start of the first iteration $C = (e)$ and $F + P$ corresponds to the chain computed using Algorithm~\ref{alg:first_shannon_chain}.
The chain and candidate chain satisfy the following properties for $\psi \defeq \Shift(\phi, C)$:
\begin{enumerate}[label=\ep{\normalfont{}\texttt{Inv}\arabic*},labelindent=15pt,leftmargin=*]
    \item\label{inv:start_F_end_C} $\Start(F) = \End(C)$ and $\vstart(F) = \vend(C)$,
    \item\label{inv:non_intersecting_shiftable} $C + F + P$ is non-intersecting and $\phi$-shiftable, and 
    \item\label{inv:hopeful_length} $F$ is either $\psi$-happy or $(\psi, \alpha\beta)$-hopeful, where $P$ is an $\alpha\beta$-path; furthermore, if $F$ is $(\psi, \alpha\beta)$-disappointed, then $\length(P) = 2\ell$. 
\end{enumerate}
We will prove that these invariants hold in the next subsection. For now, let us take them to be true. Fig.~\ref{fig:iteration_start} shows an example of $C$ (in black) and $F+P$ (in red) at the start of an iteration.

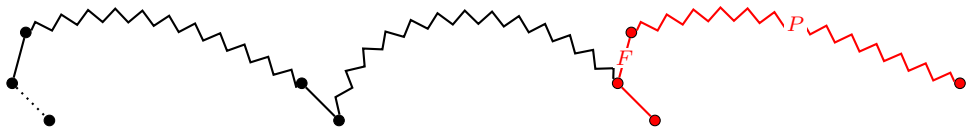
\begin{figure}[H]
    \centering
    \begin{tikzpicture}[xscale = 0.7,yscale=0.7]
        \node[circle,fill=black,draw,inner sep=0pt,minimum size=4pt] (a) at (0,0) {};
        	\path (a) ++(-45:1) node[circle,fill=black,draw,inner sep=0pt,minimum size=4pt] (b) {};
        	\path (a) ++(75:1) node[circle,fill=black,draw,inner sep=0pt,minimum size=4pt] (c) {};

                \path (a) ++(0:5.5) node[circle,fill=black,draw,inner sep=0pt,minimum size=4pt] (d) {};
            
        	\path (d) ++(-45:1) node[circle,fill=black,draw,inner sep=0pt,minimum size=4pt] (e) {};
        	
        	\path (d) ++(0:6) node[circle,fill=red,draw,inner sep=0pt,minimum size=4pt] (g) {};

                \path (g) ++(-45:1) node[circle,fill=red,draw,inner sep=0pt,minimum size=4pt] (h) {};
                \path (g) ++(75:1) node[circle,fill=red,draw,inner sep=0pt,minimum size=4pt] (i) {};
                \path (g) ++(0:6.5) node[circle,fill=red,draw,inner sep=0pt,minimum size=4pt] (j) {};
        	
        	\draw[thick,dotted] (a) -- (b);
        	
        	\draw[thick, decorate,decoration=zigzag] (c) to[out=20,in=150] (d) (e) to[out=85,in=130, looseness = 1.1]  (g);
            \draw[thick, decorate,decoration=zigzag, red](i) to[out=20,in=160] node[font=\fontsize{8}{8},midway,inner sep=1pt,outer sep=1pt,minimum size=4pt,fill=white] {$P$} (j);
        	
        	\draw[thick] (a) -- (c) (d) -- (e);
            \draw[thick, red]  (h) -- (g) to node[font=\fontsize{8}{8},midway,inner sep=1pt,outer sep=1pt,minimum size=4pt,fill=white] {$F$} (i);
    	
    \end{tikzpicture}
    \caption{The chain $C$ and candidate chain $F+P$ at the start of an iteration.}
    \label{fig:iteration_start}
\end{figure}

We first check whether $\length(P) < 2\ell$, in which case $F$ must be $(\psi, \alpha\beta)$-successful.
If so, we have found a $\phi$-happy multi-step Shannon chain and we return $C+F+P$.
If not, we let $F_k = F$ and consider an initial segment $P_k$ of $P$ of length $\ell' \in [\ell, 2\ell - 1]$ chosen uniformly at random. 
Using Algorithm~\ref{alg:next_shannon_chain} on the coloring $\Shift(\phi, C + F_k + P_k)$, we find a chain $\tilde F + \tilde P$ with $\Start(\tilde F) = \End(P_k)$ and $\vstart(\tilde{F}) = \vend(P_k)$.

At this point, we have two cases to consider. 
First, suppose $C + F_k + P_k + \tilde F + \tilde P$ is non-intersecting.
We then continue on, updating the chain to be $C+F_k+P_k$ and the candidate chain to be $\tilde F + \tilde P$. 
We call such an update a \textsf{forward} iteration.
Fig.~\ref{fig:non_intersecting_iteration} shows an example of a \textsf{forward} iteration with $\tilde F + \tilde P$ shown in blue.

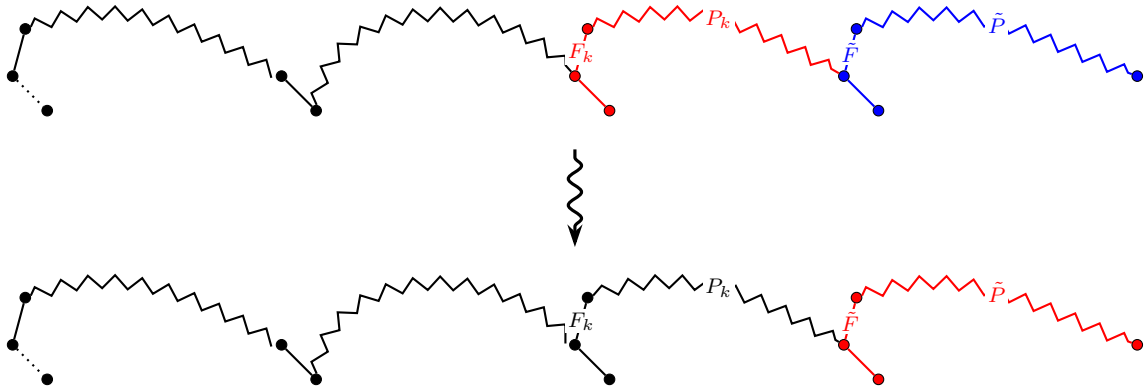
\begin{figure}[t]
    \centering
        \begin{tikzpicture}[xscale = 0.65,yscale=0.65]
            \node[circle,fill=black,draw,inner sep=0pt,minimum size=4pt] (a) at (0,0) {};
        	\path (a) ++(-45:1) node[circle,fill=black,draw,inner sep=0pt,minimum size=4pt] (b) {};
        	\path (a) ++(75:1) node[circle,fill=black,draw,inner sep=0pt,minimum size=4pt] (c) {};

                \path (a) ++(0:5.5) node[circle,fill=black,draw,inner sep=0pt,minimum size=4pt] (d) {};
            
        	\path (d) ++(-45:1) node[circle,fill=black,draw,inner sep=0pt,minimum size=4pt] (e) {};
        	
        	\path (d) ++(0:6) node[circle,fill=black,draw,inner sep=0pt,minimum size=4pt] (g) {};

                \path (g) ++(-45:1) node[circle,fill=black,draw,inner sep=0pt,minimum size=4pt] (h) {};
                \path (g) ++(75:1) node[circle,fill=black,draw,inner sep=0pt,minimum size=4pt] (i) {};
                \path (g) ++(0:5.5) node[circle,fill=red,draw,inner sep=0pt,minimum size=4pt] (j) {};

                \path (j) ++(-45:1) node[circle,fill=red,draw,inner sep=0pt,minimum size=4pt] (k) {};
                \path (j) ++(75:1) node[circle,fill=red,draw,inner sep=0pt,minimum size=4pt] (l) {};
                \path (j) ++(0:6) node[circle,fill=red,draw,inner sep=0pt,minimum size=4pt] (m) {};
        	
        	\draw[thick,dotted] (a) -- (b);
        	
        	\draw[thick, decorate,decoration=zigzag] (c) to[out=20,in=150] (d) (e) to[out=85,in=130, looseness = 1.1]  (g) (i) to[out=20,in=150] node[font=\fontsize{8}{8},midway,inner sep=1pt,outer sep=1pt,minimum size=4pt,fill=white] {$P_k$} (j);
            \draw[thick, decorate,decoration=zigzag, red](l) to[out=20,in=160] node[font=\fontsize{8}{8},midway,inner sep=1pt,outer sep=1pt,minimum size=4pt,fill=white] {$\tilde P$} (m);
        	
        	\draw[thick] (a) -- (c) (d) -- (e) (h) -- (g) to node[font=\fontsize{8}{8},midway,inner sep=1pt,outer sep=1pt,minimum size=4pt,fill=white] {$F_k$} (i);
            \draw[thick, red]  (k) -- (j) to node[font=\fontsize{8}{8},midway,inner sep=1pt,outer sep=1pt,minimum size=4pt,fill=white] {$\tilde F$} (l);

        \begin{scope}[yshift=5.5cm]
            \node[circle,fill=black,draw,inner sep=0pt,minimum size=4pt] (a) at (0,0) {};
        	\path (a) ++(-45:1) node[circle,fill=black,draw,inner sep=0pt,minimum size=4pt] (b) {};
        	\path (a) ++(75:1) node[circle,fill=black,draw,inner sep=0pt,minimum size=4pt] (c) {};

                \path (a) ++(0:5.5) node[circle,fill=black,draw,inner sep=0pt,minimum size=4pt] (d) {};
            
        	\path (d) ++(-45:1) node[circle,fill=black,draw,inner sep=0pt,minimum size=4pt] (e) {};
        	
        	\path (d) ++(0:6) node[circle,fill=red,draw,inner sep=0pt,minimum size=4pt] (g) {};

                \path (g) ++(-45:1) node[circle,fill=red,draw,inner sep=0pt,minimum size=4pt] (h) {};
                \path (g) ++(75:1) node[circle,fill=red,draw,inner sep=0pt,minimum size=4pt] (i) {};
                \path (g) ++(0:5.5) node[circle,fill=blue,draw,inner sep=0pt,minimum size=4pt] (j) {};

                \path (j) ++(-45:1) node[circle,fill=blue,draw,inner sep=0pt,minimum size=4pt] (k) {};
                \path (j) ++(75:1) node[circle,fill=blue,draw,inner sep=0pt,minimum size=4pt] (l) {};
                \path (j) ++(0:6) node[circle,fill=blue,draw,inner sep=0pt,minimum size=4pt] (m) {};
        	
        	\draw[thick,dotted] (a) -- (b);
        	
        	\draw[thick, decorate,decoration=zigzag] (c) to[out=20,in=150] (d) (e) to[out=85,in=130, looseness = 1.1]  (g);
            \draw[thick, decorate,decoration=zigzag, red](i) to[out=20,in=150] node[font=\fontsize{8}{8},midway,inner sep=1pt,outer sep=1pt,minimum size=4pt,fill=white] {$P_k$} (j);
            \draw[thick, decorate,decoration=zigzag, blue](l) to[out=20,in=160] node[font=\fontsize{8}{8},midway,inner sep=1pt,outer sep=1pt,minimum size=4pt,fill=white] {$\tilde P$} (m);
        	
        	\draw[thick] (a) -- (c) (d) -- (e);
            \draw[thick, red]  (h) -- (g) to node[font=\fontsize{8}{8},midway,inner sep=1pt,outer sep=1pt,minimum size=4pt,fill=white] {$F_k$} (i);
            \draw[thick, blue]  (k) -- (j) to node[font=\fontsize{8}{8},midway,inner sep=1pt,outer sep=1pt,minimum size=4pt,fill=white] {$\tilde F$} (l);
        \end{scope}

        \begin{scope}[yshift=3cm]
            \draw[-{Stealth[length=3mm,width=2mm]},very thick,decoration = {snake,pre length=3pt,post length=7pt,},decorate] (11.5,1) -- (11.5,-1);
        \end{scope}
        	
        \end{tikzpicture}
    \caption{Example of a \textsf{forward} iteration.}
    \label{fig:non_intersecting_iteration}
\end{figure}

Now suppose $\tilde F + \tilde P$ intersects $C + F_k + P_k$. The edges and vertices of $\tilde{F} + \tilde{P}$ are naturally ordered, and we let $0 \leq j \leq k$ be the index such that the first intersection point of $\tilde{F} + \tilde{P}$ with $C + F_k + P_k$ occurred at $F_j+P_j$. 
Then we update $C$ to $C' \defeq F_0 + P_0 + \cdots + F_{j-1} + P_{j-1}$ and $F+P$ to $F_j + P'$, where $P'$ is the path of length $2\ell$ from which $P_j$ was obtained as an initial segment.
For $r = k-j$, we call such an update an \textsf{$r$-backward} iteration.
Fig.~\ref{fig:intersecting_iteration} shows an example of an \textsf{$r$-backward} iteration.

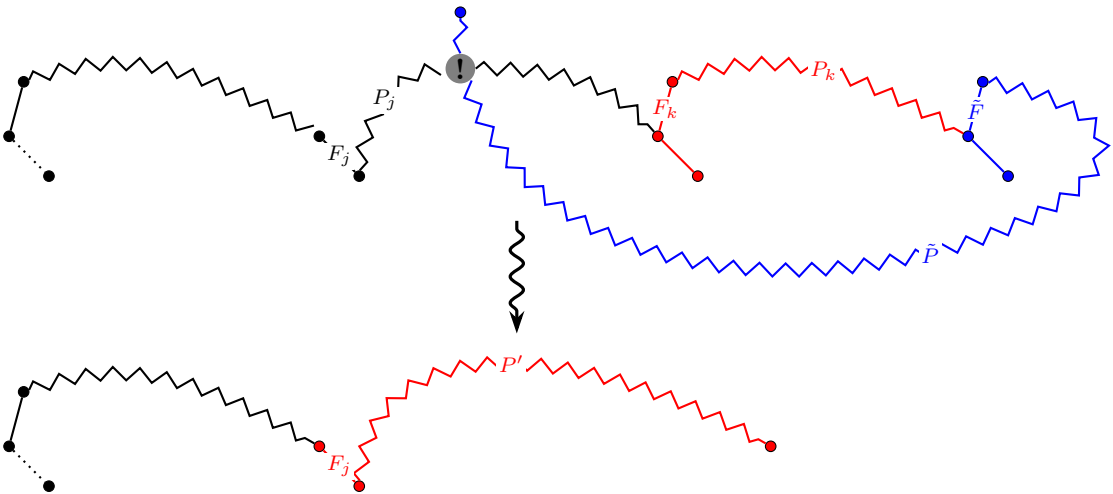
\begin{figure}[!b]
    \centering
        \begin{tikzpicture}[xscale = 0.75,yscale=0.75]
        \clip (-0.5, 8) rectangle (20, -1);
            \node[circle,fill=black,draw,inner sep=0pt,minimum size=4pt] (a) at (0,0) {};
        	\path (a) ++(-45:1) node[circle,fill=black,draw,inner sep=0pt,minimum size=4pt] (b) {};
        	\path (a) ++(75:1) node[circle,fill=black,draw,inner sep=0pt,minimum size=4pt] (c) {};

                \path (a) ++(0:5.5) node[circle,fill=red,draw,inner sep=0pt,minimum size=4pt] (d) {};
            
        	\path (d) ++(-45:1) node[circle,fill=red,draw,inner sep=0pt,minimum size=4pt] (e) {};
        	
        	\path (d) ++(0:8) node[circle,fill=red,draw,inner sep=0pt,minimum size=4pt] (g) {};

        	\draw[thick,dotted] (a) -- (b);
        	
        	\draw[thick, decorate,decoration=zigzag] (c) to[out=20,in=150] (d);
            \draw[thick, decorate,decoration=zigzag, red] (e) to[out=85,in=150, looseness=1.1] node[font=\fontsize{8}{8},midway,inner sep=1pt,outer sep=1pt,minimum size=4pt,fill=white] {$P'$} (g);
        	
        	\draw[thick] (a) -- (c);
            \draw[thick, red] (d) to node[font=\fontsize{8}{8},midway,inner sep=1pt,outer sep=1pt,minimum size=4pt,fill=white] {$F_j$} (e);

        \begin{scope}[yshift=5.5cm]
            \node[circle,fill=black,draw,inner sep=0pt,minimum size=4pt] (a) at (0,0) {};
        	\path (a) ++(-45:1) node[circle,fill=black,draw,inner sep=0pt,minimum size=4pt] (b) {};
        	\path (a) ++(75:1) node[circle,fill=black,draw,inner sep=0pt,minimum size=4pt] (c) {};

                \path (a) ++(0:5.5) node[circle,fill=black,draw,inner sep=0pt,minimum size=4pt] (d) {};
            
        	\path (d) ++(-45:1) node[circle,fill=black,draw,inner sep=0pt,minimum size=4pt] (e) {};

                \node[circle, fill=gray,inner sep=1pt] (z) at (8, 1.2) {\textbf{!}};
        	
        	\path (d) ++(0:6) node[circle,fill=red,draw,inner sep=0pt,minimum size=4pt] (g) {};

                \path (g) ++(-45:1) node[circle,fill=red,draw,inner sep=0pt,minimum size=4pt] (h) {};
                \path (g) ++(75:1) node[circle,fill=red,draw,inner sep=0pt,minimum size=4pt] (i) {};
                \path (g) ++(0:5.5) node[circle,fill=blue,draw,inner sep=0pt,minimum size=4pt] (j) {};

                \path (j) ++(-45:1) node[circle,fill=blue,draw,inner sep=0pt,minimum size=4pt] (k) {};
                \path (j) ++(75:1) node[circle,fill=blue,draw,inner sep=0pt,minimum size=4pt] (l) {};
                \path (z) ++(90:1) node[circle,fill=blue,draw,inner sep=0pt,minimum size=4pt] (m) {};
        	
        	\draw[thick,dotted] (a) -- (b);
        	
        	\draw[thick, decorate,decoration=zigzag] (c) to[out=20,in=150] (d) (e) to[out=85,in=180] node[font=\fontsize{8}{8},midway,inner sep=1pt,outer sep=1pt,minimum size=4pt,fill=white] {$P_j$} (z) to[out=0,in=130] (g);
            \draw[thick, decorate,decoration=zigzag, red] (i) to[out=20,in=150] node[font=\fontsize{8}{8},midway,inner sep=1pt,outer sep=1pt,minimum size=4pt,fill=white] {$P_k$} (j);
            \draw[thick, decorate,decoration=zigzag, blue](l) to[out=-10,in=-70, looseness=2] node[font=\fontsize{8}{8},midway,inner sep=1pt,outer sep=1pt,minimum size=4pt,fill=white] {$\tilde P$} (z) -- (m);
        	
        	\draw[thick] (a) -- (c) (d) to node[font=\fontsize{8}{8},midway,inner sep=1pt,outer sep=1pt,minimum size=4pt,fill=white] {$F_j$} (e);
            \draw[thick, red]  (h) -- (g) to node[font=\fontsize{8}{8},midway,inner sep=1pt,outer sep=1pt,minimum size=4pt,fill=white] {$F_k$} (i);
            \draw[thick, blue]  (k) -- (j) to node[font=\fontsize{8}{8},midway,inner sep=1pt,outer sep=1pt,minimum size=4pt,fill=white] {$\tilde F$} (l);
        \end{scope}

        \begin{scope}[yshift=3cm]
            \draw[-{Stealth[length=3mm,width=2mm]},very thick,decoration = {snake,pre length=3pt,post length=7pt,},decorate] (9,1) -- (9,-1);
        \end{scope}
        	
        \end{tikzpicture}
    \caption{Example of a \textsf{$1$-backward} iteration.}
    \label{fig:intersecting_iteration}
\end{figure}

In order to track intersections we define a hash map $\visited$ with key set $V(G)\cup E(G)$. 
Let $C = F_0 + P_0 + \cdots + F_{k-1} + P_{k-1}$ be the chain at the start of an iteration.
We maintain the following:
\[\visited(v) = \left\{\begin{array}{cc}
    1 & v \in V(F_i) \text{ for some } 0 \leq i < k; \\
    0 & \text{otherwise;}
\end{array}\right. \]
and
\[\visited(e) = \left\{\begin{array}{cc}
    1 & e \in \IE(P_i) \text{ for some } 0 \leq i < k; \\
    0 & \text{otherwise.}
\end{array}\right.\]
The formal statement of our \hyperref[alg:multi_shannon_chain]{Multi-Step Shannon Algorithm} is given in Algorithm~\ref{alg:multi_shannon_chain}.

\begin{algorithm}[h]\small
\caption{Multi-Step Shannon Algorithm (MSSA)}\label{alg:multi_shannon_chain}
\begin{flushleft}
\textbf{Input}: A proper partial $\Sha$-edge-coloring $\phi$, an uncolored edge $e$, and a vertex $x \in V(e)$. \\
\textbf{Output}: A $\phi$-happy multi-step Shannon chain $C$ with $\Start(C) = e$.
\end{flushleft}
\begin{algorithmic}[1]
    \State $\visited(f) \gets 0, \quad \visited(v) \gets 0$ \quad \textbf{for each} $f \in E(G)$, $v \in V(G)$
    \State $(F,P) \gets \hyperref[alg:first_shannon_chain]{\mathsf{FirstChain}}(\phi, e, x)$ \Comment{Algorithm~\ref{alg:first_shannon_chain}}
    \State $C\gets (e), \quad \psi \gets \phi, \quad k \gets 0$
    \medskip
    \While{true}
        \If{$\length(P) < 2\ell$}
            \State \Return $C+F+P$ \label{step:success}\Comment{Success}
        \EndIf
        \State Let $\ell' \in [\ell,2\ell-1]$ be chosen uniformly at random.
        \State $F_k \gets F,\quad P_k\gets P|\ell'$ \label{step:Pk}\Comment{Randomly shorten the path}
        \State Let $\alpha$, $\beta$ be such that $P_k$ is an $\alpha\beta$-path where $\psi(\End(P_k)) = \beta$.
        \State $\psi \gets \Shift(\psi, F_k+P_k)$ 
        \State $\visited(v) \gets 1$ \textbf{for each} $v \in V(F_k)$
        \State $\visited(f) \gets 1$ \textbf{for each} $f \in \IE(P_k)$
        \State $e' \gets \End(P_k), \quad w \gets \vend(P_k), \quad u \in V(e')$ distinct from $w$
        \State $(\tilde F , \tilde P) \gets \hyperref[alg:next_shannon_chain]{\mathsf{NextChain}}(\psi, e', u, \alpha, \beta)$ \label{step:alpha_beta_order}\Comment{Algorithm~\ref{alg:next_shannon_chain}}
        \If{$\visited(z) = 1$ or $\visited(h) = 1$ for some $z\in V(\tilde F + \tilde P)$, $h \in E(\tilde F + \tilde P)$}
            \State Let $0 \leq j \leq k$ be such that the first intersection occurs at $F_j + P_j$.
            \State $\psi \gets \Shift(\psi, (F_j + P_j + \cdots + F_k + P_k)^*)$ \label{step:psi}
            \State $\visited(v) \gets 0$ \textbf{for each} $v \in V(F_j) \cup \ldots \cup V(F_k)$
            \State $\visited(f) \gets 0$ \textbf{for each} $f \in \IE(P_j) \cup \ldots \cup \IE(P_k)$ \label{step:visited}
            \State $C\gets F_0 + P_0 + \cdots + F_{j-1} + P_{j-1}, \quad k \gets j$ \label{step:truncate} \Comment{\textsf{$(k-j)$-backward} iteration}
            \State $F\gets F_j, \quad P \gets P'$ \quad where $P_j$ is an initial segment of $P'$ as described earlier.
        \ElsIf{$2 \leq \length(\tilde P) < 2\ell$ \textbf{and} $\vend(\tilde P) = \Pivot(\tilde F)$}
            \State \Return \textbf{\textsf{FAIL}} \label{step:fail_chain}\Comment{Failure}
        \Else
            \State $ C \gets C + F_k + P_k, \quad F \gets \tilde F, \quad P \gets \tilde P, \quad k \gets k + 1$ \label{step:append}\Comment{\textsf{forward} iteration}
        \EndIf
    \EndWhile
\end{algorithmic}
\end{algorithm}

Note that in steps~\ref{step:psi}--\ref{step:visited} of the algorithm, we can update $\visited$, the missing colors hash map $M(\cdot)$ and $\psi$ simultaneously. 
By construction, $\length(P_k) \geq \ell > 2$ for all $k$. 
This ensures that at least one edge of each color $\alpha$, $\beta$ is on the $\alpha\beta$-path $P_k$ and also guarantees that $V(\End(F_j)) \cap V(\Start(F_{j+1})) = \0$ for all $j$.
As noted in the overview, in step~\ref{step:truncate} of the algorithm, we truncate the current chain at the first vertex $z$ or edge $h$ on $\tilde F + \tilde P$ such that $\visited(z) = 1$ or $\visited(h) = 1$.
These observations will be important for the proofs in subsequent sections.

\subsection{Proof of Correctness}\label{subsection:mssa_poc}

In this subsection, we prove the correctness of Algorithm~\ref{alg:multi_shannon_chain} as well as some auxiliary results on the chain it outputs.
These results will be important for the proofs in the remainder of the paper.
First, let us consider the output of Algorithm~\ref{alg:first_shannon_chain}.

\begin{lemma}\label{lemma:first_chain}
    Let $\phi$ be a proper partial coloring and let $e$ be an uncolored edge such that $V(e) = \set{x,y}$. 
    Let $F$ and $P$ be the fan and path returned by Algorithm~\ref{alg:first_shannon_chain} on input $(\phi, e, x)$, where $P$ is an $\alpha\beta$-path. Then no edge in $F$ is colored $\alpha$ or $\beta$, and
    \begin{itemize}
        \item either $F$ is $\phi$-happy and $P = (\End(F))$, or
        \item $F$ is $(\phi, \alpha\beta)$-hopeful, and $\length(P) = 2\ell$, or
        \item $F$ is $(\phi, \alpha\beta)$-successful.
    \end{itemize}
\end{lemma}

\begin{proof}
    Let $(F', \alpha, \beta)$ be the output of Algorithm~\ref{alg:first_shannon_fan} at step~\ref{step:first_fan}.
    By Lemma~\ref{lemma:first_fan_shannon}, we must have no edge in $F'$ is colored $\alpha$ or $\beta$, and
    \begin{itemize}
        \item either $\beta \in M(\phi, x)$ and $F'$ is $\phi$-happy, or
        \item $\length(F') = 2$ and $F'$ is $(\phi, \alpha\beta)$-successful, or
        \item $e$ is $(\phi, \alpha\beta)$-successful.
    \end{itemize}
    If the algorithm reaches step~\ref{step:first_happy_fan}, then $F = F'$ is $\phi$-happy and $P = (\End(F))$.
    If not, $\length(F') = 2$ and let 
    \[P' \defeq P(\End(F'); \Shift(\phi, F'), \alpha\beta), \quad P'' \defeq P(e; \phi, \alpha\beta).\]
    At step~\ref{step:first_length_2}, we first check whether $\length(P') > 2\ell$ or $F'$ is $(\phi, \alpha\beta)$-successful. 
    In this case, we return $F = F'$ and $P = P'|2\ell$ ($\length(P) = 2\ell$ if $F$ is $(\phi, \alpha\beta)$-disappointed). 
    If $\length(P') < 2\ell$ and $F'$ is $(\phi, \alpha\beta)$-disappointed, we pick the $(\phi, \alpha\beta)$-successful fan $F = (e)$ and the initial segment $P = P''|2\ell$.
\end{proof}

Note that in the setting of Lemma~\ref{lemma:first_chain}, the chain $F+P$ is $\phi$-shiftable.
Furthermore, if $\length(P) < 2\ell$, then $F+P$ is $\phi$-happy.
Let us now consider the output of Algorithm~\ref{alg:next_shannon_chain}.

\begin{lemma}\label{lemma:next_chain}
    Suppose that $\phi$ is a proper partial coloring, $e$ is an uncolored edge such that $V(e) = \set{x,y}$, and $\alpha$, $\beta$ are colors such that $\alpha \in M(\phi, x) \setminus M(\phi, y)$ and $\beta \in M(\phi, y)$. 
    Let $\tilde F$ and $\tilde P$ be the fan and the path returned by Algorithm~\ref{alg:next_shannon_chain} on input $(\phi, e, x, \alpha, \beta)$, where $\tilde P$ is a $\gamma\delta$-path for some $\gamma, \delta \in [\Sha]$. 
    Then no edge in $\tilde F$ is colored $\alpha,\,\beta,\,\gamma$ or $\delta$, and
    \begin{itemize}
        \item either $\tilde F$ is $\phi$-happy and $\tilde P = (\End(\tilde F))$, or
        \item $\length(\tilde F) = 2$, $\tilde F$ is $(\phi, \gamma\delta)$-hopeful and $\{\gamma, \delta\} = \{\alpha, \beta\}$, or
        \item $\tilde F$ is $(\phi, \gamma\delta)$-hopeful, $\{\gamma, \delta\} \cap \{\alpha, \beta\} = \0$, and $\length(\tilde P) = 2\ell$, or
        \item $\tilde F$ is $(\phi, \gamma\delta)$-successful and $\{\gamma, \delta\} \cap \{\alpha, \beta\} = \0$.
    \end{itemize}
\end{lemma}

\begin{proof}
    Let $(F', \gamma, \delta)$ be the output of Algorithm~\ref{alg:next_shannon_fan} at step~\ref{step:next_fan}.
    By Lemma~\ref{lemma:next_shannon_fan},
    no edge in $F'$ is colored $\alpha,\,\beta,\,\gamma$ or $\delta$, and at least one of the following holds:
    \begin{itemize}
        \item either $\delta \in M(\phi, x)$ and $F'$ is $\phi$-happy, or
        \item $\length(F') = 2$, $\delta = \beta$ and $F'$ is $(\phi, \alpha\beta)$-hopeful, or
        \item $\length(F') = 2$, $\delta \neq \beta$ and $F'$ is $(\phi, \gamma\delta)$-successful, or 
        \item $\delta \neq \beta$ and $e$ is $(\phi, \gamma\delta)$-successful.
    \end{itemize}
    If the algorithm reaches step~\ref{step:next_happy_fan}, then $\tilde F = F'$ is $\phi$-happy.
    If not, $\length(F') = 2$.
    If the algorithm reaches step~\ref{step:delta_beta}, then $\tilde F = F'$ is a $(\phi, \alpha\beta)$-hopeful fan.
    Note that according to step~\ref{step:gamma_neq_alpha} in Algorithm~\ref{alg:next_shannon_fan}, at this point we have $\gamma \neq \alpha$.
    Furthermore, we have $M(\phi, x) \cap M(\phi, y) = \0$ and $M(\phi, x) \cap M(\phi, \vend(F')) = \0$.
    The former implies $\gamma \neq \beta$ and the latter implies $\delta \neq \alpha$.
    It follows that $\{\gamma, \delta\} \cap \{\alpha, \beta\} = \0$.
    The rest of the proof follows identically to that of Lemma~\ref{lemma:first_chain}.
\end{proof}

It follows that Algorithm~\ref{alg:multi_shannon_chain} outputs a $\phi$-happy chain as long as 
\begin{enumerate}[label=\ep{\normalfont{}\texttt{Happy}\arabic*},labelindent=15pt,leftmargin=*]
    \item\label{item:valid_input} the input to Algorithm~\ref{alg:next_shannon_chain} at step~\ref{step:alpha_beta_order} is valid,

    \item\label{item:never_fail} we never reach step~\ref{step:fail_chain}, and

    \item\label{item:invariants} the invariants~\ref{inv:start_F_end_C}--\ref{inv:hopeful_length} hold for each iteration of the \textsf{while} loop.
    
\end{enumerate}
The proofs for~\ref{item:valid_input}-\ref{item:invariants} are identical to those in \cite[Lemmas 5.4, 5.5, 5.6]{fastEdgeColoring} \textit{mutatis mutandis} and so we omit them here.

We conclude this subsection with two lemmas.
The first lemma describes an implication of Lemma~\ref{lemma:next_chain} on a certain kind of intersection.

\begin{lemma}\label{lemma:intersection_prev}
    Suppose we have a {\upshape\textsf{$0$-backward}} iteration of the {\upshape\textsf{while}} loop in Algorithm~\ref{alg:multi_shannon_chain}.
    Then the first intersection must occur at a vertex in $V(F_k)$.
\end{lemma}

\begin{proof}
    Consider such an iteration of the \textsf{while} loop.
    Let $C = F_0 + P_0 + \cdots + F_{k-1} + P_{k-1}$ and $F+P$ be the chain and the candidate chain at the start of the iteration such that $P$ is an $\alpha\beta$-path. 
    Let $P_k$ be the random truncation of $P$ computed on step~\ref{step:Pk}
    and let $\tilde F$, $\tilde P$ be the output of Algorithm~\ref{alg:next_shannon_chain} on step~\ref{step:alpha_beta_order} such that $\tilde P$ is a $\gamma\delta$-path.
    Since we reach step~\ref{step:truncate} with $j = k$, we must have
    \begin{itemize}
        \item either $E(\tilde F + \tilde P) \cap \IE(P_k) \neq \0$, or
        \item $V(\tilde F + \tilde P) \cap V(F) \neq \0$.
    \end{itemize}
    Suppose the first intersection is at an edge. 
    Let $\psi \defeq \Shift(\phi, C + F + P_k)$, and,
    without loss of generality, let $\alpha \in M(\psi, \Pivot(\tilde F))$.
    By Lemma~\ref{lemma:next_chain}, no edge in $\tilde F$ is colored $\alpha,\,\beta,\,\gamma$ or $\delta$ under $\psi$ and
    \begin{itemize}
        \item either $\tilde F$ is $\psi$-happy and $\tilde P = (\End(\tilde F))$, or
        \item $\length(\tilde F) = 2$, $\tilde F$ is $(\psi, \gamma\delta)$-hopeful and $\{\gamma, \delta\} = \{\alpha, \beta\}$, or
        \item $\tilde F$ is $(\psi, \gamma\delta)$-hopeful, $\{\gamma, \delta\} \cap \{\alpha, \beta\} = \0$, and $\length(\tilde P) = 2\ell$, or
        \item $\tilde F$ is $(\psi, \gamma\delta)$-successful and $\{\gamma, \delta\} \cap \{\alpha, \beta\} = \0$.
    \end{itemize}
    As no edge in $\tilde F$ is colored $\alpha$ or $\beta$ under $\psi$, $E(\tilde F) \cap E(P_k) = \set{\Start(\tilde F)}$, which is not a violation since $\Start(\tilde F)$ is not an internal edge of $P_k$.
    So we must have $E(\tilde P) \cap \IE(P_k) \neq \0$.
    Let $\psi'\defeq \Shift(\psi, \tilde F)$.
    As $E(\tilde F) \cap \IE(P_k) = \0$, $\psi'(e) = \psi(e) \in \set{\alpha, \beta}$ for all $e \in \IE(P_k)$.
    Therefore, we conclude that $\set{\gamma, \delta} \cap \set{\alpha, \beta} \neq \0$ and hence $\set{\gamma, \delta} = \set{\alpha, \beta}$ and $\length(\tilde F) = 2$.
    The rest of the proof follows identically to \cite[Lemma 5.8]{fastEdgeColoring}.
\end{proof}

The next lemma describes some properties of non-intersecting chains that will be useful in the proofs presented in the subsequent sections.


\begin{lemma}\label{lemma:non-intersecting_degrees}
    Let $\phi$ be a proper partial coloring and let $e$ be an uncolored edge such that $V(e) = \set{x,y}$. 
    Consider running Algorithm~\ref{alg:multi_shannon_chain} with input $(\phi, e, x)$.
    Let $C = F_0+P_0+\cdots+F_{k-1}+P_{k-1}$ be the multi-step Shannon chain at the beginning of an iteration of the {\upshape\textsf{while}} loop and let $F_k + P_k$ be the chain formed at step~\ref{step:Pk} such that, for each $j$, $P_j$ is an $\alpha_j\beta_j$-path in the coloring $\Shift(\phi, F_0+P_0+\cdots +F_{j-1}+P_{j-1})$. Then:
    \begin{enumerate}[label=\ep{\normalfont{}\texttt{Chain}\arabic*},labelindent=15pt,leftmargin=*]
        \item\label{item:degree_end} $\deg(\vend(F_j); \phi, \alpha_j\beta_j) = 1$  for each $0 \leq j \leq k$, and
        
        \item\label{item:related_phi} for each $0 \leq j \leq k$, all edges of $P_j$ except $\Start(P_j)$ are colored $\alpha_j$ or $\beta_j$ under $\phi$. 
    \end{enumerate}
\end{lemma}

\begin{proof}
    Let $\psi_0 \defeq \phi$ and $\psi_j \defeq \Shift(\phi, F_0+P_0+\cdots +F_{j-1}+P_{j-1})$ for $1 \leq j \leq k$. 
    By construction, items~\ref{item:degree_end},~\ref{item:related_phi} hold with $\phi$ replaced by $\psi_{j}$.
    Furthermore, as~\ref{inv:non_intersecting_shiftable} holds, the chain $C' \defeq C + F_k + P_k$ is non-intersecting and $\phi$-shiftable.
    The proof for~\ref{item:related_phi} follows identically to \cite[Lemma 5.7]{fastEdgeColoring}.
    Similarly, if $\length(F_j) = 2$, an identical proof as in \cite[Lemma 5.7]{fastEdgeColoring} proves~\ref{item:degree_end}.
    It remains to consider the case that $\length(F_j) = 1$.

    If $j = 0$, $\psi_0 = \phi$ and so~\ref{item:degree_end} holds. 
    Let $j > 0$ such that $\length(F_j) = 1$ and let $0 \leq i < j$ be arbitrary.
    We claim that $\deg(\vend(F_j); \psi_i, \alpha_j\beta_j) = \deg(\vend(F_j); \psi_{i+1}, \alpha_j\beta_j)$.
    As $i$ is arbitrary, this implies~\ref{item:degree_end}.
    We will consider two cases based off of the value of $i$.
    \begin{enumerate}[label=\ep{\normalfont{}\textbf{Case\arabic*}},labelindent=0pt,leftmargin=*]
        \item $i < j-1$. 
        If $\vend(F_j) \notin V(F_i + P_i)$, the claim follows, so let us assume $\vend(F_j) \in V(F_i + P_i)$.
        As $C'$ is non-intersecting, $\vend(F_j) \notin V(F_i)\cup V(F_{i+1})$ and so, $\vend(F_j) \in \IV(P_i) \setminus (V(F_i)\cup V(F_{i+1}))$.
        Therefore, the only change in the neighborhood of $\vend(F_j)$ from $\psi_i$ to $\psi_{i+1}$ is that the edges colored $\alpha_i$ and $\beta_i$ swap values.
        In particular, $\deg(\vend(F_j); \psi_i, \alpha_j\beta_j) \\= \deg(\vend(F_j); \psi_{i+1}, \alpha_j\beta_j)$.

        \item $i = j-1$. 
        Here, we have $\vend(F_j) = \vend(P_{j-1})$ as $\length(F_j) = 1$.
        Without loss of generality, let $\psi_{j-1}(\End(P_{j-1})) = \beta_{j-1}$.
        From Lemma~\ref{lemma:next_chain}, we may conclude that $\set{\alpha_j, \beta_j} \cap \set{\alpha_{j-1}, \beta_{j-1}} = \0$.
        As the only change on the neighborhood of $\vend(F_j)$ is that the edge colored $\beta_{j-1}$ is now blank, $\deg(\vend(F_j); \psi_i, \alpha_j\beta_j) = \deg(\vend(F_j); \psi_{i+1}, \alpha_j\beta_j)$.
        
    \end{enumerate}
    This covers all cases and completes the proof.
\end{proof}

\subsection{Analysis of the MSSA}\label{subsection:analysis_MSSA}

For this section, we will fix a proper partial coloring $\phi$. The main result is the following:

\begin{theorem}\label{theo:MSSA}
    Let $e$ be an uncolored edge such that $V(e) = \set{x, y}$. For any $t > 0$ and $\ell \geq 1296\Delta^{16}$, Algorithm~\ref{alg:multi_shannon_chain} with input $(\phi, e, x)$ computes an $e$-augmenting multi-step Shannon chain $C$ of length $O(\ell\,t)$ in time $O(\Delta\,\ell\,t)$ with probability at least $1 - 4m(1296\Delta^{15}/\ell)^{t/2}$.
\end{theorem}

As mentioned in \S \ref{subsection:overview}, we will employ the entropy compression argument to prove Theorem~\ref{theo:MSSA}.
As we may assume $\phi$ is fixed, for any $f \in E(G)$ and $z \in V(f)$, the first $t$ iterations of the \textsf{while} loop in Algorithm~\ref{alg:multi_shannon_chain} is uniquely determined by the \emphd{input sequence} $(f, z, \ell_1, \ldots, \ell_t)$, where $\ell_i \in [\ell, 2\ell-1]$ is the random choice made at step~\ref{step:Pk} during the $i$-th iteration.
We let $\mathcal{I}^{(t)}$ be the set of all input sequences for which Algorithm~\ref{alg:multi_shannon_chain} does not terminate within $t$ iterations.
We define \emphd{records and termini} identically to \cite[Definition 6.2]{fastEdgeColoring}:

\begin{definition}[Records and termini]\label{def:record}
    Let $I = (f, z, \ell_1, \ldots, \ell_t) \in \mathcal{I}^{(t)}$. Consider running Algorithm~\ref{alg:multi_shannon_chain} for the first $t$ iterations with this input sequence.
    We define the \emphd{record} of $I$ to be the tuple $D(I) = (d_1, \ldots, d_t) \in \Z^t$, where for each $i$, $d_i = 1$ if the $i$-th iteration is a {\upshape{\textsf{forward}}} iteration and $d_i = -r$ if it is an {\upshape{\textsf{$r$-backward}}} iteration.
    Note that in the second case, $d_i \leq 0$. 
    Let $C$ be the multi-step Shannon chain produced after the $t$-th iteration. 
    The \emphd{terminus} of $I$ is the pair $\tau(I) = (\End(C), \vend(C))$ (in the case that $C = (f)$, we define $\vend(C) \in V(f)$ to be distinct from $z$).
\end{definition}

Let $\mathcal{D}^{(t)}$ denote the set of all tuples $D \in \Z^t$ such that $D = D(I)$ for some $I \in \mathcal{I}^{(t)}$. 
Given $D \in \mathcal{D}^{(t)}$ and a pair $(f, u)$ such that $f \in E(G)$ and $u \in V(f)$, we let $\mathcal{I}^{(t)}(D, f, u)$ be the set of all input sequences $I \in \mathcal{I}^{(t)}$ such that $D(I) = D$ and $\tau(I) = (f, u)$.
The following functions will assist with our proof:
\begin{align*}
    \val(z) &\defeq \left\{\begin{array}{cc}
        3\Delta^4/2  & \text{if } z = 1 \\
        15\Delta^9  & \text{if } z = 0 \\
        36\ell\Delta^{11} & \text{if } z < 0
    \end{array}\right., & z \in \Z,\, z \leq 1, \\
    \wt(D) &\defeq \prod_{i = 1}^t\val(d_i), & D = (d_1, \ldots, d_t) \in \mathcal{D}^{(t)}.
\end{align*}
Let us now prove a key result in the analysis of Algorithm~\ref{alg:multi_shannon_chain}.

\begin{lemma}\label{lemma:records_bound_by_wt}
    Let $D = (d_1, \ldots, d_t) \in \mathcal{D}^{(t)}, \, f \in E(G)$ and $u \in V(f)$.
    Then $|\mathcal{I}^{(t)}(D, f, u)| \leq \wt(D)$, for $\ell \geq \Delta \geq 2$.
\end{lemma}

\begin{figure}[t]
    \centering
    \begin{subfigure}[t]{0.45\textwidth}
        \centering
    	\begin{tikzpicture}
                \clip (-0.5, -1) rectangle (6.6, 4.5);
                
                \node[circle,fill=blue,draw,inner sep=0pt,minimum size=4pt] (a) at (0,0) {};
                \path (a) ++(-30:1) node[circle,fill=black,draw,inner sep=0pt,minimum size=4pt] (b) {};
                \path (a) ++(30:1) node[circle,fill=blue,draw,inner sep=0pt,minimum size=4pt] (c) {};
                \path (c) ++(20:2.5) node[circle,fill=purple,draw,inner sep=0pt,minimum size=4pt] (d) {};
                \path (d) ++(-30:0.75) node[circle,fill=purple,draw,inner sep=0pt,minimum size=4pt] (e) {};
                \path (e) ++(-30:1) node[circle,fill=blue,draw,inner sep=0pt,minimum size=4pt] (f) {};
                \path (f) ++(-30:1) node[circle,fill=blue,draw,inner sep=0pt,minimum size=4pt] (g) {};

                \path (f) ++(30:1) node[circle,fill=black,draw,inner sep=0pt,minimum size=4pt] (h) {};
                \path (h) ++(130:3.5) node[circle,fill=red,draw,inner sep=0pt,minimum size=4pt] (i) {};
                \path (i) ++(-140:1) node[circle,fill=red,draw,inner sep=0pt,minimum size=4pt] (j) {};
                \path (e) ++(-80:1.5) node[circle,fill=red,draw,inner sep=0pt,minimum size=4pt] (k) {};

                \draw[thick,dotted] (a) -- (b);
                \draw[thick, color=blue] (a) -- (c) (f) -- (g);
                \draw[thick, color=red] (i) -- (j);
                \draw[thick, color=purple] (d) to node[font=\fontsize{8}{8},midway,inner sep=1pt,outer sep=1pt,minimum size=4pt,fill=white] {$h$} (e);
                \draw[thick] (f) -- (h);
                \draw[thick,decorate,decoration=zigzag,color=blue] (c) to[out=30, in=160] node[font=\fontsize{8}{8},midway,inner sep=1pt,outer sep=1pt,minimum size=4pt,fill=white] {$P_j$} (d) (e) -- (f);
                \draw[thick,decorate,decoration=zigzag,color=red] (j) to[out=-110, in=120] node[font=\fontsize{8}{8},midway,inner sep=1pt,outer sep=1pt,minimum size=4pt,fill=white] {$\tilde P$} (d) (e) to[out=-70,in=90] (k);
                \draw[thick,decorate,decoration=zigzag](h) to[out=45,in=30, looseness=1.8] (i);

    	\end{tikzpicture}
    	\caption{Intersection between $P_j,\, \tilde P$ at an edge $h$.}\label{fig:Sha_path_path_intersect}
    \end{subfigure}
    \begin{subfigure}[t]{0.45\textwidth}
        \centering
    	\begin{tikzpicture}
                \clip (-0.5, -1) rectangle (6.7, 4.5);
                
    	    \node[circle,fill=blue,draw,inner sep=0pt,minimum size=4pt] (a) at (0,0) {};
                \path (a) ++(-30:1) node[circle,fill=black,draw,inner sep=0pt,minimum size=4pt] (b) {};
                \path (a) ++(30:1) node[circle,fill=blue,draw,inner sep=0pt,minimum size=4pt] (c) {};
                \path (c) ++(20:2.5) node[circle,fill=purple,draw,inner sep=0pt,minimum size=4pt] (d) {};
                \path (d) ++(-30:1) node[circle,fill=purple,draw,inner sep=0pt,minimum size=4pt] (e) {};
                \path (e) ++(-30:1) node[circle,fill=blue,draw,inner sep=0pt,minimum size=4pt] (f) {};
                \path (f) ++(-30:1) node[circle,fill=blue,draw,inner sep=0pt,minimum size=4pt] (g) {};

                \path (f) ++(30:1) node[circle,fill=black,draw,inner sep=0pt,minimum size=4pt] (h) {};
                \path (h) ++(130:3.5) node[circle,fill=black,draw,inner sep=0pt,minimum size=4pt] (i) {};
                \path (i) ++(-140:1) node[circle,fill=black,draw,inner sep=0pt,minimum size=4pt] (j) {};
                \path (d) ++(-90:1) node[circle,fill=red,draw,inner sep=0pt,minimum size=4pt] (k) {};

                \draw[thick,dotted] (a) -- (b);
                \draw[thick, color=blue] (a) -- (c) (f) -- (g);
                \draw[thick, color=red] (d) to node[font=\fontsize{8}{8},midway,inner sep=1pt,outer sep=1pt,minimum size=4pt,fill=white] {$\tilde F$} (k);
                \draw[thick, color=purple] (d) to node[font=\fontsize{8}{8},midway,inner sep=1pt,outer sep=1pt,minimum size=4pt,fill=white] {$h$} (e);
                \draw[thick] (f) -- (h) (i) -- (j);
                \draw[thick,decorate,decoration=zigzag,color=blue] (c) to[out=30, in=160] node[font=\fontsize{8}{8},midway,inner sep=1pt,outer sep=1pt,minimum size=4pt,fill=white] {$P_j$} (d) (e) -- (f);
                \draw[thick,decorate,decoration=zigzag] (j) to[out=-110, in=120] (d);
                \draw[thick,decorate,decoration=zigzag](h) to[out=45,in=30, looseness=1.8] (i);
    		
    	\end{tikzpicture}
    	\caption{Intersection between $P_j,\,\tilde F$ at an edge $h$.}\label{fig:Sha_path_fan_intersect}
    \end{subfigure}

    \vspace{15pt}
    
    \begin{subfigure}[b]{0.45\textwidth}
        \centering
    	\begin{tikzpicture}
                \clip (-0.5, -1.3) rectangle (6, 1);
                
    	    \node[circle,fill=blue,draw,inner sep=0pt,minimum size=4pt] (a) at (0,0) {};
                \path (a) ++(-30:1) node[circle,fill=purple,draw,inner sep=0pt,minimum size=4pt] (b) {};
                \path (a) ++(30:1) node[circle,fill=blue,draw,inner sep=0pt,minimum size=4pt] (c) {};

                \path (b) ++(175:0.5) node[circle,fill=red,draw,inner sep=0pt,minimum size=4pt] (g) {};

                \path (c) ++(0:4) node[circle,fill=red,draw,inner sep=0pt,minimum size=4pt] (d) {};
                \path (d) ++(-30:1) node[circle,fill=black,draw,inner sep=0pt,minimum size=4pt] (e) {};
                \path (d) ++(-90:1) node[circle,fill=red,draw,inner sep=0pt,minimum size=4pt] (f) {};

                \draw[thick,dotted,color=blue] (a) -- (b);
                \draw[thick, color=blue] (a) to node[font=\fontsize{8}{8},midway,inner sep=1pt,outer sep=1pt,minimum size=4pt,fill=white] {$F_j$} (c);
                \draw[thick,decorate,decoration=zigzag] (c) to[out=20,in=160] (d);
                \draw[thick] (d) -- (e);
                \draw[thick,color=red] (d) -- (f) (b) -- (g);
                \draw[thick,decorate,decoration=zigzag,color=red] (f) to[out=200, in=-10] node[font=\fontsize{8}{8},midway,inner sep=1pt,outer sep=1pt,minimum size=4pt,fill=white] {$\tilde P$} (b);

                \node[anchor=north] at (b) {$w$};
    		
    	\end{tikzpicture}
    	
    	\caption{Intersection between $F_j,\, \tilde P$ at a vertex $w$.}\label{fig:Sha_fan_path_intersect}
    \end{subfigure}
    \begin{subfigure}[b]{0.45\textwidth}
        \centering
    	\begin{tikzpicture}
                \clip (-0.5, -2) rectangle (6, 1);
                
    	    \node[circle,fill=blue,draw,inner sep=0pt,minimum size=4pt] (a) at (0,0) {};
                \path (a) ++(-30:1) node[circle,fill=purple,draw,inner sep=0pt,minimum size=4pt] (b) {};
                \path (a) ++(30:1) node[circle,fill=blue,draw,inner sep=0pt,minimum size=4pt] (c) {};

                \path (b) ++(-45:1) node[circle,fill=red,draw,inner sep=0pt,minimum size=4pt] (g) {};
                \path (g) ++(195:1) node[circle,fill=red,draw,inner sep=0pt,minimum size=4pt] (h) {};

                \path (c) ++(0:4) node[circle,fill=black,draw,inner sep=0pt,minimum size=4pt] (d) {};
                \path (d) ++(-30:1) node[circle,fill=black,draw,inner sep=0pt,minimum size=4pt] (e) {};
                \path (d) ++(-90:1) node[circle,fill=black,draw,inner sep=0pt,minimum size=4pt] (f) {};

                \draw[thick,dotted,color=blue] (a) -- (b);
                \draw[thick, color=blue] (a) to node[font=\fontsize{8}{8},midway,inner sep=1pt,outer sep=1pt,minimum size=4pt,fill=white] {$F_j$} (c);
                \draw[thick,decorate,decoration=zigzag] (c) to[out=20,in=160] (d) (f) to[out=220, in=-10] (g);
                \draw[thick] (d) -- (e) (d) -- (f);
                \draw[thick,color=red] (b) -- (g) to node[font=\fontsize{8}{8},midway,inner sep=1pt,outer sep=1pt,minimum size=4pt,fill=white] {$\tilde F$} (h);

                \node[anchor=north east] at (b) {$w$};
    		
    	\end{tikzpicture}
    	
    	\caption{Intersection between $F_j,\, \tilde F$ at a vertex $w$.}\label{fig:Sha_fan_fan_intersect}
    \end{subfigure}
    \caption{Intersecting Multi-Step Shannon Chains at step~\ref{step:truncate}.}
    \label{fig:Sha_intersect}
\end{figure}
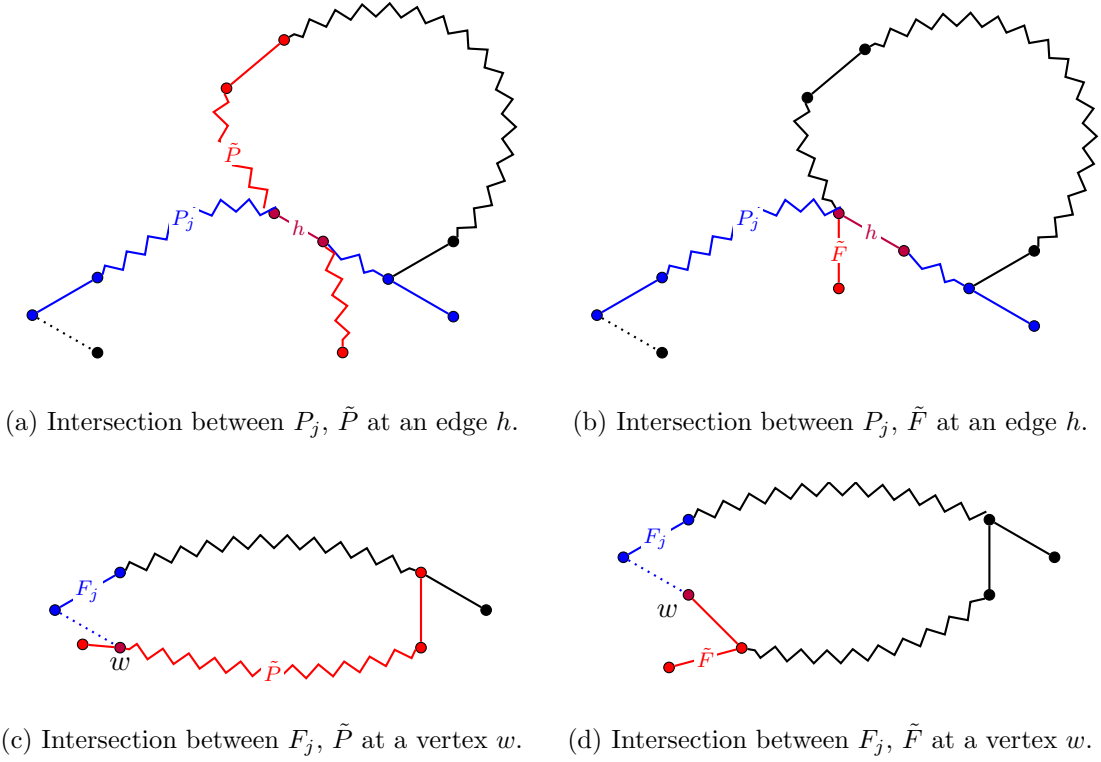

\begin{proof}
    Let $v \in V(f)$ be distinct from $u$.
    We will prove this by induction on $t$. For $t = 0$, we have
    \[\mathcal{I}^{(0)} = \set{(h, z)\,:\, h \in E(G), z \in V(h)}.\]
    Here, we have $\mathcal{I}^{(0)}((), f, u) = \set{(f, v)}$.
    In particular, $|\mathcal{I}((), f, u)| = 1 = \wt(())$.

    Now assume $t \geq 1$. 
    Consider any $(h, z, \ell_1, \ldots, \ell_t) \in \mathcal{I}^{(t)}(D, f, u)$ and run Algorithm~\ref{alg:multi_shannon_chain} with this input.
    Let $C = F_0 + P_0 + \cdots + F_{k-1} + P_{k-1}$ and $F+ P$ be the chain and candidate chain at the start of the $t$-th iteration of the \textsf{while} loop.
    On step~\ref{step:Pk}, we compute $P_k = P|\ell_t$ and the second step chain $\tilde F + \tilde P$. Note that for
    \[D' \defeq (d_1, \ldots, d_{t-1}), \quad f' \defeq \Start(F), \quad u' \defeq \vstart(F),\]
    the input sequence $(h, z, \ell_1, \ldots, \ell_{t-1}) \in \mathcal{I}^{(t-1)}(D', f', u')$.
    By the induction hypothesis, it is now enough to show that the number of choices $(f', u', \ell_t)$ is at most $\val(d_t)$. 
    The following claim will assist with the remainder of the proof.

    \begin{claim}\label{claim:one_step_chain}
        Given the vertex $\vend(P_k) = \vstart(\tilde F)$, there are at most $3\Delta^4/2$ options for $(f', u', \ell_t)$.
    \end{claim}

    \begin{claimproof}
        Given $\vstart(\tilde F)$, there are at most $\Delta$ choices for the edge $\Start(\tilde F) = \End(P_k)$.
        With this in hand, we can determine $\Pivot(\tilde F) \in V(\Start(\tilde F))$ to be distinct from $\vstart(\tilde F)$.
        By Lemma~\ref{lemma:non-intersecting_degrees} \ref{item:related_phi} all edges of $P_k$ except $\Start(P_k)$ are colored one of two colors, say $\alpha$ or $\beta$.
        One of these colors is $\phi(\End(P_k))$ and there are at most $\Sha$ choices for the other.
        Consider the maximal $\alpha\beta$-path under $\phi$ containing $\End(P_k)$.
        By Lemma~\ref{lemma:non-intersecting_degrees}~\ref{item:degree_end} and~\ref{item:related_phi}, $\vend(F_k)$ is the endpoint of this path lying on the side of $\Pivot(\tilde F)$.
        As $v' = \Pivot(F_k) \in N_G(\vend(F_k))$ and $f' \in E_G(v')$, there are at most $\Delta^2$ choices for $f'$ and $v'$.
        Given $f'$ and $v'$, we can determine $u'\in V(f')$ to be distinct from $v'$.
        Finally, the path $P_k$ starts at $v'$, passes through $\vend(F_k)$ and ends at $\vstart(\tilde F)$. 
        With this information, we can determine $\ell_t = \length(P_k)$, completing the proof.
    \end{claimproof}

    As a consequence of Claim~\ref{claim:one_step_chain}, it is now enough to show that there are at most $\val(d_t)/(3\Delta^4/2)$ choices for $\vend(P_k) = \vstart(\tilde F)$.
    We will consider each case separately.

    \begin{enumerate}[label=\ep{\normalfont{}\textbf{Case\arabic*}},labelindent=0pt,leftmargin=*]
        \item $d_t = 1$. Here, $u = \vend(P_k)$. So the number of choices for $\vend(P_k)$ is $1 = \val(d_t)/(3\Delta^4/2)$.

        \item $d_t < 0$. Here we have $v = \Pivot(F_j)$, where $j = k+d_t$ and the first intersection between $\tilde F + \tilde P$ and $C + F_k + P_k$ occurred at $F_j + P_j$.

        Recall from Definition~\ref{defn:non-int}, we must have either $E(\tilde F + \tilde P) \cap \IE(P_j) \neq \0$ or $V(\tilde F + \tilde P) \cap V(F_j) \neq \0$. The possible configurations are in Fig.~\ref{fig:Sha_intersect}.
        Let us consider each case separately.

        \begin{enumerate}[label=\ep{\normalfont{}\texttt{Subcase}2\alph*},labelindent=0pt,leftmargin=*]

            \item $E(\tilde F + \tilde P) \cap \IE(P_j) \neq \0$ (Fig.~\ref{fig:Sha_path_path_intersect},~\ref{fig:Sha_path_fan_intersect}).
            Note that in this case, we must have $V(\tilde F + \tilde P) \cap \IV(P_j) \neq \0$. 
            The vertex $\vend(F_j) \in N_G(v)$ and so there are at most $\Delta$ choices for it.
            By Lemma~\ref{lemma:non-intersecting_degrees}~\ref{item:related_phi} all edges of $P_j$ except $\Start(P_j)$ are colored one of two colors, say $\gamma$ or $\delta$, for which there are at most $\Sha^2$ choices.
            By Lemma~\ref{lemma:non-intersecting_degrees}~\ref{item:degree_end}, $P_j$ is a path starting at $\End(F_j)$ and followed by an initial segment of $P' \defeq G(\vend(F_j); \phi, \gamma\delta)$ of length at most $2\ell$.
            Let $w \in \IV(P_j)$ be the first vertex of intersection. 
            There are at most $2\ell$ choices for $w$.
            If $w \in V(\tilde F)$, then $\vstart(\tilde F)$ is at most distance $2$ away from $w$ and so there are at most $\Delta^2 + 1$ choices for $\vstart(\tilde F)$.
            Now suppose $w \in V(\tilde P)$.
            Note that $\tilde P$ is a $\xi\eta$-path in the coloring $\Shift(\phi, C + F_k + P_k)$ for some $\xi,\eta \in [\Sha]$.
            By our choice of $j$ and $w$, an identical argument to Lemma~\ref{lemma:non-intersecting_degrees} shows that $\deg(\vend(\tilde F); \phi, \xi\eta) = 1$, and $w$ and $\vend(\tilde F)$ are $(\phi, \xi\eta)$-related.
            Given $\vend(\tilde F)$, there are at most $\Delta^2$ choices for $\vstart(\tilde F)$. 
            In particular, the number of choices for $\vstart(\tilde F)$ is at most
            \[\Delta\,\Sha^2\,2\ell(\Delta^2 + 1 + 2\Sha^2\Delta^2) \leq \frac{45}{2}\ell\Delta^7,\]
            for $\Delta \geq 2$.
            
            \item\label{subcase:vtx} $V(\tilde F + \tilde P) \cap V(F_j) \neq \0$ (Fig.~\ref{fig:Sha_fan_path_intersect},~\ref{fig:Sha_fan_fan_intersect}).
            If $V(\tilde F) \cap V(F_j) \neq \0$, $\vstart(\tilde F)$ is at most distance $3$ away from $v$ and so there are at most $\Delta^3 + 1$ choices for $\vstart(\tilde F)$.
            If $V(\tilde P) \cap V(F_j) \neq \0$, there are at most $\Delta + 1$ choices for the intersection vertex $w \in N_G[v]$.
            Following an identical argument in the previous case, the number of choices for $\vstart(\tilde F)$ is at most
            \[\Delta^3 + 1 + (\Delta + 1)\,2\Sha^2\,\Delta^2 \leq 10\Delta^5,\]
            for $\Delta \geq 2$.
            
        \end{enumerate}
        Putting together both cases, the number of choices for $\vstart(\tilde F)$ is at most
        \[\frac{45}{2}\ell\Delta^7 + 10\Delta^5 \leq 24\ell\Delta^7 = \frac{\val(d_t)}{3\Delta^4/2},\]
        for $\ell \geq \Delta \geq 2$.

        \item $d_t = 0$. Here, we have $j = k$. As a result of Lemma~\ref{lemma:intersection_prev}, we need only consider~\ref{subcase:vtx} above. In particular, the number of choices for $\vstart(\tilde F)$ is at most
        \[10\Delta^5 = \frac{\val(d_t)}{3\Delta^4/2},\]
        as desired.
    \end{enumerate}
    This covers all cases and completes the proof.
\end{proof}

Next we prove an upper bound on $\wt(D)$ in terms the sum of the entries of $D$. To this end, let
\[\mathcal{D}_s^{(t)} \defeq \left\{D = (d_1, \ldots, d_t)\in \mathcal{D}^{(t)}\,:\, \sum_{i = 1}^td_i = s\right\}.\]

\begin{lemma}\label{lemma:wD_bound}
    Let $D \in \mathcal{D}_s^{(t)}$. Then $\wt(D) \leq (81\Delta^{15}\ell)^{t/2}\,(36\Delta^7\ell)^{-s/2}$, for $\ell \geq 7\Delta^3$.
\end{lemma}

\begin{proof}
    Let $D = (d_1, \ldots, d_t) \in \mathcal{D}_s^{(t)}$.
    Define
    \begin{align*}
        I\defeq \{i\,:\,d_i = 1\}, \quad J\defeq \{i\,:\,d_i = 0\}, \quad K\defeq [t] \setminus (I\cup J).
    \end{align*}
    Using the definition of $\wt(D)$ and $\val(z)$, we write
    \begin{align*}
        \wt(D) \,&=\, (3\Delta^4/2)^{|I|}\,(15\Delta^9)^{|J|}\, (36\ell\Delta^{11})^{|K|} \\
        &=\, (3/2)^{|I|}\,15^{|J|}\,36^{|K|}\,\ell^{|K|}\,\Delta^{4|I|+9|J|+11|K|} \\
        &=\, (3/2)^{|I|}\,15^{|J|}\,36^{|K|}\,\ell^{|K|}\,\Delta^{4t+5|J|+7|K|},
    \end{align*}
    where we use the fact that $|I| + |J| + |K| = t$.
    Note the following:
    \[s \,=\, |I| - \sum_{k \in K}|d_k| \,\leq\, |I| - |K| \,=\, t - |J| - 2|K|.\]
    It follows that $|K| \leq (t - s - |J|)/2$. With this in mind, we have 
    \begin{align*}
        \wt(D) \,&\leq\, (3/2)^{|I|}\,15^{|J|}\, (36\,\ell)^{(t-s - |J|)/2} \, \Delta^{4t+ 5|J| + 7(t-s-|J|)/2}\\
        &=\,(3/2)^{|I|}\, (36\Delta^{15}\ell)^{t/2}\,(36\Delta^7\ell)^{-s/2}\,\left(\frac{225\Delta^3}{36\ell}\right)^{|J|/2} \\
        & \leq\,  (81\Delta^{15}\ell)^{t/2}\,(36\Delta^7\ell)^{-s/2},
    \end{align*}
    for $\ell \geq 7\Delta^3$.
\end{proof}

The following result bounds the size of $\mathcal{D}_s^{(t)}$.

\begin{lemma}[{\cite[Lemma 6.5]{fastEdgeColoring}}]\label{lemma:Ds_bound}
    $|\mathcal{D}_s^{(t)}| \leq 4^t$.
\end{lemma}

With this in hand, we are now ready to prove Theorem~\ref{theo:MSSA}.

\begin{proof}[Proof of Theorem~\ref{theo:MSSA}]
    Note that each $\ell_i$ is chosen uniformly at random from $[\ell, 2\ell-1]$. 
    In particular, the probability Algorithm~\ref{alg:multi_shannon_chain} lasts at least $t$ iterations is bounded above by $|\mathcal{I}^{(t)}|/\ell^t$.
    As a result of Lemmas~\ref{lemma:records_bound_by_wt},~\ref{lemma:wD_bound} and~\ref{lemma:Ds_bound}, we have
    \begin{align*}
        |\mathcal{I}^{(t)}| \leq \sum_{\substack{f \in E(G),\, u \in V(f), \\ D \in \mathcal{D}^{(t)}}}|\mathcal{I}^{(t)}(D, f, u)| &\leq  \sum_{\substack{f \in E(G),\, u \in V(f), \\ D \in \mathcal{D}^{(t)}}} \wt(D) \\
        &\leq 2m\sum_{s = 0}^t\sum_{D\in \mathcal{D}_s^{(t)}} \wt(D) \\
        &\leq 2m\sum_{s = 0}^t|\mathcal{D}_s^{(t)}|(81\Delta^{15}\ell)^{t/2}\,(36\Delta^7\ell)^{-s/2} \\
        &\leq 2m (1296\Delta^{15}\ell)^{t/2} \sum_{s = 0}^t(36\Delta^7\ell)^{-s/2} \\
        &\leq 4m (1296\Delta^{15}\ell)^{t/2}.
    \end{align*}
    Therefore, we may conclude that
    \[\Pr[\text{Algorithm~\ref{alg:multi_shannon_chain} does not terminate in $t$ iterations}] \leq 4m\left(\frac{1296\Delta^{15}}{\ell}\right)^{t/2}.\]
    It remains to bound the runtime.
    Consider running Algorithm~\ref{alg:multi_shannon_chain} on an input sequence $I \in \mathcal{I}^{(t)}$. Let $D(I) = (d_1, \ldots, d_t)$.
    We will bound the runtime of each iteration by considering $d_i$.
    As shown in \S\ref{subsec:alg_overview}, Algorithms~\ref{alg:first_shannon_chain} and~\ref{alg:next_shannon_chain} run in $O(\Delta\ell)$ time. 
    Also, while updating the coloring $\psi$, we must update the sets of missing colors as well. 
    Each update to $M(\cdot)$ takes $O(\Delta)$ time. Let us now consider two cases depending on whether $d_i$ is positive.
    \begin{enumerate}[label=\ep{\normalfont{}\textbf{Case\arabic*}},labelindent=0pt,leftmargin=*]
        \item $d_i = 1$. Here, we test for success, shorten the path $P$ to $P_k$, update $\psi$ and the hash map $\visited$,
        call Algorithm~\ref{alg:next_shannon_chain}, and check that the resulting chain is non-intersecting. All of these steps can be performed in time $O(\Delta\ell)$.
        
        \item $d_i \leq 0$. Here we again conduct all the steps mentioned in the previous case, which takes $O(\Delta \ell)$ time.
        Once we locate the intersection, we then update $\psi$ and the hash map $\visited$, which takes $O(\Delta(|d_i| + 1)\ell)$ time (we are adding $1$ to account for the case $d_i = 0$).
    \end{enumerate}
    Since $\sum_{i=1}^t d_i \geq 0$, we have
    \[\sum_{i,\, d_i \leq 0}|d_i| \leq \sum_{i,\, d_i = 1}d_i \,\leq\, t.\]
    It follows that the running time of $t$ iterations of the \textsf{while} loop is $O(\Delta\ell\,t)$, as desired.
\end{proof}

%% file: shannon.tex
\section{Proof of Theorem~\ref{theo:shannon}}\label{section:shannon_proof}

In this section, we will prove Theorem~\ref{theo:shannon}.
In particular, we will show how to apply the \hyperref[alg:multi_shannon_chain]{MSSA} as a subprocedure to define efficient algorithms for $\Sha$-edge-coloring.
We will split this section into two subsections. 
The first will consider the randomized sequential algorithm, and the second, the distributed algorithms.

\subsection{Sequential Algorithm}

Let us describe our randomized sequential algorithm (see Algorithm~\ref{alg:random_seq} for full details). 
The algorithm takes as input a multigraph $G$ of maximum degree $\Delta$ and outputs a proper $\Sha$-edge-coloring of $G$.
At each iteration, the algorithm picks an uncolored edge uniformly at random and colors it using Algorithm~\ref{alg:multi_shannon_chain} as a subprocedure.

\begin{algorithm}[h]
\caption{Sequential Coloring with Multi-Step Shannon Chains}\label{alg:random_seq}
\begin{flushleft}
\textbf{Input}: A multigraph $G = (V, E)$ of maximum degree $\Delta$. \\
\textbf{Output}: A proper $\Sha$-edge-coloring $\phi$ of $G$.
\end{flushleft}
\begin{algorithmic}[1]
    \State $U \gets E$, $\phi(e) \gets \blank$ for each $e \in U$.
    \While{$U \neq \0$}
        \State Pick an edge $e \in U$ and a vertex $x \in V(e)$ uniformly at random.
        \State $C \gets \hyperref[alg:multi_shannon_chain]{\mathsf{MSSA}}(\phi, e, x)$ \Comment{Algorithm~\ref{alg:multi_shannon_chain}}
        \State $\phi \gets \aug(\phi, C)$
        \State $U \gets U \setminus \set{e}$
    \EndWhile
    \State \Return $\phi$
\end{algorithmic}
\end{algorithm}

The correctness of the algorithm follows from the results of \S\ref{subsection:mssa_poc}. 
To assist with the runtime analysis, we define the following parameters:
\begin{align*}
    \phi_i &\defeq \text{ the coloring at the start of the $i$-th iteration}, \\
    T_i &\defeq \text{the number of iterations of the \textsf{while} loop in the $i$-th call to Algorithm~\ref{alg:multi_shannon_chain}}, \\
    C_i &\defeq \text{the chain output by the $i$-th call to Algorithm~\ref{alg:multi_shannon_chain}}, \\
    U_i &\defeq \set{e \in E(G) \,:\, \phi_i(e) = \blank}.
\end{align*}
Augmenting a chain $C$ can be done in time $O(\Delta\length(C))$ (the factor of $\Delta$ comes from updating the missing sets).
It follows from Theorem~\ref{theo:MSSA} that the runtime of Algorithm~\ref{alg:random_seq} is
\[ \sum_{i = 1}^mO(\Delta\,\ell\,T_i + \Delta\length(C_i)) = \sum_{i = 1}^mO(\Delta\,\ell\,T_i) = O(\Delta\,\ell\,T),\]
where $T = \sum_{i = 1}^mT_i$.
It remains to show that $T = O(m)$ with high probability. 
To this end, we prove the following lemma.

\begin{lemma}\label{lemma:martingale}
    For all $t_1, \ldots, t_i$, we have 
    \[\Pr[T_i \geq t_i\mid T_1 \geq t_1, \ldots, T_{i-1} \geq t_{i-1}] \leq \frac{2m}{|U_i|}\left(\frac{1296\Delta^{15}}{\ell}\right)^{t_i/2}.\]
\end{lemma}

\begin{proof}
    Recall the definition of input sequences from \S\ref{subsection:analysis_MSSA}.
    We note that $I = (f, z, \ell_1, \ldots, \ell_t)$ is defined for an arbitrary starting edge.
    As we are choosing the starting edge and pivot vertex uniformly at random at each iteration, we have
    \[\Pr[T_i \geq t_i\mid \phi_i] \leq \frac{|\mathcal{I}^{(t_i)}|}{2|U_i|\ell^{t_i}}.\]
    It was shown in the proof of Theorem~\ref{theo:MSSA} that
    \[|\mathcal{I}^{(t_i)}| \leq 4m(1296\Delta^{15}\ell)^{t_i/2}.\]
    As $T_i$ is independent of $T_1, \ldots, T_{i-1}$ given $\phi_i$, it follows that
    \[\Pr[T_i \geq t_i\mid T_1 \geq t_1, \ldots, T_{i-1} \geq t_{i-1}] \leq \frac{|\mathcal{I}^{(t_i)}|}{2|U_i|\ell^{t_i}} \leq \frac{2m}{|U_i|}\left(\frac{1296\Delta^{15}}{\ell}\right)^{t_i/2},\]
    as desired.
\end{proof}

Let us now bound $\Pr[T \geq t]$. 
We note that $T\geq t$ implies that there exists $t_1, \ldots, t_m$ such that $\sum_i t_i = t$ and $T_i \geq t_i$ for all $i$.
With this in hand, we have:
\begin{align*}
    \Pr[T \geq t] \,&\leq\, \sum_{\substack{t_1, \ldots, t_m \\ \sum_it_i = t}}\Pr[T_1 \geq t_1, \ldots, T_m \geq t_m] \\
    &=\, \sum_{\substack{t_1, \ldots, t_m \\ \sum_it_i = t}}\,\prod_{i = 1}^m\Pr[T_i \geq t_i|T_1 \geq t_1, \ldots, T_{i-1} \geq t_{i-1}] \\
    &\leq\, \sum_{\substack{t_1, \ldots, t_m \\ \sum_it_i = t}}\,\prod_{i = 1}^m\frac{2m}{|U_i|}\,\left(\frac{1296\Delta^{15}}{\ell}\right)^{t_i/2} \\
    &=\, \sum_{\substack{t_1, \ldots, t_m \\ \sum_it_i = t}}\frac{(2m)^m}{m!}\,\left(\frac{1296\Delta^{15}}{\ell}\right)^{t/2} \\
    &\leq\, \binom{t+m-1}{m-1} \,(2e)^{m}\,\left(\frac{1296\Delta^{15}}{\ell}\right)^{t/2} \\
    &\leq\, \left(e\left(1 + \frac{t}{m-1}\right)\right)^{m-1}\,(2e)^{m}\,\left(\frac{1296\Delta^{15}}{\ell}\right)^{t/2} \\
    &\leq\, (2e^2)^m\,\left(\frac{1296\,e^2\,\Delta^{15}}{\ell}\right)^{t/2}.
\end{align*}
For $\ell = \Theta(\Delta^{16})$ and $t = \Theta(m)$, the above is at most $1/\Delta^n$.
It follows that Algorithm~\ref{alg:random_seq} computes a proper $\Sha$-edge-coloring in $O(\Delta^{18}\,n)$ time with probability at least $1 - 1/\Delta^n$, completing the proof of Theorem~\ref{theo:shannon}\,\eqref{item:seq_sha}.

\subsection{Distributed Algorithms}

In this subsection, we will define our distributed algorithms and prove the bounds on their runtimes.
The main result is the following.

\begin{theorem}\label{theo:disjoint_with_bounds}
    Let $\phi$ be a proper partial $\Sha$-edge-coloring of a multigraph $G$ with domain $\dom(\phi) \subset E(G)$ and let $U \defeq E(G) \setminus \dom(\phi)$ be the set of uncolored edges. Then there exists a randomized \LOCAL algorithm that in $O(\Delta^{16} \log n)$ rounds outputs a set $W \subseteq U$ of expected size $\E[|W|] = \Omega(|U|/\Delta^{20})$ and an assignment of connected $e$-augmenting subgraphs $H_e$ to the edges $e \in W$ such that:
    \begin{itemize}
        \item the multigraphs $H_e$, $e \in W$ are pairwise vertex-disjoint, and
        \item $|E(H_e)| = O(\Delta^{16}\log n)$ for all $e \in W$.
    \end{itemize}
\end{theorem}

We will prove the above theorem in \S\ref{subsubsection:many_disjoint}. Before we do so, let us show how it can be used to prove our results.

\begin{proof}[Proof of Theorem~\ref{theo:shannon}\,\eqref{item:dist_rand_sha}]
    Our algorithm will proceed in stages.
    Let $\phi_0$ be the blank coloring, i.e., $\phi_0(e) = \blank$ for all $e \in E(G)$.
    At the start of the $i$-th stage we have a partial coloring $\phi_{i-1}$ where $U_{i-1} \subseteq E(G)$ is the set of uncolored edges under $\phi_{i-1}$.
    We run the following steps:
    \begin{itemize}
        \item Compute $W_i \subseteq U_{i-1}$ and $H_e$ for each $e \in W_i$ by running the algorithm from Theorem~\ref{theo:disjoint_with_bounds}.
        \item Augment the coloring $\phi_{i-1}$ with the multigraphs $H_e$ to get $\phi_i$.
    \end{itemize}
    From Theorem~\ref{theo:disjoint_with_bounds} it follows that each stage can be implemented in $O(\Delta^{16}\log n)$ rounds.
    We must now show that we can color the multigraph in $\poly(\Delta) \log n$ stages with high probability.
    To this end, note by Theorem~\ref{theo:disjoint_with_bounds}
    \[\EE[|W_i|\mid U_{i-1}] = \Omega\left(\frac{|U_{i-1}|}{\Delta^{20}}\right) \implies \EE[|U_{i}|\mid U_{i-1}] = \left(1 - \Omega\left(\Delta^{-20}\right)\right)|U_{i-1}|.\]
    In particular, as $|U_0| = m$, we have
    \[\EE[|U_T|] \leq m\exp\left(-T\,\Omega\left(\Delta^{-20}\right)\right).\]
    By Markov's inequality, we have
    \[\Pr[|U_T| \geq 1] \leq \EE[|U_T|] \leq m\exp\left(-T\,\Omega\left(\Delta^{-20}\right)\right).\]
    For $T = \Theta(\Delta^{20}\log n)$, the above is at most $1/\poly(n)$.
    It follows that our randomized distributed algorithm computes a proper $\Sha$-edge-coloring of $G$ in $O(\Delta^{36}\log^2 n)$ rounds.
\end{proof}

For our deterministic distributed Algorithm, we will utilize approximation algorithms for hypergraph maximal matching exactly as in \cite{VizingChain, Christ, fastEdgeColoring}.
Let us first introduce some new definitions and notation. 
Recall that $\log^*n$ denotes the iterated logarithm of $n$, i.e., the number of iterative applications of the logarithm function to $n$ after which the output becomes less than $1$.
Let $\mathcal{H}$ be a hypergraph. 
We let
\[r(\mathcal{H}) \defeq \max_{e \in E(\mathcal{H})}|e|, \quad \deg_\mathcal{H}(x) \defeq |\set{e\,:\,x\in e}|, \quad \Delta(\mathcal{H}) \defeq \max_{x\in V(\mathcal{H})}\deg_\mathcal{H}(x)\]
denote the \emphd{rank} of an edge $e$, \emphd{degree} of a vertex x and \emphd{maximum degree} of $\mathcal{H}$, repectively.
A \emphd{matching} in $\mathcal{H}$ is a set of disjoint hyperedges. 
We let $\mu(\mathcal{H})$ be the size of a maximum matching in $\mathcal{H}$.

While the usual \LOCAL model is defined for graphs, there is an analogous model operating on a hypergraph $\mathcal{H}$. 
Namely, in a single communication round of the \LOCAL model on $\mathcal{H}$, each vertex $x \in V(\mathcal{H})$ is allowed to send messages to every vertex $y \in V(\mathcal{H})$ such that $x$ and $y$ belong to a common hyperedge.

We shall use the following result due to Harris (based on earlier work of Ghaffari, Harris, and Kuhn \cite{derandomizing}):

\begin{theorem}[{Harris \cite[Theorem 1.1]{harrisdistributed}}]\label{theo:hypergraph_matching}
There exists a deterministic distributed algorithm in the \LOCAL model on an $n$-vertex hypergraph $\mathcal{H}$ that outputs a matching $M \subseteq E(\mathcal{H})$ with $|M| = \Omega(\mu/r)$ in
$\tilde{O}(r \log d + \log^2 d + \log^* n)$
rounds, where $r \defeq r(H)$, $d \defeq \Delta(H)$, and $\mu \defeq \mu(H)$.
\end{theorem}

With this in hand, we are ready to prove our result.

\begin{proof}[Proof of Theorem~\ref{theo:shannon}\,\eqref{item:dist_det_sha}]
    Once again our algorithm will proceed in stages. 
    Let $\phi_0$ be the blank coloring.
    At the start of the $i$-th stage we have a partial coloring $\phi_{i-1}$ where $U_{i-1} \subseteq E(G)$ is the set of uncolored edges under $\phi_{i-1}$.
    Let us describe a single stage in detail. We will run the following steps:
    \begin{itemize}
        \item First, we define an auxiliary hypergraph $\Gamma_i$ as follows: Set $V(\Gamma_i) \defeq V(G)$ and let $S \subseteq V(G)$ be a hyperedge of $\Gamma_i$ if and only if there exists an uncolored edge $e \in U_{i-1}$ and a connected $e$-augmenting subgraph $H$ of $G$ (with respect to $\phi_{i-1}$) with $S = V(H)$ and $|E(H)| = O(\Delta^{16}\log n)$ (the implied constant factors in this bound are the same as in Theorem~\ref{theo:disjoint_with_bounds}).

        \item Next, we run the Algorithm defined by Theorem \ref{theo:hypergraph_matching} to obtain a matching $M_i$ of $\Gamma_i$.

        \item By definition, each edge $S \in M_i$ contains an augmenting subgraph $H$ such that $V(H) = S$.
        For each $S$, we may pick an arbitrary such graph $H_S$.
        As $M_i$ is a matching, the multigraphs $H_S$ are vertex disjoint and can be simultaneously augmented to obtain the coloring $\phi_i$.
        
    \end{itemize}
    By definition, $r(\Gamma_i) = O(\Delta^{16} \log n)$, and, by Theorem~\ref{theo:disjoint_with_bounds}, $\mu(\Gamma_i) = \Omega(|U_{i-1}|/\Delta^{20})$. 
    To bound $\Delta(\Gamma_i)$, consider an edge $S \in E(\Gamma_i)$. 
    Since the multigraph $G[S]$ is connected, we can order the vertices of $S$ as $(x_0, \ldots, x_{|S| - 1})$ so that each $x_j$, $j \geq 1$ is adjacent to at least one of $x_0$, \ldots, $x_{j-1}$. 
    This means that once $(x_0, \ldots, x_{j-1})$ are fixed, there are at most $j\Delta \leq r(\Gamma_i)\Delta$ choices for $x_j$, and hence,
    \[\Delta(\Gamma_i) \,\leq\, (r(\Gamma_i)\Delta)^{r(\Gamma_i)} \,\leq\, \exp\left(2r(\Gamma_i) \log r(\Gamma_i)\right).\]
    By Theorem~\ref{theo:hypergraph_matching}, the matching $M_i \subseteq E(\Gamma_i)$ satisfies
    \[
        |M_i| \,=\, \Omega\left(\frac{\mu(\Gamma_i)}{r(\Gamma_i)}\right) \,=\, \Omega\left(\frac{|U_{i-1}|}{\Delta^{36}\,\log n}\right).
    \]
    Furthermore, $M_i$ can be found in $\tilde{O}(\Delta^{32} \log^2 n)$ rounds in the \LOCAL model on $\Gamma_i$. Since a single round of the \LOCAL model on $\Gamma_i$ can be simulated by $O(\Delta^{16}\log n)$ rounds in the \LOCAL model on $G$, $M_i$ can be found in $\tilde{O}(\Delta^{48} \log^3 n)$ rounds in the \LOCAL model on $G$. 
    Augmenting the subgraphs $H_S$ can be done in $O(\Delta^{16} \log n)$ rounds and so each stage takes $\tilde O(\Delta^{48}\log^3 n)$ rounds.
    
    It remains to show that we can color the multigraph in $\poly(\Delta) \log^2 n$ stages.
    To this end, we note that
    \[|U_i| \leq \left(1 - \Omega\left(\frac{1}{\Delta^{36}\log n}\right)\right)|U_{i-1}|.\]
    In particular, as $|U_0| = m$, we have
    \[|U_T| \leq m\exp\left(-T\,\Omega\left(\frac{1}{\Delta^{36}\log n}\right)\right).\]
    For $T = \Theta(\Delta^{36}\log^2 n)$, the above is less than $1$.
    It follows that our deterministic distributed algorithm computes a proper $\Sha$-edge-coloring of $G$ in $O(\Delta^{84}\log^5 n)$ rounds.
\end{proof}

\subsubsection{Proof of Theorem~\ref{theo:disjoint_with_bounds}}\label{subsubsection:many_disjoint}

Let $\phi$, $U$ be as in Theorem~\ref{theo:disjoint_with_bounds}. 
Throughout this section, we will fix $\ell = \Theta(\Delta^{16})$ and $t = \Theta(\log n)$, where the implicit constants will be assumed to be sufficiently large.
The \LOCAL algorithm for this section will use a simple randomized independent set algorithm as a subprocedure.
The details are in Algorithm~\ref{alg:rand_ind_set}.
The following result gives a lower bound on the expected size of the output.

\begin{lemma}[{\cite[Lemma 8.2]{fastEdgeColoring}}]\label{lemma:rand_ind_set}
    Let $W$ be the independent set returned by Algorithm \ref{alg:rand_ind_set} on a graph $\Gamma$ of average degree $d(\Gamma)$. 
    Then $\E[|W|] = |V(\Gamma)|/(d(\Gamma) + 1)$.
\end{lemma}

\begin{algorithm}[t]
    \caption{Distributed Random Independent Set \cite[Algorithm~8.1]{fastEdgeColoring}}\label{alg:rand_ind_set}
    \begin{flushleft}
        \textbf{Input}: A graph $\Gamma$. \\
        \textbf{Output}: An independent set $W$.
    \end{flushleft}
    \begin{algorithmic}[1]
        \State Let $x_v$ be i.i.d. $\mathcal{U}(0,1)$ for each $v \in V(\Gamma)$.
        \If{$x_v > \max_{u \in N_\Gamma(v)}x_u$}
            \State $W \gets W \cup \set{v}$
        \EndIf
        \State \Return $W$
    \end{algorithmic}
\end{algorithm}

With this in hand, let us describe the steps of our \LOCAL algorithm:
\begin{enumerate}
    \item\label{step:find_chains} For each $e \in U$, run Algorithm~\ref{alg:multi_shannon_chain} in parallel for $t$ iterations.
    Let $S \subseteq U$ be the set of edges for which the algorithm terminates and let $C_e$ be the chain computed for $e \in S$.

    \item Define a graph $\Gamma$ as follows: Set $V(\Gamma) = S$ and let $ef \in E(\Gamma)$ if $V(C_e) \cap V(C_f) \neq \0$ (see Fig.~\ref{fig:chain_intersect_Gamma} for an example of such an edge).

    \item\label{step:ind_set} Run Algorithm~\ref{alg:rand_ind_set} on $\Gamma$ to compute the set $W \subseteq S$.
\end{enumerate}
We note that step~\ref{step:find_chains} can be implemented in $O(\ell t)$ rounds.
For $H_e = G[C_e]$, by Theorem~\ref{theo:MSSA} we have $|E(H_e)| = \length(C_e) = O(\ell t)$.
A single round in $\Gamma$ can be implemented in $O(\ell t)$ rounds in $G$ and so step~\ref{step:ind_set} can also be implemented in $O(\ell t)$ rounds.
It follows that the \LOCAL algorithm runs in $O(\Delta^{16}\log n)$ rounds.

It remains to bound $\EE[|W|]$.
To this end, we prove the following lemma.

\begin{lemma}\label{lemma:exp_num_edges}
    $\EE[|E(\Gamma)|] \leq 144\ell\,\Delta^{4}\,|U|$.
\end{lemma}

\begin{figure}[!b]
    \centering
    \begin{tikzpicture}[xscale = 0.65,yscale=0.65, rotate=10]
            \node[circle,fill=black,draw,inner sep=0pt,minimum size=4pt] (a) at (0,0) {};
        	\path (a) ++(-45:1) node[circle,fill=black,draw,inner sep=0pt,minimum size=4pt] (b) {};
        	\path (a) ++(75:1) node[circle,fill=black,draw,inner sep=0pt,minimum size=4pt] (c) {};

                \path (a) ++(0:5.5) node[circle,fill=black,draw,inner sep=0pt,minimum size=4pt] (d) {};
            
        	\path (d) ++(-45:1) node[circle,fill=black,draw,inner sep=0pt,minimum size=4pt] (e) {};
        	
        	\path (d) ++(0:6) node[circle,fill=black,draw,inner sep=0pt,minimum size=4pt] (g) {};

                \path (g) ++(-45:1) node[circle,fill=black,draw,inner sep=0pt,minimum size=4pt] (h) {};
                \path (g) ++(75:1) node[circle,fill=black,draw,inner sep=0pt,minimum size=4pt] (i) {};
                \path (g) ++(0:5.5) node[circle,fill=black,draw,inner sep=0pt,minimum size=4pt] (j) {};

                \path (j) ++(-45:1) node[circle,fill=black,draw,inner sep=0pt,minimum size=4pt] (k) {};
                \path (j) ++(75:1) node[circle,fill=black,draw,inner sep=0pt,minimum size=4pt] (l) {};
                \path (j) ++(0:5.5) node[circle,fill=black,draw,inner sep=0pt,minimum size=4pt] (m) {};
        	
        	\draw[thick,dotted] (a) to node[font=\fontsize{8}{8},midway,inner sep=1pt,outer sep=1pt,minimum size=4pt,fill=white] {$e$} (b);
        	
        	\draw[thick, decorate,decoration=zigzag] (c) to[out=20,in=150] (d) (e) to[out=85,in=130, looseness = 1.1]  (g) (i) to[out=20,in=150] (j) (l) to[out=20,in=160] (m);
        	
        	\draw[thick] (a) -- (c) (d) -- (e) (h) -- (g) -- (i) (k) -- (j) -- (l);
            
            \node[circle,fill=black,draw,inner sep=0pt,minimum size=4pt] (x) at (14,1.15) {};
            \node[anchor=north east] at (x) {$x$};

        \begin{scope}[xshift=6cm,yshift=6.5cm,rotate=-40]
            \node[circle,fill=black,draw,inner sep=0pt,minimum size=4pt] (a) at (0,0) {};
        	\path (a) ++(-45:1) node[circle,fill=black,draw,inner sep=0pt,minimum size=4pt] (b) {};
        	\path (a) ++(75:1) node[circle,fill=black,draw,inner sep=0pt,minimum size=4pt] (c) {};

                \path (a) ++(0:5.5) node[circle,fill=black,draw,inner sep=0pt,minimum size=4pt] (d) {};
            
        	\path (d) ++(-45:1) node[circle,fill=black,draw,inner sep=0pt,minimum size=4pt] (e) {};
        	
        	\path (d) ++(0:6) node[circle,fill=black,draw,inner sep=0pt,minimum size=4pt] (g) {};

                \path (g) ++(-45:1) node[circle,fill=black,draw,inner sep=0pt,minimum size=4pt] (h) {};
                \path (g) ++(75:1) node[circle,fill=black,draw,inner sep=0pt,minimum size=4pt] (i) {};
                \path (g) ++(0:6) node[circle,fill=black,draw,inner sep=0pt,minimum size=4pt] (j) {};

        	\draw[thick,dotted] (a) to node[font=\fontsize{8}{8},midway,inner sep=1pt,outer sep=1pt,minimum size=4pt,fill=white] {$f$} (b);
        	
        	\draw[thick, decorate,decoration=zigzag] (c) to[out=20,in=150] (d) (e) to[out=85,in=130, looseness = 1.1]  (g) (i) to[out=20,in=150] (j);
        	
        	\draw[thick] (a) -- (c) (d) -- (e)  (h) -- (g) -- (i);
        \end{scope}
        	
        \end{tikzpicture}
    \caption{Intersecting chains $C_e$ and $C_f$}
    \label{fig:chain_intersect_Gamma}
\end{figure}
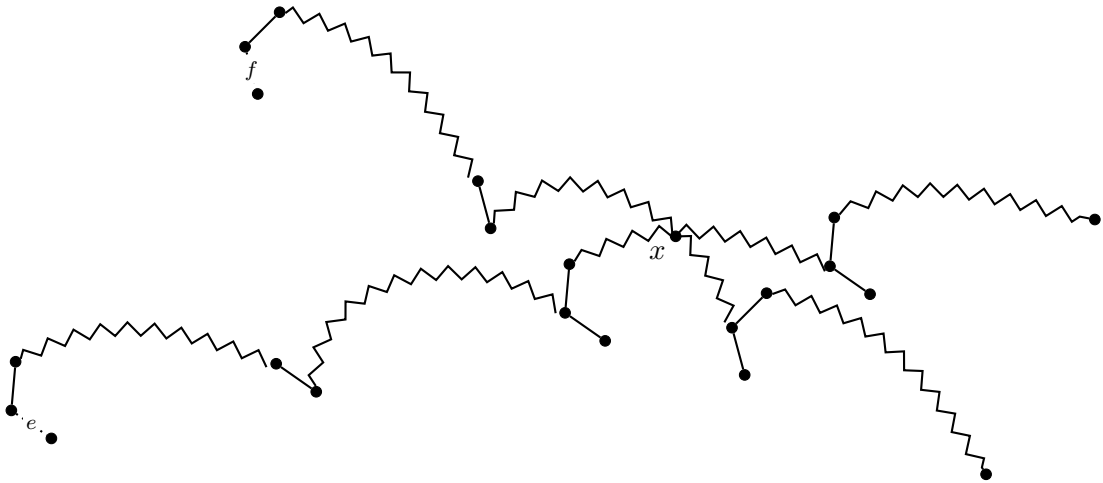

\begin{proof}
    Recall the definition of input sequences, records and termini introduced in \S\ref{subsection:analysis_MSSA}.
    For an input sequence $I \in \mathcal{I}^{(t)}$, we let $|I| = t$ and $C(I) = F_0 + P_0 + \cdots + F_{k(I)-1} + P_{k(I)-1}$ be the resulting multi-step Shannon chain after the first $t$ iterations of Algorithm~\ref{alg:multi_shannon_chain}.
    Consider a vertex $x \in V(G)$ and an integer $1 \leq j \leq t$.
    We define the subset of input sequences $\mathcal{I}_j(x)$ to contain $I = (f,z, \ell_1, \ldots, \ell_{t_0})$ if:
    \begin{itemize}
        \item $t_0 \leq t$, and if $t_0 < t$, then Algorithm~\ref{alg:multi_shannon_chain} with this input sequence terminates after the $t_0$-th iteration of the \textsf{while} loop, and
        
        \item $j \leq k(I)$ and $x \in V(F_{j-1} + P_{j-1})$, i.e., $x$ lies on the $j$-th Shannon chain in $C(I)$.
    \end{itemize}
    Let $c_j(x)$ be the random variable equal to the number of multi-step Shannon chains $C_e$ for $e \in U$ such that $x$ appears on the $j$-th chain in $C_e$.
    For example in Fig.~\ref{fig:chain_intersect_Gamma}, $C_e$ would be included in $c_3(x)$ and $C_f$ would be included in $c_2(x)$.
    Since input sequences include information about the uncolored edge input to Algorithm~\ref{alg:multi_shannon_chain}, we have 
    \[
        \E[c_j(x)] \leq \sum_{I \in \mathcal{I}_j(x)}\frac{1}{\ell^{|I|}} = \frac{1}{\ell^t}\sum_{I \in \mathcal{I}_j(x)}\ell^{t - |I|}.
    \]
    Let us define $\tilde{\mathcal{I}}_j(x)$ to contain all input sequences $I \in \mathcal{I}^{(t)}$ such that some initial segment $I'$ of $I$ is in $\mathcal{I}_j(x)$.
    Note that since no input sequence in $\mathcal{I}_j(x)$ is an initial segment of another sequence in $\mathcal{I}_j(x)$, $I'$ is unique.
    In particular, for each $I \in \mathcal{I}_j(x)$, there are $\ell^{t - |I|}$ input sequences in $\tilde{\mathcal{I}}_j(x)$.
    It now follows that
    \begin{align}\label{eq:exp_cj}
        \E[c_j(x)] \leq |\tilde{\mathcal{I}}_j(x)|/\ell^t.
    \end{align}
    Let us now bound $|\tilde{\mathcal{I}}_j(x)|$.
    Consider any $I = (f, z, \ell_1, \ldots, \ell_{t_0}) \in \mathcal{I}_j(x)$ and let $D(I) = (d_1, \ldots, d_{t_0})$ and $C(I) = F_0 + P_0 + \cdots + F_{k(I)-1} + P_{k(I)-1}$.
    We note that there must be an index $j-1 \leq t' \leq t_0$ such that, for $I' = (f, z, \ell_1, \ldots, \ell_{t'})$, we have
    \begin{itemize}
        \item $I' \in \mathcal{I}^{(t')}(D(I'), \Start(F_{j-1}), \vstart(F_{j-1}))$,
        
        \item $D(I') \in \mathcal{D}_{j-1}^{(t')}$, and

        \item $C(I') = F_0 + P_0 + \cdots + F_{j-2} + P_{j-2}$.
    \end{itemize}
    Let us now consider the number of possible tuples $(\Start(F_{j-1}), \vstart(F_{j-1}))$.
    We will split into two cases depending on whether $x$ belongs to $V(F_{j-1})$ or $V(P_{j-1})$.
    \begin{enumerate}[label=\ep{\textbf{Case \arabic*}}, labelindent=0pt,leftmargin=*]
        \item $x \in V(F_{j-1})$. 
        Here, $\Pivot(F_{j-1}) \in N_G[x]$ and $\Start(F_{j-1}) \in E_G(\Pivot(F_{j-1}))$.
        Finally, $\vstart(F_{j-1}) \in V(\Start(F_{j-1}))$ distinct from $\Pivot(F_{j-1})$. 
        So there are at most $(\Delta + 1)\Delta \leq 2\Delta^2$ choices for $(\Start(F_{j-1}), \vstart(F_{j-1}))$.

        \item $x \in V(P_{j-1}) \setminus V(F_{j-1})$.
        By Lemma \ref{lemma:non-intersecting_degrees} \ref{item:related_phi}, all edges of $P_{j-1}$ are colored with two colors, say $\alpha$ and $\beta$, and there are at most $\Sha^2$ choices for them.
        By Lemma \ref{lemma:non-intersecting_degrees} \ref{item:degree_end} and \ref{item:related_phi}, given $\alpha$ and $\beta$, we can locate $\vend(F_{j-1})$ as one of the two endpoints of the maximal $\alpha\beta$-path under the coloring $\phi$ containing $x$.
        Since $\Pivot(F_{j-1}) \in N_G(\vend(F_{j-1}))$, $\Start(F_{j-1}) \in E_G(\Pivot(F_{j-1}))$ and $\vstart(F_{j-1}) \in V(\Start(F_{j-1}))$ distinct from $\Pivot(F_{j-1})$, there are at most $\Delta^2$ choices for $(\Start(F_{j-1}),$ $\vstart(F_{j-1}))$ given $\vend(F_{j-1})$.
        In conclusion, there are at most $2\Sha^2\Delta^2 \leq 9\Delta^4/2$ choices for $(\Start(F_{j-1}), \vstart(F_{j-1}))$ in this case.
    \end{enumerate}
    Therefore, the total number of options for $(\Start(F_{j-1}), \vstart(F_{j-1}))$ is at most
    \[
        2\Delta^2 + 9\Delta^4/2 \,\leq\, 5\Delta^4,
    \]
    for $\Delta \geq 2$.
    Note that, given $I'$, we have at most $\ell^{t-t'}$ choices for $(\ell_{t'+1}, \ldots, \ell_{t})$.
    With this and Lemmas \ref{lemma:records_bound_by_wt}, \ref{lemma:wD_bound}, and \ref{lemma:Ds_bound}, we obtain the following chain of inequalities:
    \begin{align}
        |\tilde{\mathcal{I}}_j(x)| \,&\leq\, 5 \Delta^4 \sum_{t' = j-1}^t\sum_{D \in \mathcal{D}_{j-1}^{(t')}}\wt(D)\,\ell^{t-t'} \nonumber\\
        &\leq\, 5\Delta^4\ell^t\sum_{t' = j-1}^t|\mathcal{D}_{j-1}^{(t')}|\left(\frac{81\Delta^{15}}{\ell}\right)^{t'/2}(36\Delta^7\ell)^{-(j-1)/2} \nonumber\\
        &\leq\, 5\Delta^4\ell^t(36\Delta^7\ell)^{-(j-1)/2}\sum_{t' = j-1}^t\left(\frac{1296\Delta^{15}}{\ell}\right)^{t'/2} \nonumber\\
        [\text{assuming $\ell \geq 5200\Delta^{15}$}] \qquad &\leq\, 10\Delta^4\ell^t\left(\frac{6\Delta^{4}}{\ell}\right)^{j-1} \nonumber\\
        &\leq\, 2\ell^{t+1} \left(\frac{6\Delta^{4}}{\ell}\right)^j. \label{eq:Ij}
    \end{align}
    From \eqref{eq:exp_cj}, we conclude that
    \[\E[c_j(x)] \leq 2\ell \left(\frac{6\Delta^{4}}{\ell}\right)^j.\]
    Let us now compute $\E[|E(\Gamma)|]$.
    For each uncolored edge $e \in U$, let
    \[
        C_e \,=\, F_0^e + P_0^e + \cdots + F_{k_e-1}^e + P_{k_e-1}^e \qquad \text{and} \qquad V^e_j \,\defeq\, V\left(F^e_{j-1} + P^e_{j-1}\right).
    \]
    Let $N^+(e)$ be the set of all edges $f \in U \setminus \set{e}$ such that for some $j \leq j'$, we have
    \[
        V^e_j \cap V^f_{j'} \,\neq\, \0.
    \]
    By definition, if $ef \in E(\Gamma)$, then $V(C_e) \cap V(C_f) \neq \0$, and hence $f \in N^+(e)$ or $e \in N^+(f)$ (or both). Therefore, we have $|E(\Gamma)| \leq \sum_{e \in U} |N^+(e)|$. Note that
    \[
        |N^+(e)| \,\leq\, \sum_{j = 1}^{k_e} \sum_{x \in V^e_j} \sum_{j' = j}^t \sum_{f \in U \setminus \set{e}} \bbone{x \in V^f_{j'}}.
    \]
    Since Algorithm~\ref{alg:multi_shannon_chain} is executed independently for every uncolored edge, we have
    \begin{align*}
        \E[|N^+(e)|\mid C_e] \,&\leq\, \sum_{j = 1}^{k_e} \sum_{x \in V^e_j} \sum_{j' = j}^t \sum_{f \in U \setminus \set{e}} \Pr\left[x \in V^f_{j'} \mid C_e\right] \\
        &=\, \sum_{j = 1}^{k_e} \sum_{x \in V^e_j} \sum_{j' = j}^t \sum_{f \in U \setminus \set{e}} \Pr\left[x \in V^f_{j'} \right] \\
        &\leq\, \sum_{j = 1}^{k_e} \sum_{x \in V^e_j} \sum_{j' = j}^t \E[c_{j'}(x)] \\
        &\leq\, \sum_{j = 1}^{k_e} \sum_{x \in V^e_j} \sum_{j' = j}^t 2\ell\left(\frac{6\Delta^4}{\ell}\right)^{j'}.
    \end{align*}
    Using that $k_e \leq t$ and $|V^e_j| \leq \Delta + 1 + 2\ell \leq 3\ell$ for all $e \in U$ and $j \leq k_e$ and assuming that $\ell \geq 12\Delta^4$, we can bound the last expression as follows:
    \begin{align*}
        \sum_{j = 1}^{k_e} \sum_{x \in V^e_j} \sum_{j' = j}^t 2\ell\left(\frac{6\Delta^4}{\ell}\right)^{j'} \,&\leq\, 6 \ell^2 \sum_{j = 1}^{t} \sum_{j' = j}^t \left(\frac{6\Delta^4}{\ell}\right)^{j'} \\
        &\leq\, 12\ell^2 \sum_{j = 1}^{t} \left(\frac{6\Delta^4}{\ell}\right)^{j} \\
        &\leq\, 144 \ell\, \Delta^4.
    \end{align*}
    As this bound does not depend on $C_e$, we conclude that $\E[|N^+(e)|] \leq 144 \ell\, \Delta^4$ unconditionally, and
    \[
        \E[|E(\Gamma)|] \,\leq\, \sum_{e \in U} \E[|N^+(e)|] \,\leq\, 144 \ell\, \Delta^4 |U|,
    \]
    as desired.
\end{proof}

Note the following as a result of Theorem~\ref{theo:MSSA}:
\[\Pr[|S| \geq |U|/2] = 1 - \Pr[|S| < |U|/2] \geq 1 - \Pr[S \neq U] \geq 1 - \frac{|U|}{\poly(n)} \geq \frac{1}{2},\]
where the inequalities follow by a union bound and appropriate choices for $\ell$ and $t$.
With this in hand, by Lemma~\ref{lemma:rand_ind_set} and Jensen's inequality, we have
\begin{align*}
    \EE[|W|] \geq \EE\left[\frac{|S|}{d(\Gamma) + 1}\right] &\geq \EE\left[\frac{|S|}{d(\Gamma) + 1} \,\bigg|\, |S| \geq |U|/2\right]\,\Pr[|S| \geq |U|/2] \\
    &\geq \frac{|U|}{4}\,\EE\left[\frac{1}{d(\Gamma) + 1} \,\bigg|\, |S| \geq |U|/2\right] \\
    &\geq \frac{|U|}{4}\,\frac{1}{\EE[d(\Gamma)\mid |S| \geq |U|/2] + 1}.
\end{align*}
Note the following as a result of Lemma~\ref{lemma:exp_num_edges}:
\begin{align*}
    \EE[d(\Gamma)\mid |S| \geq |U|/2] &= \EE\left[\frac{2|E(\Gamma)|}{|S|} \,\bigg|\, |S| \geq |U|/2\right] \\
    &\leq \frac{4}{|U|}\EE[|E(\Gamma)| \mid |S| \geq |U|/2] \\
    &\leq \frac{4}{|U|}\,\frac{\EE[|E(\Gamma)|]}{\Pr[|S| \geq |U|/2]} \\
    &\leq 1152\ell\Delta^4.
\end{align*}
Plugging this above, we get
\[\EE[|W|] \geq \frac{|U|}{4}\,\frac{1}{1152\ell\Delta^4 + 1} = \Omega\left(\frac{|U|}{\Delta^{20}}\right),\]
as desired.

%% file: vizing.tex
\section{Proof of Theorem~\ref{theo:vizing}}\label{section:vizing}

The proof of Theorem~\ref{theo:vizing} is nearly identical to that of the analogous results in \cite{fastEdgeColoring}.
The only modifications made are to the algorithms described in \cite[\S 3.5]{fastEdgeColoring}.

Let us first discuss the \hyperref[alg:first_vizing_fan]{First Vizing Fan Algorithm}, which is described formally in Algorithm~\ref{alg:first_vizing_fan} and is an extension of \cite[Algorithm~3.1]{fastEdgeColoring} to multigraphs.
The algorithm takes as input a proper partial coloring $\phi$, an uncolored edge $e \in E(G)$, and a choice of a pivot vertex $x \in V(e)$.
The output is a tuple $(F, \beta, j)$ such that:
\begin{itemize}
    \item $F$ is a fan with $\Start(F) = e$ and $\Pivot(F) = x$,
    \item $\beta \in [\Delta + \mu]$ is a color and $j$ is an index such that $\beta \in M(\phi, \vend(F))\cap M(\phi, \vend(F|j))$.
\end{itemize}
We first assign $\beta(z) = M(\phi, z)$ to each $z \in N_G(x)$. 
To construct $F$, we follow a series of iterations. 
At the start of each iteration, we have a fan $F = (e_0, \ldots, e_k)$ where $e_0 = e$ and $x \in V(e_i)$ for all $i$. 
Let $y_i \in V(e_i)$ be distinct from $x$ (the $y_i$'s are not necessarily distinct).
If $\min \beta(y_k) \in M(\phi, x)$, then $F$ is $\phi$-happy and we return $(F, \min \beta(y_k), k+1)$.
If not, let $f \in E_G(x)$ be the unique edge adjacent to $x$ such that $\phi(f) = \min \beta(y_k)$. 
We now have two cases.
\begin{enumerate}[label=\ep{\emph{Case \arabic*}},labelindent=5pt,leftmargin=*]
    \item $f \notin \set{e_0, \ldots, e_k}$. Then we update $F$ to $(e_0, \ldots, e_k, f)$, remove $\min \beta(y_k)$ from $\beta(y_k)$ and continue.
    \item $f = e_j$ for some $0 \leq j \leq k$. Note that $\phi(e_0) = \blank$ and $\min \beta(y_k) \in M(\phi, y_k)$, so we must have $1 \leq j < k$. 
    In this case, we return $(F, \min \beta(y_k), j)$.
\end{enumerate}
We remark that $|M(\phi, z)| \geq \mu$ and so $\beta(z)$ is never empty when defining $\eta$ at step~\ref{step:beta_empty}.

\begin{algorithm}[h]\small
\caption{First Vizing Fan}\label{alg:first_vizing_fan}
\begin{flushleft}
\textbf{Input}: A proper partial $(\Delta+\mu)$-edge-coloring $\phi$, an uncolored edge $e\in E(G)$, and a vertex $x \in V(e)$. \\
\textbf{Output}: A fan $F$ with $\Start(F) = e$ and $\Pivot(F) = x$, a color $\beta \in [\Delta + \mu]$, and an index $j$ such that $\beta \in M(\phi, \vend(F))\cap M(\phi, \vend(F|j))$.
\end{flushleft}


\begin{algorithmic}[1]
    \State Let $y \in V(e)$ be distinct from $x$.
    \State $\mathsf{nbr}(\eta) \gets \blank$ \textbf{for each} $\eta \in [\Delta+\mu]$, \quad $\beta(z) \gets M(\phi, z)$ \textbf{for each} $z \in N_G(x)$
    \For{$f\in E_G(x)$}
        \State $\mathsf{nbr}(\phi(f)) \gets f$
    \EndFor
    \medskip
    \State $\mathsf{index}(f) \gets \blank$ \textbf{for each} $f \in E_G(x)$
    \State $F \gets (e)$, \quad $k \gets 0$, \quad $y_k \gets y$, \quad $\mathsf{index}(e) \gets k$
    \While{$k < \deg_G(x)$}
        \State $\eta \gets \min\beta(y_k)$, \quad $\mathsf{remove}(\beta(y_k), \eta)$ \label{step:beta_empty}
        \If{$\eta \in M(\phi, x)$}
            \State \Return $(F, \eta, k + 1)$ \label{step:happy_fan_vizing}
        \EndIf
        \State $k \gets k + 1$
        \State $e_k \gets \mathsf{nbr}(\eta)$
        \If{$\mathsf{index}(e_k) \neq \blank$}
            \State \Return $(F, \eta, \mathsf{index}(e_k))$
        \EndIf
        \State $\mathsf{index}(e_k) \gets k$, \quad $y_k \in V(e_k)$ distinct from $x$
        \State $\mathsf{append}(F, e_k)$
    \EndWhile
\end{algorithmic}
\end{algorithm}

We will prove the following lemma analogous to \cite[Lemma 3.9]{fastEdgeColoring}.

\begin{lemma}\label{lemma:first_fan_vizing}
    Let $(F, \beta, j)$ be the output of Algorithm~\ref{alg:first_vizing_fan} on input $(\phi, e, x)$. Then either $\beta \in M(\phi, x)$ and $F$ is $\phi$-happy, or else, for $F'\defeq F|j$ and any $\alpha \in M(\phi, x)$, we have that either $F$ or $F'$ is $(\phi, \alpha\beta)$-successful.
\end{lemma}

\begin{proof}
    Let $F = (e_0, \ldots, e_{k-1})$ and let $y_j \in V(e_j)$ be distinct from $x$.
    We note the edges $e_i$ are distinct by construction, while the vertices $y_i$ may not be.
    Let us first show that $F$ is $\phi$-shiftable.
    Note that as the edges $e_i$ are distinct, $x \in V(e_i)$ for each $i$ and $\phi$ is proper, the colors $\phi(e_i)$ are distinct as well.
    By our choice of $\eta$ in the \textsf{while} loop of Algorithm~\ref{alg:first_vizing_fan}, we have $\phi(e_i) \in M(\phi, y_{i-1})$ for each $1 \leq i < k$ and so the coloring $\Shift(\phi, F)$ is proper.

    If $j = k$, we have reached step~\ref{step:happy_fan_vizing}. 
    In this case, $\beta \in M(\phi, x)$ and $F$ is $\phi$-happy.
    If not, we have $j < k$, $\beta = \phi(e_j)$ and $\beta \in M(\phi, y_{k-1})$.
    By construction, we also have $\beta \in M(\phi, y_{j-1})$.
    Therefore, $F$ and $F'$ are both $(\phi, \alpha\beta)$-hopeful.
    Let us show at least one of them is $(\phi, \alpha\beta)$-successful.
    To this end, we let $\psi \defeq \Shift(\phi, F)$ and $\psi' \defeq \Shift(\phi, F')$.


    First, we claim that $e_{k-1}$ (resp. $e_{j-1}$) is $(\psi, \alpha\beta)$-hopeful (resp. $(\psi', \alpha\beta)$-hopeful).
    We will consider $e_{k-1}$ as the argument is similar for $e_{j-1}$.
    The goal is to apply Fact~\ref{fact:fan}. To this end, let $I \defeq \set{0 \leq i < k-1\,:\, e_i \in E_G(x, y_{k-1})}$.
    It is enough to show that $\phi(e_{i+1}) \neq \beta$ for any $i \in I$.
    We note that at the $j$-th iteration, we add $e_j$ such that $\phi(e_j) = \beta \in M(\phi, y_{j-1})$.
    As the edges are distinct, the only case we need to consider is $y_{j-1} = y_{k-1}$.
    We claim that this is not possible.
    To see this, we note that at step~\ref{step:beta_empty} in the $j$-th iteration, we remove $\beta$ from $\beta(y_{j-1})$.
    In particular, if $y_{j-1} = y_{k-1}$, then $\eta \neq \beta$ on the $k$-the iteration of the \textsf{while} loop, contradicting the output being $(F, \beta, j)$.

    Now suppose that $F$ is $(\phi, \alpha\beta)$-disappointed, i.e., $x$ and $y_{k-1}$ are $(\psi, \alpha\beta)$-related. 
    As $e_{k-1}$ is $(\psi, \alpha\beta)$-hopeful, it follows that $x$ and $y_{k-1}$ are the endpoints of a path component of $G(\psi, \alpha\beta)$.
    Let us define
    \[P \defeq P(e_{k-1}; \psi, \beta\alpha) = (f_0 = e_{k-1}, f_1, \ldots, f_l).\]
    As $\alpha \in M(\psi, x)$ and $\psi(e_{j-1}) = \beta$, we have $f_1 = e_{j-1}$.
    For $2 \leq i \leq l$, $x \notin V(f_i)$ and so shifting a fan with pivot $x$ does not change the colors on any of these edges.
    In particular, they have the same colors under $\phi,\, \psi$ and $\psi'$ which implies that $y_{j-1}$ and $y_{k-1}$ are $(\psi', \alpha\beta)$-related.
    As $y_{j-1} \neq y_{k-1}$, we have $\beta \in M(\psi', y_{k-1})$. 
    We can conclude that $y_{j-1}$ and $y_{k-1}$ are the endpoints of a  path component of $G(\psi', \alpha\beta)$.
    Since $\alpha \in M(\psi', x)$, it follows that $\deg(x; \psi', \alpha\beta) < 2$. As $x \neq y_{k-1}$, we have $x$ and $y_{j-1}$ are not $(\psi', \alpha\beta)$-related.
    Therefore, $F'$ is $(\phi, \alpha\beta)$-successful, as desired.
\end{proof}

By storing $\beta(z)$ as a hash map with key set $V(G)$, we can compute $\min \beta (z)$ in $O(\Delta + \mu) = O(\Delta)$ time.
Similarly, removing an element from $\beta(z)$ can be done in $O(\Delta)$ time.
By storing $\mathsf{nbr}(\cdot)$ as a hash map as well, the first loop runs in $O(\Delta)$ time.
Each operation in the second loop takes $O(\Delta)$ time and so Algorithm~\ref{alg:first_vizing_fan} runs in $O(\Delta^2)$ time.


The next algorithm we describe is the \hyperref[alg:next_vizing_fan]{Next Vizing Fan Algorithm}, which is formally described in Algorithm~\ref{alg:next_vizing_fan} and is an extension of \cite[Algorithm~3.2]{fastEdgeColoring} to multigraphs.
The algorithm takes as input a proper partial coloring $\phi$, an uncolored edge $e \in E(G)$, a pivot vertex $x \in V(e)$ and a color $\beta$ such that for $y \in V(e)$ distinct from $x$:
\begin{itemize}
    \item $M(\phi, x) \setminus M(\phi, y) \neq \0$, and
    \item $\beta \in M(\phi, y)$.
\end{itemize}
The output is a tuple $(F, \delta, j)$ such that:
\begin{itemize}
    \item $F$ is a fan with $\Start(F) = e$ and $\Pivot(F) = x$,
    \item $\delta \in [\Delta + \mu]$ is a color and $j$ is an index such that $\delta \in M(\phi, \vend(F))\cap M(\phi, \vend(F|j))$.
\end{itemize}
The algorithm is similar to Algorithm~\ref{alg:first_vizing_fan} with a few minor differences. 
First, we use $\delta(\cdot)$ in place of $\beta(\cdot)$ as a notational change.
Second, we have a restriction on $\delta(y)$.
Finally, in the loop we additionally check whether $\min \delta(z) = \beta$, in which case the algorithm returns $(F, \beta, k+1)$.
We remark that as $y$ is adjacent to an uncolored edge, we must have $|M(\phi, y)| \geq \mu + 1$.
In particular, $|\delta(z)| \geq \mu$ for every $z \in N_G(x)$ at the start of the while loop and so $\delta(z)$ is never empty when defining $\eta$ at step~\ref{step:delta_empty}.

\begin{algorithm}[h]\small
\caption{Next Vizing Fan}\label{alg:next_vizing_fan}
\begin{flushleft}
\textbf{Input}: A proper partial $(\Delta+\mu)$-edge-coloring $\phi$, an uncolored edge $e \in E(G)$, a vertex $x \in V(e)$, a color $\beta \in M(\phi, y)$ where $y \in V(e)$ distinct from $x$. \\
\textbf{Output}: A fan $F$ with $\Start(F) = e$ and $\Pivot(F) = x$, a color $\delta \in [\Delta + \mu]$, and an index $j$ such that $\delta \in M(\phi, \vend(F)) \cap M(\phi, \vend(F|j))$.
\end{flushleft}
\begin{algorithmic}[1]
    \State Let $y \in V(e)$ be distinct from $x$.
    \State $\mathsf{nbr}(\eta) \gets \blank$ \textbf{for each} $\eta \in [\Delta+\mu]$, \quad $\delta(z) \gets M(\phi, z)$ \textbf{for each} $z \in N_G(x)$
    \State $\mathsf{remove}(\delta(y), \beta)$
    \For{$f\in E_G(x)$}
        \State $\mathsf{nbr}(\phi(f)) \gets f$
    \EndFor
    \medskip
    \State $\mathsf{index}(f) \gets \blank$ \textbf{for each} $f \in E_G(x)$
    \State $F \gets (e)$, \quad $k \gets 0$, \quad $y_k \gets y$, \quad $\mathsf{index}(e) \gets k$
    \While{$k < \deg_G(x)$}
        \State\label{step:delta_empty} $\eta \gets \min \delta(y_k)$, \quad $\mathsf{remove}(\delta(y_k), \eta)$
        \If{$\eta \in M(\phi, x)$}
            \State \Return $(F, \eta, k+1)$ \label{step:next_vizing_happy}
        \ElsIf{$\eta = \beta$}
            \State \Return $(F, \eta, k+1)$ \label{step:vizing_next_beta}
        \EndIf
        \State $k \gets k+1$
        \State $e_k \gets \mathsf{nbr}(\eta)$
        \If{$\mathsf{index}(e_k) \neq \blank$}
            \State \Return $(F, \eta, \mathsf{index}(e_k))$
        \EndIf
        \State $\mathsf{index}(e_k) \gets k$, \quad $y_k \in V(e_k)$ distinct from $x$
        \State $\mathsf{append}(F, e_k)$
    \EndWhile
\end{algorithmic}
\end{algorithm}

We will prove the following lemma analogous to \cite[Lemma 3.10]{fastEdgeColoring}.

\begin{lemma}\label{lemma:next_fan_vizing}
    Let $(F, \delta, j)$ be the output of Algorithm~\ref{alg:next_vizing_fan} on input $(\phi, e, x, \beta)$, let $F' \defeq F|j$ and let $\alpha \in M(\phi, x) \setminus M(\phi, y)$ be arbitrary.
    Then no edge in $F$ is colored $\alpha$ or $\beta$ and at least one of the following statements holds:
    \begin{itemize}
        \item $F$ is $\phi$-happy, or
        \item $\delta = \beta$ and $F$ is $(\phi, \alpha\beta)$-hopeful, or
        \item either $F$ or $F'$ is $(\phi, \gamma\delta)$-successful for some $\gamma \in M(\phi, x) \setminus \set{\alpha}$.
    \end{itemize}
\end{lemma}

\begin{proof}
    An identical argument as in the proof of Lemma~\ref{lemma:first_fan_vizing} shows that $F$ is $\phi$-shiftable.
    We also note that $x$ is adjacent to an uncolored edge and so $|M(\phi, x)| \geq \mu + 1 \geq 2$.
    In particular, $\gamma$ is well defined.

    If $j = k$, we have either reached step~\ref{step:next_vizing_happy} or~\ref{step:vizing_next_beta}. 
    In the former case, $\delta \in M(\phi, x)$ and $F$ is $\phi$-happy.
    In the latter, $\delta = \beta$.
    By construction, $F$ is $(\phi, \alpha\beta)$-hopeful.
    Furthermore, as $\alpha \in M(\phi, x)$ and we stop iterating when $\eta = \beta$, no edge in $F$ is colored $\alpha$ or $\beta$.
    The proof for the case that $\delta \neq \beta$ follows identically to that of Lemma~\ref{lemma:first_fan_vizing}.
\end{proof}

As for Algorithm~\ref{alg:first_vizing_fan}, one can bound the runtime of Algorithm~\ref{alg:next_vizing_fan} by $O(\Delta^2)$.
The rest of the proof of Theorem~\ref{theo:vizing} is identical to that in \cite{fastEdgeColoring}, \textit{mutatis mutandis}.